%% file: main.tex
\RequirePackage{fix-cm}
\documentclass[smallextended]{svjour3}       
\smartqed  
\usepackage{graphicx}
\usepackage{tabularx}
\usepackage{tikz} 
\usetikzlibrary{calc}
\usetikzlibrary{arrows}
\usepackage{array, makecell}
\usepackage{multirow}
\usepackage{enumitem}
\usepackage{color}
\usepackage{amsmath,amssymb} 
\usepackage{wrapfig}

\usetikzlibrary{arrows, shapes, snakes}

\begin{document}

\input{pictures}

\input{commands}

\title{A More General Theory of Static Approximations for Conjunctive Queries
\thanks{Barcel\'o is funded by Millennium Institute for Foundational Research on Data and Fondecyt Grant 1170109. Zeume acknowledges the financial support by the European Research Council
(ERC), grant agreement No 683080. Romero and Zeume thank the Simons Institute for the
Theory of Computing for hosting them. This project has received funding from the European
Research Council (ERC) under the European Union's Horizon 2020 research and
innovation programme (grant agreement No 714532). The paper reflects only the
authors' views and not the views of the ERC or the European Commission. The
European Union is not liable for any use that may be made of the information
contained therein.}
}


\author{Pablo Barcel\'o \and Miguel Romero \and Thomas Zeume}


\institute{P. Barcel\'o \at
               DCC, University of Chile \& IMFD Chile \\
              Tel.: +56229784813\\
              \email{pbarcelo@dcc.uchile.cl}           
           \and
           M. Romero \at
              Department of Computer Science, University of Oxford \\  
              \email{miguel.romero@cs.ox.ac.uk}  
           \and
            T. Zeume \at
            TU Dortmund \\
            \email{thomas.zeume@cs.tu-dortmund.de}
}

\date{Received: date / Accepted: date}

\maketitle

\begin{abstract}
Conjunctive query (CQ) evaluation is NP-complete, but becomes tractable for fragments of bounded hypertreewidth. Approximating a hard CQ by a query from such a fragment can thus allow for an efficient approximate evaluation. While underapproximations (i.e., approximations that return correct answers only) are well-understood, the dual notion of overapproximations (i.e, approximations that return complete -- but not necessarily sound -- answers), and also a 
more general notion of approximation based on the symmetric difference of query results, are almost unexplored.  In fact, the decidability of the basic problems of evaluation, identification, and existence of those approximations has been open. 

This article establishes a connection between overapproximations and existential pebble games that allows for studying such problems systematically. Building on this connection, it is shown that the evaluation and identification problem for overapproximations can be solved in polynomial time. While the general existence problem remains open, the problem is shown to be decidable in 2EXPTIME over the class of acyclic CQs and in PTIME for Boolean CQs 
over binary schemata. Additionally we  propose a more liberal notion of overapproximations to remedy the known shortcoming that queries might not have an overapproximation, and study how queries can be overapproximated in the presence of tuple generating and equality generating dependencies.

The techniques are then extended to symmetric difference approximations and used to provide several complexity results for the identification, existence, and evaluation problem for this type of approximations.

\end{abstract}

\section{Introduction}
\label{intro}

Due to the growing number of scenarios in which exact query evaluation is infeasible -- e.g.,
when the volume of the data being queried is very large, or when
queries are inherently complex --   
approximate query answering has become an important area of study in databases  
(see, e.g. \cite{GG01,Ioa03,Liu09,FLMTWW10,FO11}). 
Here we focus on approximate query answering  for the fundamental class of conjunctive queries (CQs). 

%

Exact query evaluation for CQs, that is, determining whether a tuple $\bar a$ is contained in the result of a query $q$ on a database $\D$, is NP-complete. It is known that the complexity of evaluation of a CQ depends on its {\em degree of acyclicity}, which can be formalized using different notions. One of the most general and well-studied such notions corresponds to {\em generalized hypertreewidth} 
\cite{GLS02}.  Notably, the classes of CQs of bounded generalized hypertreewidth can be evaluated in polynomial time (see \cite{GGLS16} for a survey).

Following recent work on approximate query answering for CQs and some related query languages \cite{BLR14,BRV16},  
we study approximation of CQs by queries of bounded generalized hypertreewidth. If a CQ can be approximated by such a restricted query, this provides 
a certificate for an efficient approximation of the evaluation problem. It is worth noticing that 
the approximations studied here are  {\em static} in the sense that they depend only on the 
CQ $q$ and not on the underlying database $\D$. This has clear benefits in terms of the cost of the approximation 
process, as $q$ is often orders of magnitude smaller than $\D$ and an approximation that has been computed once can be used for all databases. Moreover, it allows us to construct a principled approach  to CQ approximation based on the well-studied notion of CQ containment \cite{CM77}. 
Recall that a CQ $q$ is {\em contained} in a CQ $q'$, written $q \subseteq q'$, if $q(\D) \subseteq q'(\D)$ over each database $\D$. 
This notion constitutes the theoretical basis for the study of several CQ optimization problems \cite{AHV95}.

Before stating our contributions, we recall the precise notions of approximation under consideration in this article, as well as the algorithmic problems of interest. 

As mentioned above, we are interested in approximating CQs by queries from the class $\HW(k)$ of CQs of generalized hypertreewidth with width at most $k$, for some $k \geq 1$. Intuitively, an approximation of a CQ $q$ is a 
query $q' \in \HW(k)$ whose result when evaluated on a database $\D$ is so close to the result of $q$, that no result of another query from $\HW(k)$ is closer. A formalization of this notion was first introduced in \cite{BLR12}, based on the following partial order $\sqsubseteq_q$ over the  CQs in $\HW(k)$: if $q',q'' \in \HW(k)$, then $q' \sqsubseteq_q q''$ iff over every database $\D$ the {\em symmetric difference}\footnote{Recall that the symmetric difference 
between sets $A$ and $B$ is $(A \setminus B) \cup (B \setminus A)$.} 
 between $q(\D)$ and $q''(\D)$ is contained in the symmetric difference  between $q(\D)$ and $q'(\D)$.Intuitively, $q' \sqsubseteq_q q''$ if  $q''$ is a better $\HW(k)$-approximation of $q$ than $q'$. The $\HW(k)$-approximations of $q$ then correspond to maximal elements with respect to
$\sqsubseteq_q$ among a distinguished class of CQs in $\HW(k)$. 

Three notions of approximation were introduced in \cite{BLR12}, 
by imposing different ``reasonable'' conditions on such a class. These are:

\begin{itemize}
\item {\em Underapproximations:}  
In this case we look for approximations in the set of CQs $q'$ in $\HW(k)$ that are contained in $q$, i.e., $q' \subseteq q$. 
  This ensures that the evaluation of such approximations always produces correct (but not necessarily complete) answers to $q$.  A 
$\HW(k)$-underapproximation of $q$ is then a CQ $q'$ amongst these CQs that is {\em maximal} 
with respect to the partial order defined by $\sqsubseteq_q$. 
Noticeably, the latter coincides with being maximal with respect to the containment 
partial order $\subseteq$ among the CQs in $\HW(k)$ that are contained in $q$; i.e., no other CQ in such a set 
strictly contains $q'$. 

\item {\em Overapproximations:}  
This is the dual notion of underapproximations, in which we look for minimal elements
 in the class of CQs $q'$ in $\HW(k)$ that contain $q$, i.e., $q \subseteq q'$. 
Hence, $\HW(k)$-overapproximations produce complete (but not necessarily correct) answers to $q$. 

\item {\em Symmetric difference approximations:} 
While underapproximations must be contained in the original query, and overapproximations must contain it, 
symmetric difference approximations impose no constraints on approximations with respect to the partial order $\subseteq$. 
Thus, the symmetric difference $\HW(k)$-approximations of $q$ (or $\HW(k)$-$\Delta$-approximations from now on)
are the maximal CQs in $\HW(k)$ with respect to $\sqsubseteq_q$.
\end{itemize} 

The approximations presented above provide ``qualitative'' guarantees for evaluation, 
as they are as close as possible to $q$ among all CQs in $\HW(k)$ of a certain kind. 
In particular, 
under- and overapproximations are dual notions which provide lower and upper bounds for the exact evaluation of a CQ, 
while 
$\Delta$-approximations can give us useful information that complements the one provided by under- and overapproximations. 
Then, in order to develop a robust theory of bounded generalized hypertreewidth 
static approximations for CQs, it is necessary to have a good understanding of all three notions. 

The notion of underapproximation is by now well-understood~\cite{BLR14}.
Indeed, it is known that for each $k \geq 1$ the 
$\HW(k)$-underapproximations  have 
good properties that justify their application: 
(a)
they always exist, 
and (b) 
evaluating all $\HW(k)$-underapproximations of a CQ $q$ over a database $\D$ 
is {\em fixed-parameter tractable} with the size of $q$ as parameter.
 This is an improvement over general CQ evaluation 
for which the latter
is believed not to hold \cite{PY99}. 

In turn, the notions of 
overapproximations and $\Delta$-approximations, while introduced in \cite{BLR12}, are much less understood. In fact, 
no general tools have been identified so far for studying the decidability of basic problems such as:   
\begin{itemize}
\item {\em Existence:} Does CQ $q$ have a $\HW(k)$-overapproximation (or $\HW(k)$-$\Delta$-approximation)? 
\item {\em Identification:} Is it the case that $q'$ is a $\HW(k)$-overapproximation (resp., 
$\HW(k)$-$\Delta$-approximation) of $q$? 
\item{\em Evaluation:} Given a CQ $q$, a database $\D$, and a  tuple $\bar a$ in $\D$, is it the case that 
$\bar a \in q'(\D)$, for some $\HW(k)$-overapproximation (resp., $\HW(k)$-$\Delta$-approximation) $q'$ of $q$?
\end{itemize} 
Partial results were obtained in \cite{BLR12}, but based 
on ad-hoc tools. It has also been observed that  some CQs have no $\HW(k)$-overapproximations (in contrast to underapproximations, that always exist), which was seen as a negative result. 


\paragraph{{\em {\bf Contributions.}}} 
We develop tools for the study of overapproximations and $\Delta$-approximations. 
While we mainly focus on the former, we provide a detailed account of how our techniques can be extended to deal with the latter.  
In the context of $\HW(k)$-overapproximations, 
we apply our tools to pinpoint 
the complexity of 
evaluation and identification, and make progress
in the problem of existence. 
We also study when 
overapproximations do not exist and suggest how this can be alleviated. Finally, we study 
overapproximations in the presence of integrity constraints.  
Our contributions are as follows: 

\begin{enumerate}
\item
\underline{Link to existential pebble games.} 
We establish a link between $\HW(k)$-overapproximations  and  
 {\em existential pebble games} \cite{KV95}. 
Such games have been used to show that several classes of 
CQs of bounded width 
can be evaluated efficiently \cite{DKV02,CD05}. 
Using the fact that the existence of winning conditions 
in the existential pebble game can be checked in polynomial time~\cite{CD05}, we show 
that the identification and evaluation problems for $\HW(k)$-overapproximations are tractable.



%

\item
\underline{A more liberal notion of overapproximation.} 
We observe that non-existence of overapproximations is due to the fact that in some 
cases overapproximations 
require expressing conjunctions of infinitely many atoms. By relaxing our notion,  
we get that each CQ $q$ has a (potentially infinite) $\HW(k)$-overapproximation $q'$. This $q'$ is unique 
(up to equivalence). Further, it can be evaluated efficiently -- in spite of being potentially infinite -- by checking 
a winning condition 
for the existential $k$-pebble game on $q$ and $\D$. 

\item
\underline{Existence of overapproximations.}
It is still useful to check if a CQ 
$q$ has a {\em finite} $\HW(k)$-overapproximation $q'$, and compute it if possible. 
This might allow us to optimize $q'$ before evaluating it. 
There is also a difference in complexity, 
as existential pebble game techniques
are \ptime-complete in general \cite{KP03}, and thus inherently sequential, 
while evaluation of CQs in $\HW(k)$ is highly parallelizable (Gottlob et al. \cite{GLS02}). 

By exploiting automata techniques, we show that checking if a CQ $q$ has a (finite) $\HW(1)$-overapproximation $q'$ 
is in {\sc 2Exptime}. Also, when such $q'$ exists it can be computed in 
{\sc 3Exptime}. 
This is important since $\HW(1)$ coincides with the well-known class of {\em acyclic} CQs \cite{Yan81}. 
If the arity of the schema is fixed, these bounds drop to {\sc Exptime} and {\sc 2Exptime}, respectively. 
Also, we look at the case of binary schemas, which are for instance used in \emph{graph databases} \cite{Bar13} and {\em description logics} \cite{DL-handbook}. 
In this case, we show that for \emph{Boolean} CQs, $\HW(1)$-overapproximations can be computed efficiently via a greedy algorithm. 
This is optimal, as over ternary schemas we prove an exponential lower bound for the size of 
$\HW(1)$-overapproximations. 

We do not know if the existence problem is decidable for $k > 1$. However, we show that 
it can be recast as an 
unexplored boundedness condition for the existential pebble game. Understanding the decidability boundary for such
conditions is often difficult \cite{Otto2006,BOW14}.

\item
\underline{Overapproximations under constraints.} It has been observed that semantic information about the data, in the form 
of integrity constraints, enriches the quality of approximations \cite{BGP16}. 
This is based on the fact that approximations are now defined over a restricted set of databases; namely, those that satisfy the constraints. 

We study  
$\HW(k)$-overapproximations for the practical classes of {\em equality-generating dependencies} (egds), which subsume functional dependencies, 
and {\em tuple-generating dependencies} (tgds), which subsume inclusion dependencies. 
By extending our previously derived techniques, we show that each CQ $q$ admits an infinite $\HW(k)$-overapproximation under a set of constraints~$\Sigma$, in any of the following cases: 
 $\Sigma$ consists exclusively of egds, 
or $\Sigma$ is a set of {\em guarded} tgds \cite{CaGK08a}. Recall that the latter corresponds to a well-studied extension of the
class of inclusion dependencies. 
Such an infinite $\HW(k)$-overapproximation can be evaluated in polynomial time for the case of functional dependencies and guarded tgds 
(and so, for inclusion dependencies), and by a {\em fixed-parameter tractable} algorithm for egds.

\end{enumerate}
Figure \ref{table:summary} shows a summary of these results in comparison with previously known results about 
underapproximations (in the absence of constraints). 
 
\begin{figure}
\centering 
\scalebox{0.87}{
\begin{tabular}{|c|c|c|c|l|l|}
\hline
& {\bf Existence?} & {\bf Unique?} & {\bf Evaluation} & {\bf Existence check} \\ \hline
\emph{$\HW(k)$-underapp.} & always &  not always & NP-hard  & N/A \\ \hline
\emph{$\HW(k)$-overapp.} & not always & always & \ptime$^*$ & For $k = 1$:\\ \cline{6-6} 
& & & & \hspace{4mm} {\sc 2Exptime}$^*$   \\ \cline{6-6}
& & & & \hspace{4mm} \ptime$^*$ on binary schemas  \\ \cline{6-6}
& & & & For $k > 1$: Open \\ 
\hline 
\end{tabular}
}
\caption{Summary of results on under- and overapproximations of bounded generalized hypertreewidth (in the absence of constraints). 
The complexity of identification coincides with that of evaluation in both cases.  
New results are marked with ($^*$). All remaining results follow from \cite{BLR12,BLR14}.}
\label{table:summary}

\end{figure}

Our contributions for $\HW(k)$-$\Delta$-approximations are as follows. 
As a preliminary step, we show that $\HW(k)$-under- and $\HW(k)$-overapproximations are particular cases of $\HW(k)$-$\Delta$-approximations, 
but not vice versa. 
Afterwards, as for $\HW(k)$-overapproximations, we 
provide a link between $\HW(k)$-$\Delta$-overapproximations and the existential pebble game, 
and use it to characterize when a CQ $q$ has at least one $\HW(k)$-$\Delta$-approximation that is neither a 
$\HW(k)$-underapproximation nor a $\HW(k)$-overapproximation (a so-called {\em incomparable $\HW(k)$-$\Delta$-approximation}). 
This allows us to show that the identification problem for such 
$\Delta$-approximations is coNP-complete. 
As for the problem of checking for the existence of incomparable $\HW(k)$-$\Delta$-approximations, we
extend our automata techniques to prove that it is in {\sc 2Exptime} for $k = 1$ 
(and in {\sc Exptime} for fixed-arity schemas). In case such a
$\HW(1)$-$\Delta$-approximation exists, we can evaluate it using a fixed-parameter tractable algorithm. 
We also provide results on existence and evaluation of infinite incomparable $\HW(1)$-$\Delta$-approximations. 

\paragraph{{\em {\bf Organization.}}} 
Section \ref{preliminaries} contains background notions and results, while Section \ref{approximations} introduce  
approximations. 
Basic properties of overapproximations are presented in Section \ref{overapproximations},
while the existence of overapproximations is studied in Section \ref{decidability}. 
In Section \ref{constraints} we study overapproximations under constraints, while in Section 
\ref{beyond} we deal with $\Delta$-approximations. We conclude in 
Section \ref{section:conclusion}
with final remarks. 

\section{Preliminaries} \label{preliminaries} 

\subsection{Databases, homomorphisms, and conjunctive queries}

\paragraph{{\em {\bf Databases and homomorphisms.}}}
A {\em relational schema} $\sigma$ is a finite set of relation symbols, each one of which has
an arity $n > 0$. A {\em database} $\D$ over $\sigma$ is a finite 
set of facts of the form $R(\bar a)$, where $R$ is a relation symbol in $\sigma$ of arity $n$ and 
$\bar a$ is an $n$-tuple of constants. We often abuse notation and write $\D$ also for the set of constants mentioned 
 in the facts of $\D$. 

Let $\D$ and $\D'$ be databases 
over $\sigma$. A {\em homomorphism} from $\D$ to $\D'$ 
is a mapping $h$ from $\D$ to $\D'$ such that for every atom  
$R(\bar a)$ in $\D$ it is the case that $R(h(\bar a)) \in \D'$. Here, we use the convention that if 
$\bar a = (a_1,\dots,a_n)$ then $h(\bar a) = (h(a_1),\dots,h(a_n))$.
If $\bar a$ and $\bar b$ are $n$-ary  
tuples (for $n \geq 0$) in $\D$ and $\D'$, respectively, we 
write $(\D,\bar a) \to (\D',\bar b)$ if there is a homomorphism $h$ from $\D$ to $\D'$ such that 
$h(\bar a) = \bar b$. Checking if $(\D,\bar a) \to (\D',\bar b)$ is a well-known \np-complete problem. 


\paragraph{{\em {\bf Conjunctive queries.}}}
A {\em conjunctive query} (CQ) over schema $\sigma$ is a formula of the form  
$q(\bar x)=\exists \bar y \bigwedge_{1 \leq i \leq m} R_i(\bar x_i)$, where  each $R_i(\bar x_i)$ is an atom over $\sigma$ for each $i$ with $1 \leq i \leq m$, 
and the tuple $\bar x$ contains precisely the {\em free variables}, i.e., the variables that do not appear existentially quantified in $\bar y$. 
Note that our definition allows for repetitions of variables in $\bar x$. 
We often refer implicitly to the free variables $\bar x$ and write  $q$ for $q(\bar x)$.  
If $\bar x$ is empty, then $q$ is {\em Boolean}.  

As customary, we define the evaluation of CQs in terms of homomorphisms. 
Recall that the {\em canonical database} $\D_q$ of a CQ $q(\bar x)=\exists \bar y \bigwedge_{1 \leq i \leq m} R_i(\bar x_i)$ 
consists precisely of 
the atoms $R_i(\bar x_i)$, for $1 \leq i \leq m$. 
We then define the {\em result of $q$ over $\D$}, denoted 
$q(\D)$, as the set of all tuples $\bar a$ of elements (i.e., constants) in $\D$ 
such that $(\D_q,\bar x) \to (\D,\bar a)$. 
We often do not distinguish between a CQ $q$ and its canonical database $\D_q$, i.e., we write $q$ for $\D_q$. 
If $q$ is Boolean, then its evaluation over a database $\D$ correspond to the Boolean
values {\sf true} or {\sf false} depending on whether $q(\D) = \{()\}$ or $q(\D) = \emptyset$, respectively. 


\paragraph{{\em {\bf Evaluation and tractable classes of CQs.}}}
The {\em evaluation problem for CQs} is as follows: 
Given a CQ $q$, a database~$\D$, and a tuple $\bar a$ of elements 
in~$\D$, is $\bar a \in q(\D)$? Since this problem corresponds to checking if $(q,\bar x) \to (\D,\bar a)$,  
it is \np-complete \cite{CM77}. 
This led to a flurry of activity for finding 
classes of CQs for which evaluation is tractable 
(see, e.g., \cite{Yan81,CR00,GLS02,GM06}).

Here we deal with one of the most studied such classes: CQs of 
bounded {\em generalized hypertreewidth} \cite{GLS02}, also 
called \emph{coverwidth} \cite{CD05}. 
We adopt the definition of \cite{CD05} which is better suited for working with non-Boolean queries. 
A \emph{tree decomposition} of a CQ $q = \exists \bar y \bigwedge_{1 \leq i \leq m} R_i(\bar x_i)$ is a pair $(T,\chi)$, where $T$ is a tree and $\chi$  is a mapping that assigns a subset of the existentially quantified variables in $\bar y$ to each node $t \in T$, such that:
%
\begin{enumerate}
  \item For each $1 \leq i \leq m$, the variables in $\bar x_i \cap \bar y$ are contained in $\chi(t)$, for some $t \in T$.
  \item For each variable $y$ in $\bar y$, the set of nodes $t \in T$ for which $y$ occurs in $\chi(t)$ is connected. 
\end{enumerate} 

The {\em width} of  node $t$ in $(T,\chi)$ is the minimal  
size of an $I \subseteq \{1,\dots,m\}$ such that $\bigcup_{i \in I} \bar x_i$ 
covers $\chi(t)$ (where we slightly abuse notation and write $\bar x_i$ also for the set 
of variables mentioned in the tuple $\bar x_i$). The width of $(T,\chi)$ is the maximum width of the nodes of $T$.  
The \emph{generalized hypertreewidth} 
of $q$ is the minimum width of its tree decompositions.

For a fixed $k \geq 1$, we denote by $\HW(k)$ the class of CQs of generalized hypertreewidth 
at most $k$. 
The CQs in $\HW(k)$ can be evaluated in polynomial time; see \cite{GGLS16}.

\paragraph{{\em {\bf Containment of CQs.}}}
A CQ $q$ is {\em contained} in a CQ $q'$, written as $q\subseteq q'$,
if $q(\D) \subseteq q'(\D)$ over every database $\D$.
Two CQs $q$ and $q'$ are {\em equivalent}, denoted $q \equiv q'$, if $q \subseteq q'$ 
and $q' \subseteq q$. 

It is known 
that CQ containment and CQ evaluation are, essentially, the same problem \cite{CM77}. In particular, 
%
let $q(\bar x)$ and $q'(\bar x')$ be CQs. Then  
\begin{equation} 
\label{eq:cont} 
q \subseteq q' \quad \Longleftrightarrow \quad \bar x \in q'(\D_q)  \quad \Longleftrightarrow \quad 
(\D_{q'},\bar x') \to (\D_q,\bar x).
\end{equation} 
Thus, $q \subseteq q'$ and $(q',\bar x')\to (q,\bar x)$  -- i.e., $(\D_{q'},\bar x') \to (\D_q,\bar x)$ --  
are used interchangeably. 

\paragraph{{\em {\bf Cores of CQs.}}}
A CQ $q$ is a {\em core} \cite{CM77,HN92} if there is no CQ $q'$ with fewer atoms than $q$ such that $q \equiv q'$. 
Given CQs $q$ and $q'$, we say that $q'$ is a core of $q$ if $q'$ is a core and 
$q \equiv q'$. In other words, a core of $q$ is a minimal CQ (in terms of number of atoms) that is equivalent to $q$. 
The following result summarizes some important properties of cores.  

\begin{proposition} \label{prop:cores} 
{\em \cite{HN92}} The following statements hold:  
\begin{enumerate}
\item Each CQ $q$ has a core. Moreover, there is a unique core of $q$ 
up to renaming of variables. Therefore, we can talk about {\em the} core of $q$. 
\item Each core $q'$ of $q$ is a {\em retract} of $q$. That is, each atom of $q'$ is also an atom of $q$ and 
we can choose the homomorphism $h$ witnessing $(q,\bar x)\to (q',\bar x')$ to be a $\emph{retraction}$, 
i.e., to be the identity on the variables of $q'$. 
\end{enumerate} 
\end{proposition}


\subsection{The existential cover game}
Several results in the paper require applying techniques based on  
existential pebble games. 
We use a version of the {\em existential cover game}, 
that is tailored for CQs of bounded generalized hypertreewidth \cite{CD05}. 

Let $k \geq 1$. The existential $k$-cover game is played by \emph{Spoiler} and \emph{Duplicator}
on pairs $(\D,\bar a)$ and $(\D',\bar b)$,   
where $\D$ and $\D'$ are databases 
and $\bar a$ and $\bar b$ are $n$-ary (for $n \geq 0$) tuples over the elements (i.e., the constants) in 
$\D$ and $\D'$, respectively. 
The game proceeds in
rounds. 
In each round, Spoiler places (resp., removes) a
pebble on (resp., from) an element in $\D$, and Duplicator responds by
placing (resp., removing) 
its corresponding pebble on an element in (resp., from)~$\D'$.
The number of pebbles is not bounded, 
but Spoiler is constrained as follows: At any round $p$ of the
game, if $c_1,\dots,c_\l$ ($\l \leq p$)
are the elements marked by Spoiler's pebbles in~$\D$, 
there must be a set of at most $k$ atoms in $\D$ that contain all such elements (this
is why the game is called $k$-cover game, as pebbled elements are {\em covered} by no more than $k$ atoms).

Duplicator wins if she has a 
{\em winning strategy}, 
i.e., if she can indefinitely continue playing the game in such a way that 
after each round, if $c_1,\dots,c_\l$ are the elements that are marked by Spoiler's pebbles in $\D$ 
and $d_1,\dots,d_\l$  are the elements marked by the corresponding pebbles of Duplicator in~$\D'$, then 
$$\big((c_1,\dots,c_\l,\bar a),(d_1,\dots,d_\l,\bar b)\big)$$ 
is a {\em partial homomorphism} from $\D$ to~$\D'$. 
That is, for every atom $R(\bar c) \in \D$, where each element $c$ of $\bar c$ 
appears in $(c_1,\dots,c_\l,\bar a)$, 
it is the case that $R(\bar d) \in \D'$, where $\bar d$ is the tuple obtained from $\bar c$ 
by replacing each element $c$ of 
$\bar c$ by its corresponding element $d$ in $(d_1,\dots,d_\l,\bar b)$. 
We write $(\D,\bar a) \to_{k} (\D',\bar b)$ if Duplicator has a winning strategy.

Notice that if $(\D, \bar a) \to (\D', \bar b)$ via some homomorphism $h$ then $h$ witnesses that Duplicator can win the $k$-pebble game for each $k$. Hence,
$\to_k$ can be seen as an ``approximation'' of $\to$:
$$\to \ \subset  \ \cdots \ \subset \ \to_{k+1}  \
\subset  \ \to_k  \ \subset  \ \cdots  \ \subset  \ \to_{1}.$$   
These approximations are convenient complexity-wise: 
Checking if $(\D,\bar a) \to (\D',\bar b)$ is \np-complete, 
while $(\D,\bar a) \to_k (\D',\bar b)$ can be solved
efficiently. 

\begin{proposition}  \label{prop:games-poly}  {\em \cite{CD05}}  
Fix $k \geq 1$. Checking whether $(\D,\bar a) \to_k (\D',\bar b)$
is in polynomial time.  
\end{proposition} 

Moreover, there is a connection between $\to_k$ and the evaluation of 
CQs in $\HW(k)$ that we heavily exploit in our work. 

\begin{proposition} \label{prop:games-tw} {\em \cite{CD05}}   
Fix $k \geq 1$. 
 Then 
 $(\D,\bar a) \to_{k} (\D',\bar b)$ iff for each CQ $q(\bar x)$ in $\HW(k)$ we have that 
if $(q,\bar x) \to (\D,\bar a)$ then $(q,\bar x) \to (\D',\bar b)$. 
\end{proposition} 

In particular, if $q(\bar x) \in \HW(k)$ then for every $\D$ and $\bar a$ it is the case that 
\begin{equation} \label{eq:tw} 
\bar a \in q(\D) \quad \Longleftrightarrow \quad (q,\bar x) \to (\D,\bar a) 
\quad \Longleftrightarrow \quad (q,\bar x) \to_k (\D,\bar a).  
\end{equation} 
That is, the ``approximation'' of $\to$ provided by $\to_{k}$
is sufficient for evaluating CQs in $\HW(k)$. 
Together with Proposition \ref{prop:games-poly}, this proves that CQs in $\HW(k)$
can be evaluated efficiently. 

As a matter of fact, the equivalences established in Equation \eqref{eq:tw} hold even if $q$ itself 
is not in $\HW(k)$, but its core $q'$ is in $\HW(k)$. That is, the evaluation problem for the class 
of CQs whose core is in $\HW(k)$ can be solved efficiently via the existential $k$-cover game. 
As established by Greco and Scarcello, this is precisely the boundary for when 
this good property holds. Indeed, 
for any Boolean CQ $q$ whose core is not in $\HW(k)$ it is possible to find a database $\D$ such that 
$q \to_k \D$ but $q \not\to \D$ \cite{GS17}.    

\subsection{Expressing the existential cover game as a CQ in $\HW(k)$}
\label{sec:compact}

An instrumental tool in several of our results is that for any given pair $(\D,\bar a)$, where $\D$ is a database and $\bar a$ is a tuple of elements (i.e., constants) 
over $\D$, we can construct a CQ $q(\bar x)$ in $\HW(k)$ that represents 
the possible moves of Spoiler in the existential $k$-cover 
game played from $\D$, but only up to certain number of rounds. 

For simplicity, we consider a {\em compact} version of the existential $k$-cover game as in~\cite{CD05}. 
For a database $\D$, we say that a set $S$ of elements of $\D$ is a \emph{$k$-union} if there exist $p$ atoms $R_1(\bar a_1),\dots R_p(\bar a_p)\in \D$ with $p\leq k$, 
such that $S=\bigcup_{1\leq i\leq p} \bar a_i$. That is, the $k$-unions of $\D$ are the sets that appear in a union of at most $k$ atoms of $\D$. 
In the compact existential $k$-cover game, Spoiler is allowed, in each round, to remove and place as many pebbles as desired, as long as 
the resulting pebbled elements form a $k$-union. (Notice the difference with the 
standard existential $k$-cover game, in which in each round Spoiler is allowed to either remove or place 
exactly one pebble). As before, 
Duplicator wins the compact existential $k$-cover game on $(\D,\bar a)$ and $(\D',\bar b)$ 
iff she has a winning strategy, i.e., she can indefinitely continue playing the game in such a way that after each round, 
if $c_1,\dots,c_\l$ are the elements marked by Spoiler's pebbles in $\D$ 
and $d_1,\dots,d_\l$  are the elements marked by the corresponding Duplicator's pebbles in~$\D'$, then 
$((c_1,\dots,c_\l,\bar a),(d_1,\dots,d_\l,\bar b))$ 
is a partial homomorphism from $\D$ to~$\D'$. 

As can be easily seen, the compact existential $k$-cover game is not more powerful than the standard one. 
That is, for each $k \geq 1$ 
Duplicator wins the existential $k$-cover game on pairs $(\D,\bar a)$ and $(\D',\bar b)$ iff 
she wins the compact existential $k$-cover game on such pairs \cite{CD05}.

We then write $(\D,\bar a) \to_{k}^c (\D',\bar b)$, for $k \geq 1$ and $c \geq 1$, if Duplicator 
has a winning strategy {\em in the first $c$ rounds} of the (compact) existential $k$-cover 
game on $(\D,\bar a)$ and $(\D',\bar b)$. That is, $(\D,\bar a) \to_k (\D',\bar b)$ if and only if $(\D,\bar a) \to_k^c (\D',\bar b)$ for every $c \geq 1$. 
The following is our desired technical result, which establishes that there is a CQ in $\HW(k)$ that represents 
the possible moves of Spoiler -- up to $c$ rounds -- in the existential $k$-cover 
game played from $(\D,\bar a)$.

\begin{lemma} \label{lemma:tedious} 
Let $\D$ be a database and $\bar a$ an $n$-ary tuple of elements in $\D$, for $n \geq 0$. 
For each $c \geq 1$ there is a CQ $q_c(\bar x_c)$ in $\HW(k)$ 
such that for each database $\D'$ and tuple $\bar b$ of the same arity as $\bar a$ we have that 
$$\bar b \in q_c(\D') \ \ \Longleftrightarrow \ \ 
(q_c,\bar x_c) \to (\D',\bar b) \ \ \Longleftrightarrow \ \ (\D,\bar a) \to_{k}^c (\D',\bar b).$$ 
\end{lemma} 

Results of a similar kind have been obtained in \cite{KV95}, and actually the proof 
of Lemma \ref{lemma:tedious} can easily be obtained 
by adapting techniques in such paper. We provide a proof, nevertheless, 
as we use it in several results along the article. 

\begin{proof}[Lemma \ref{lemma:tedious}] Fix $c\geq 1$. We shall define a database $\D_c$ such that  $(\D_c,\bar a) \to (\D',\bar b)$ iff $(\D,\bar a) \to_{k}^c (\D',\bar b)$, for every pair $(\D',\bar b)$. 
The lemma will follow by choosing $q_c(\bar x_c)$ such that $(\D_{q_c}, \bar x_c)=(\D_c,\bar a)$. 

Let $S_1,\dots,S_N$ be an enumeration of all the $k$-unions of $\D$, where $N\geq 1$.  
In order to define $\D_c$, we first consider a pair $(T_c,\lambda_c)$ where $T_c$ is an ordered rooted tree (i.e., the children of a node are ordered) with rank $N$ and height $c$, and $\lambda_c$ is a labelling that maps each $t\in T_c$ to a $k$-union of $\D$ such that 
(i) $\lambda_c(r)=\emptyset$, for the root $r$, and (ii) if $t_i$ is the $i$-th child of $t$, then $\lambda_c(t_i)=S_i$. 
The pair $(T_c,\lambda_c)$ is simply a representation of all the possible moves of Spoiler on $\D$ in the compact existential $k$-cover game, up to $c$ rounds. 

Now we define a pair $(T_c,\beta_c)$ as follows. 
The intuition is that $(T_c,\beta_c)$ is obtained from $(T_c,\lambda_c)$ via renaming elements in the $\lambda_c(t)$'s in such a way that  
$(T_c,\lambda_c)$ satisfies the connectivity condition of tree decompositions. 
An \emph{occurrence} of an element $d\in \D$ in $(T_c,\lambda_c)$ is a maximal subtree $T$ of $T_c$ 
such that $d\in \lambda_c(t)$, for all $t\in T$. For each $d\in \D\setminus \bar a$ and each occurrence $T$ of $d$ in $(T_c,\lambda_c)$, we introduce a new element $e_{d,T}$, and modify each $\lambda_c(t)$ with $t\in T$ 
by replacing $d$ with $e_{d,T}$. The resulting labelling is $\beta_c$. Note that, for each $t\in T_c$, there is a bijection $\Phi_t$ from $\lambda_c(t)$ to $\beta_c(t)$. 
By construction, $\Phi_t(a)=a$, for all $a\in \lambda_c(t)\cap \bar a$, and $\Phi_t(d)=\Phi_{t'}(d)$, where $t'$ is the parent of $t$ and $d\in \lambda_c(t)\cap \lambda_c(t')$. 

The database $\D_c$ is obtained as follows. For each atom $R(\bar d)$ in $\D$ and node $t\in T_c$ such that $\bar d\subseteq \lambda_c(t)$, 
we add to $\D_c$ the atom $R(\Phi_t(\bar d))$. Note that the same atom $R(\bar d)$ can produce several atoms in $\D_c$. 
By definition, $\Phi_t: \lambda_c(t)\to \beta_c(t)$ and $\Phi_t^{-1}: \beta_c(t)\to \lambda_c(t)$ are partial homomorphisms from $\D$ to $\D_c$ and $\D_c$ to $\D$, respectively, for all $t\in T_c$. 
To see that the resulting CQ $q_c(\bar x_c)$ is in $\HW(k)$, we consider the pair $(T_c,\chi_c)$, 
where for each $t\in T_c$, $\chi_c(t)=\beta_c(t)\setminus \bar a$. By construction, $(T_c,\chi_c)$ is a tree decomposition of $q_c(\bar x_c)$ of width $k$. 

Let us conclude by stressing that $(\D_c,\bar a) \to (\D',\bar b)$ iff $(\D,\bar a) \to_{k}^c (\D',\bar b)$, for every pair $(\D',\bar b)$. 
Suppose that $(\D_c,\bar a) \to (\D',\bar b)$ via a homomorphism $h$. If the first move of Spoiler is the $k$-union $S_i$ of $\D$, for some $i\in \{1,\dots, N\}$, 
then the Duplicator responds with the partial homomorphism $h\circ \Phi_{t_i}$, where $t_i$ is the $i$-th child of the root $r$ (and hence, the domain of $\Phi_{t_i}$ is $\lambda_c(t_i)=S_i$). 
If in the next round Spoiler plays $S_j$ for some $j\in \{1,\dots, N\}$, then the Duplicator responds with $h\circ \Phi_{t_j}$, where $t_j$ is the $j$-th child of $t_i$. 
This is a valid move for Duplicator since, as mentioned above, $\Phi_{t_i}$ and  $\Phi_{t_j}$ are consistent. 
In this way, Duplicator can win the first $c$ rounds of the game as the height of $T_c$ is $c$. 

Conversely, assume that $(\D,\bar a) \to_{k}^c (\D',\bar b)$. 
We can define a homomorphism $h$ from $\D_c$ to $\D'$ such that $h(\bar a)=\bar b$ in a top-down fashion over $T_c$. 
We start with the root $r$, and make Spoiler play $\Phi_r^{-1}(\beta_c(r))$ as his first move. 
Suppose Duplicator responds with a partial homomorphism $g_0:\lambda_c(r)\to \D'$. Then $h_0=g_0\circ \Phi_r^{-1}$ 
is a partial homomorphism from $\beta_c(r)$ to $\D'$. 
Following this argument, and making Spoiler play accordingly, we can extend $h_0$ to partial homomorphisms $h_1,\dots,h_c$, 
where the domain of $h_p$, with $p\in\{1,\dots, c\}$, are the variables in some $\beta_c(t)$ such that the distance from $t$ to root $r$ is at most $p$. 
We can define $h=h_c$ to be our desired homomorphism from $\D_c$ to $\D'$. 
\qed
\end{proof} 

\section{Approximations of CQs} \label{approximations}  

Fix $k \geq 1$. 
Let $q$ be a CQ. 
The approximations of $q$ in $\HW(k)$ are  defined with respect to a partial order $\sqsubseteq_q$ over 
the set of CQs in $\HW(k)$. Formally, for any two  CQs $q',q''$ in $\HW(k)$ we have
$$q'\sqsubseteq_q q''  \ \ \Longleftrightarrow \ \ \Delta(q(\D),q''(\D)) \, \subseteq \, \Delta(q(\D),q'(\D)), \text{ for every database $\D$,}$$
where $\Delta(A,B)$ denotes the symmetric difference between sets $A$ and $B$. Thus, 
 $q'\sqsubseteq_q q''$, whenever the ``error'' of $q''$ with respect to $q$ -- measured in terms of the symmetric difference between 
 $q''(\D)$ and $q(\D)$ -- is contained in that of $q'$ for each database $\D$.
As usual, we write 
$q' \sqsubset_q q''$ if $q' \sqsubseteq_q q''$ but $q'' \not\sqsubseteq_q q'$. 

The approximations of $q$ in $\HW(k)$ always correspond to maximal elements, with respect to the partial order $\sqsubseteq_q$, 
over a class of CQs in $\HW(k)$ that satisfies certain conditions. 
The following three basic notions of approximation were identified in \cite{BLR12}: 

\begin{itemize} 
\item \underline{Underapproximations:}  Let $q,q'$ be CQs such that $q' \in \HW(k)$. Then $q'$ is a {\em $\HW(k)$-underapproximation} 
of $q$ if it is maximal, with respect to $\sqsubseteq_q$, among all CQs in $\HW(k)$ that are contained in $q$. That is, it holds that  
$$\text{$q' \subseteq q$, and there is no CQ $q'' \in \HW(k)$ such that $q'' \subseteq q$ and $q' \sqsubset_q q''$.}$$ 
In particular, the 
$\HW(k)$-underapproximations of $q$ produce correct (but not necessarily complete) answers with respect to $q$ over every database $\D$. 
  
\item \underline{Overapproximations:}  Analogously, $q'$ is a {\em $\HW(k)$-overapproximation} 
of $q$ if it is maximal, with respect to $\sqsubseteq_q$, among all CQs in $\HW(k)$ that contain $q$. That is, it holds that 
$$\text{$q \subseteq q'$, and there is no CQ $q'' \in \HW(k)$ such that $q \subseteq q''$ and $q' \sqsubset_q q''$.}$$
Hence,  
$\HW(k)$-overapproximations of $q$ produce complete (but not necessarily correct) answers with respect to $q$
 over every database $\D$.
  
\item \underline{$\Delta$-approximations:}  
In this case we impose no restriction on $q'$. 
That is, $q'$ is a {\em $\HW(k)$-$\Delta$-approximation} of $q$ if it is maximal with respect to the partial order $\sqsubseteq_q$, i.e., 
there is no $q''\in \HW(k)$ such that 
$q' \sqsubset_q q''$. 
\end{itemize} 

Underapproximations and overapproximations admit an equivalent, but arguably simpler 
characterization as maximal (resp., minimal) elements,
with respect to the containment partial order $\subseteq$, among all CQs in $\HW(k)$ that are contained in $q$ (resp., contain $q$). We 
state this characterization next. 

\begin{proposition}  {\em \cite{BLR12}}  
Fix $k \geq 1$. Let $q,q'$ be CQs such that $q' \in \HW(k)$. Then:
\begin{itemize} 
\item $q'$ is a $\HW(k)$-underapproximation of $q$ iff $q' \subseteq q$ and there is no CQ $q'' \in \HW(k)$ such that 
$q' \subset q'' \subseteq q$. 
\item $q'$ is a $\HW(k)$-overapproximation of $q$ iff $q \subseteq q'$ and there is no CQ $q'' \in \HW(k)$ such that 
$q \subseteq q'' \subset q'$. 
 \end{itemize}   
\end{proposition}

The basic theoretical properties of 
$\HW(k)$-underapproximations are by now well-understood \cite{BLR14}. We concentrate on $\HW(k)$-overapproximations and 
$\HW(k)$-$\Delta$-approximations in this paper. We start by studying the former.

\section{Overapproximations}\label{overapproximations} 
Recall that  
$\HW(k)$-overapproximations are minimal elements (in terms of $\subseteq$) 
in the set of CQs in $\HW(k)$ that contain $q$. We show an example of a $\HW(1)$-overapproximation below.  

\begin{example}
\label{ex:overapp} 
Figure \ref{fig:fig1} shows a CQ $q$ and its $\HW(1)$-overapproximation $q'$. The schema consists of binary symbols
$P_a$ and $P_b$. Nodes represent variables, and an edge labeled $P_a$ between $x$ and 
$y$ represents the presence of atoms 
$P_a(x,y)$ and $P_a(y,x)$. (Same for $P_b$). All variables are existentially quantified. 
Clearly, $q \subseteq q'$ (as $q' \to q$). In addition, there is no CQ $q'' \in \HW(1)$ such 
that $q \subseteq q'' \subset q'$. We provide an explanation for this later.  \qed
\end{example}

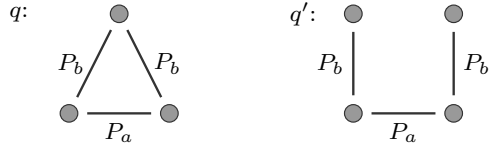
\begin{figure}
\centering 
  \scalebox{1.1}{
  \begin{tikzpicture}[
      xscale=0.6,
      yscale=0.6,      
    ]
      \node (0) at (-1,2){$q$:};
      
      \node[mnode] (1) at (0,0){};
      \node[mnode] (2) at (2,0){};
      \node[mnode] (3) at (1,2){};

      \draw [uEdge] (1) to node[below]{$P_a$} (2);
      \draw [uEdge] (2) to node[right]{$P_b$} (3);
      \draw [uEdge] (3) to node[left]{$P_b$} (1);
  \end{tikzpicture}\hspace{10mm}
    \begin{tikzpicture}[
      xscale=0.6,
      yscale=0.6,
    ]
      \node (0) at (-1,2){$q'$:};
      
      \node[mnode] (1) at (0,0){};
      \node[mnode] (2) at (2,0){}; 
      \node[mnode] (3) at (0,2){};      
      \node[mnode] (4) at (2,2){};
      
      \draw [uEdge] (1) to node[below]{$P_a$} (2);
      \draw [uEdge] (1) to node[left]{$P_b$} (3);
      \draw [uEdge] (2) to node[right]{$P_b$} (4);
  \end{tikzpicture}
  }
\caption{The CQ $q$ and its $\HW(1)$-overapproximation $q'$ from Example \ref{ex:overapp}.} 
\label{fig:fig1}
\end{figure} 

We start in Section \ref{sec:uniqueness} by stating some basic properties on 
existence and uniqueness of $\HW(k)$-overapproximations. 
Later in Section \ref{sec:pg} we establish 
a connection between $\HW(k)$-overapproximations and the existential pebble game, which allows us to show 
that both the identification and evaluation problems for $\HW(k)$-overapproximations are tractable. 
Finally, in Section~\ref{sec:liberal} we look at the case when $\HW(k)$-overapproximations do not exist, and suggest how this can 
be alleviated by allowing infinite overapproximations. 

\subsection{Existence and uniqueness of overapproximations} 
\label{sec:uniqueness} 

As  shown in \cite{BLR12}, existence of overapproximations is not a general phenomenon. In fact, for every $k>1$ there is a Boolean CQ $q$ in $\HW(k)$ that has no $\HW(1)$-overapproximation. 
Using the characterization given later in Theorem~\ref{theo:boundedness}, we can strengthen this further. 

\begin{theorem}
\label{theo:non-existence} 
For each $k>1$, there is a Boolean CQ $q \in \HW(k)$ without $\HW(\ell)$-overapproximations for any $1 \leq \ell<k$. 
\end{theorem}

Figure \ref{fig:fig2} depicts examples of CQs in $\HW(k)$, for $k = 2$ and $k = 3$, respectively, 
without $\HW(\ell)$-overapproximations for any $1 \leq \ell<k$. The proof of Theorem \ref{theo:non-existence} is long 
and quite technical, and for such a reason we relegate it to the appendix.  

Interestingly, when $\HW(k)$-overapproximations do exist, they are unique (up to equivalence). 
This is because, in this case, $\HW(k)$-overapproximations are not only the minimal elements, but also the lower bounds of the set of CQs in 
$\HW(k)$ that contain $q$. 

In order to show uniqueness, we need to introduce a simple construction that also will be useful when studying $\Delta$-approximations (Section \ref{beyond}). 
Let $q(\bar x)$ and $q'(\bar x')$ be two CQs such that $|\bar x|=|\bar x'|=n$. 
The \emph{disjoint conjunction} of $q$ and $q'$ is the CQ $(q\wedge q')(\bar z)$, with $|\bar z|=n$ defined as follows. 
First we rename 
each existentially quantified variable in $q$ and $q'$ with a different fresh variable, and then take the conjunction of the atoms in 
$q$ and~$q'$. Finally, if $\bar x=(x_1,\dots,x_n)$ and $\bar x'=(x'_1,\dots,x'_n)$, we identify $x_i$ and $x'_i$ for every $1\leq i\leq n$. 
The $i$-th variable of $\bar z$ is the variable obtained after identifying $x_i$ and $x'_i$. 
By construction, the following hold: 
\begin{enumerate}
\item $(q,\bar x)\to (q\wedge q', \bar z)$ and $(q',\bar x')\to (q\wedge q', \bar z)$,
\item if $(q,\bar x)\to (p, \bar w)$ and $(q',\bar x')\to (p, \bar w)$, for some CQ $p(\bar w)$, then $(q\wedge q', \bar z)\to (p, \bar w)$, and 
\item if $q,q'\in \HW(k)$ for $k\geq 1$, then $q\wedge q'\in \HW(k)$. 
\end{enumerate}
Note that property (1) and (2) tell us that $q\wedge q'$ is the \emph{least upper bound} of $q$ and $q'$ with respect to the order $\to$. 
We have the following result: 

\begin{proposition}
\label{prop:over-lower}  
Let $q(\bar x)$ and $q'(\bar x')$ be CQs such that $q' \in \HW(k)$. The following are equivalent: 
\begin{enumerate}
\item  $q'(\bar x')$ is a $\HW(k)$-overapproximation of $q(\bar x)$. 
\item (i) $q \subseteq q'$, and (ii) for every  CQ $q''(\bar x'') \in \HW(k)$, it is the case that $q\subseteq q''$ implies $q'\subseteq q''$.
\end{enumerate} 
\end{proposition}

\begin{proof}
We only prove the nontrivial direction $(1)\Rightarrow (2)$. 
By contradiction, suppose that there is a CQ $q''\in \HW(k)$ such that $q\subseteq q''$ but $q'\not\subseteq q''$. 
Let $(q'\wedge q'')(\bar z)$ be the disjoint conjunction of $q'$ and $q''$. 
By definition, we have that $q'\wedge q''$ is in $\HW(k)$, and $q\subseteq (q'\wedge q'') \subseteq q'$. 
But $q'$ is a $\HW(k)$-overapproximation of $q$, and thus $q' \subseteq (q'\wedge q'')$. 
Again by construction of $q'\wedge q''$, we have that $(q'\wedge q'') \subseteq q''$, 
and then
$q'\subseteq q''$. This is 
a contradiction.  
\qed
\end{proof}

As a corollary, we immediately obtain the following. 

\begin{corollary}
\label{coro:over-unique} 
Consider a CQ $q$ with $\HW(k)$-overapproximations  $q_1$ and $q_2$.  
Then it is the case that $q_1\equiv q_2$. 
\end{corollary} 

\begin{figure}
\centering 

  \scalebox{1.1}{
  \begin{tikzpicture}[
      xscale=0.6,
      yscale=0.6,      
    ]
     \node (0) at (-1,2){$q$:};
      \node (0) at (-1,-0.55){};
     
      \node[mnode] (1) at (0,0){};
      \node[mnode] (2) at (2,0){};
      \node[mnode] (3) at (1,2){};

      \draw [dEdge] (2) to (1); 
      \draw [dEdge] (3) to (2); 
      \draw [dEdge] (1) to (3); 
  \end{tikzpicture}\hspace{10mm}

  \begin{tikzpicture}[
      xscale=0.6,
      yscale=0.6,
    ]
      \node (0) at (-1,2){$q'$:};
      \node (0) at (-1,-0.6){};
      
      \node[mnode] (1) at (0,0){};
            \node[mnode] (2) at (2,0){};
      \node[mnode] (3) at (0,2){};     
       \node[mnode] (4) at (2,2){};
      
            \draw [dEdge] (1) to (2);
      \draw [dEdge] (2) to (3);
      \draw [dEdge] (3) to (1);
      
      \draw [dEdge] (4) to (1);
      \draw [dEdge] (4) to (2);
      \draw [dEdge] (4) to (3);
  \end{tikzpicture}
  }
\caption{The CQ $q$ is in $\HW(2)$ but has no $\HW(1)$-overapproximations, while $q'$ is in $\HW(3)$ but 
has no $\HW(\ell)$-overapproximations for $\ell\in \{1,2\}$.} 
\label{fig:fig2}
\end{figure}
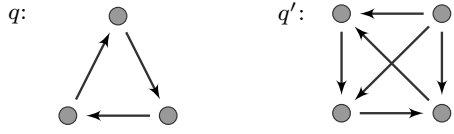

This result shows a stark difference between $\HW(k)$-overapproximations and $\HW(k)$-underapproximations: $\HW(k)$-overapproximations do not necessarily exist, but when they do they are unique; 
$\HW(k)$-underapproximations always exist but there can be exponentially many incomparable ones \cite{BLR14}. 

\subsection{A link with the existential pebble game} \label{sec:pg}


Existential cover games can be applied to obtain a semantic characterization of 
$\HW(k)$-overapproximations as follows.  

\begin{theorem} \label{theo:char} 
Fix $k\geq 1$. Let $q(\bar x)$ and $q'(\bar x')$ 
be CQs  with $q' \in \HW(k)$. The following are equivalent: 
\begin{enumerate}
\item $q'(\bar x')$ is the $\HW(k)$-overapproximation of $q(\bar x)$. 
\item 
$(q',\bar x') \to_{k} (q,\bar x)$ and $(q,\bar x) \to_{k} (q',\bar x')$. 
\end{enumerate} 
\end{theorem} 

\begin{proof} 
Assume that $q'(\bar x')$ is the $\HW(k)$-overapproximation of $q(\bar x)$. 
Then $(q',\bar x') \to (q,\bar x)$, and therefore 
$(q',\bar x') \to_{k} (q,\bar x)$. It remains to prove that \mbox{$(q,\bar x) \to_{k} (q',\bar x')$}. From  
Proposition \ref{prop:games-tw}, we need to prove that if $q''(\bar x'')$ is a CQ in 
$\HW(k)$ such that $(q'',\bar x'') \to (q,\bar x)$, then also $(q'',\bar x'') \to (q',\bar x')$. 
This follows directly from Proposition \ref{prop:over-lower}. 

Assume now that $(q',\bar x') \to_{k} (q,\bar x)$ and $(q,\bar x) \to_{k} (q',\bar x')$.
Since \mbox{$q'\in\HW(k)$} and $(q',\bar x') \to_{k} (q,\bar x)$, we have that 
$(q',\bar x') \to (q,\bar x)$ by Equation 
\eqref{eq:tw}, and hence 
$q \subseteq q'$ by Equation \eqref{eq:cont}. 
In addition, since $(q,\bar x) \to_{k} (q',\bar x')$ it follows from 
Proposition \ref{prop:games-tw} and Equation \eqref{eq:cont} that  
 if $q\subseteq q''$ and $q'' \in \HW(k)$, then~$q'\subseteq q''$. That is, 
 there is no  $q''$ in $\HW(k)$ such that $q \subseteq q'' \subset q'$. Hence $q'$ is a $\HW(k)$-overapproximation. \qed
\end{proof}

\begin{example} (Example \ref{ex:overapp} cont.) 
\label{ex:example2}
It 
is now easy to see that the CQ $q'$ in Figure~\ref{fig:fig1} is a $\HW(1)$-overapproximation of $q$. 
In fact, since $q' \to q$, we only need to show that $q \to_1 q'$. The latter is simple and left to the reader. 
\qed 
\end{example} 

Next we show that this characterization allows us to show that 
the identification and evaluation problems for $\HW(k)$-overapproximations can be solved in polynomial time. 

\subsection{Identification and evaluation of $\HW(k)$-overapproximations} 
A direct corollary of Proposition \ref{prop:games-poly} and 
Theorem \ref{theo:char} is that the {\em identification} problem for $\HW(k)$-overapproximations
is in polynomial time.

\begin{corollary}
Fix $k \geq 1$. Given CQs $q,q'$ such that $q' \in \HW(k)$, 
checking if $q'$ is the $\HW(k)$-overapproximation of $q$ can be 
solved in polynomial time. 
\end{corollary}

This corresponds to a {\em promise version} of the problem, as it is given to us that $q'$ is in fact in 
$\HW(k)$. 
 Checking the latter is NP-complete for every fixed $k > 1$~\cite{GMS09,FGP17}.

Let us assume now that we are given the {\em promise} that $q$ has a 
$\HW(k)$-overapproximation $q'$ (but $q'$ itself is not given). How hard is it to evaluate 
$q'$ over a database $\D$? We could try to compute~$q'$, 
 but so far we have no techniques to do that. Notably, 
we can use existential cover games to show that $\HW(k)$-overapproximations can be evaluated 
efficiently, without even computing them. This is based on the next result, which states 
that evaluating $q'$ over $\D$ boils down to checking $(q,\bar x) \to_{k} (\D,\bar a)$ for the tuples $\bar a$ over $\D$.

\begin{theorem} \label{theo:games-overapp} Consider a fixed $k \geq 1$. 
Let $q(\bar x)$ be a CQ with a $\HW(k)$-overapproximation~$q'(\bar x')$. Then for every $\D$ and $\bar a$ it is the case that  
$$
\bar a \in q'(\D)  \quad \Longleftrightarrow \quad (q',\bar x') \to (\D,\bar a) \quad \Longleftrightarrow \quad 
(q,\bar x) \to_{k} (\D,\bar a).
$$
\end{theorem} 

\begin{proof}
Assume first that $(q,\bar x) \rightarrow_{k} (\D,\bar a)$. 
Based on the fact that $q'$ is a $\HW(k)$-overapproximation of $q$, we have that 
$(q',\bar x') \to (q,\bar x)$. 
Since winning strategies for Duplicator compose and $(q',\bar x') \to (q,\bar x)$ implies \mbox{$(q',\bar x') \to_k (q,\bar x)$}, 
it is the case that 
$(q',\bar x')\rightarrow_{k} (\D,\bar a)$.
But $q' \in \HW(k)$, and thus \mbox{$(q',\bar x') \to (\D,\bar a)$} from Equation \eqref{eq:tw}.  

Assume, on the other hand, that $(q',\bar x') \to (\D,\bar a)$. From Theorem~\ref{theo:char}, we have that 
$(q,\bar x)\to_k (q',\bar x')$. 
By composition and the fact that $(q',\bar x') \to (\D,\bar a)$ implies $(q',\bar x') \to_k (\D,\bar a)$, 
it follows that $(q,\bar x)\to_k (\D,\bar a)$ holds. \qed
\end{proof}

As a corollary to Theorem \ref{theo:games-overapp} and Proposition \ref{prop:games-poly} we obtain the following.

\begin{corollary} \label{coro:eval-overapp} 
Fix $k \geq 1$. 
Checking if $\bar a \in q'(\D)$, given a CQ $q$ that has a $\HW(k)$-overapproximation $q'$, 
a 
database $\D$, and a tuple $\bar a$
 in $\D$, can be solved in polynomial time by checking if $(q,\bar x) \to_{k} (\D,\bar a)$. 
 Moreover, this can be done without even computing $q'$. 
\end{corollary}

\subsection{More liberal $\HW(k)$-overapproximations} \label{sec:liberal} 
CQs may not have $\HW(k)$-overapproximations, for some $k \geq 1$. We observe in this section
that this anomaly can be taken care of by extending the language of queries 
over which overapproximations are to be found.

An {\em infinite} CQ is as a finite one, save that now the number of atoms is countably infinite. 
We assume that there are finitely many free variables in an infinite CQ. 
The evaluation of an infinite CQ $q(\bar x)$ over a database $\D$ is defined analogously to the 
evaluation of a finite one. 
Similarly, the generalized hypertreewidth of an infinite CQ is defined as in the finite case, but 
now tree decompositions can be infinite. 
We write $\HW(k)^\infty$ for the class of all CQs, finite and infinite ones, 
of generalized hypertreewidth at most $k$.
The next result states a crucial relationship between the existential $k$-cover game and 
 the class $\HW(k)^\infty$.

\begin{lemma} \label{lemma:games-infty} 
Fix $k\geq 1$. For every CQ $q$ there is a $q'(\bar x')$ in $\HW(k)^\infty$ such that for every database 
$\D$ and tuple $\bar a$ of constants in $\D$ it is the case that   
$$\bar a \in q'(\D)  \quad \Longleftrightarrow \quad (q',\bar x') \to (\D,\bar a) \quad \Longleftrightarrow \quad 
 (q,\bar x) \to_{k} (\D,\bar a).$$
This holds even for countably infinite databases $\D$. 
\end{lemma} 

\begin{proof}
The lemma follows directly from the proof of Lemma~\ref{lemma:tedious} by starting with $(T_\infty, \lambda_\infty)$ instead of $(T_c, \lambda_c)$, 
where $(T_\infty, \lambda_\infty)$ is the infinite rooted tree labeled with $k$-unions of $\D_q$ representing all possible moves of Spoiler on $\D_q$ in the 
existential $k$-cover game. All arguments are still valid when $\D$ is infinite. 
\qed
\OMIT{
Recall from Lemma \ref{lemma:tedious} that for every $c \geq 1$ 
there is a CQ $q_c(\bar x_c)$ that defines the existence of a winning strategy for Duplicator -- when Spoiler plays on $q$ -- up to 
$c$ rounds. Formally,  
  $q_c \to \D$ iff $q \to_k^c \D$ for every database $\D$. As it is clear from the proof of Lemma \ref{lemma:tedious}, 
  if $Q_c$ is the set of atoms in $q_c$, for $c \geq 0$, then $Q_c \subseteq Q_{c+1}$ for each $c \geq 0$. 
  We then define $q'$ as the existential closure of the countably infinite set of atoms 
  $\bigcup_{c \geq 0} Q_c$.  In fact, $q' \to \D$ iff for each $c \geq 0$ it is the case that $q_c \to \D$. 
  But the latter is equivalent to the fact that $q \to_k^c \D$ for each $c \geq 0$, which in turn means that $q \to_k \D$. 
  This holds even if $\D$ is countably infinite. \qed     }
\end{proof} 

Since we now deal with infinite CQs and databases, 
we cannot apply Proposition \ref{prop:games-tw} directly in our analysis of $\HW(k)^\infty$-overapproximations. 
Instead, we use the following suitable reformulation of it, which 
we obtain by inspection of its proof. 

\begin{proposition} \label{prop:games-tw-infty} Fix $k \geq 1$. Consider countably infinite 
databases $\D$ and $\D'$.
Then 
 $(\D,\bar a) \to_{k} (\D',\bar b)$ iff for each CQ $q(\bar x)$ in $\HW(k)^\infty$ we have that  
 if $(q,\bar x) \to (\D,\bar a)$ then $(q,\bar x) \to (\D',\bar b)$. 
\end{proposition}

\paragraph{{\em {\bf $\HW(k)^\infty$-overapproximations.}}}
 We  expand
the notion of overapproximation by allowing infinite CQs. 
Let $q' \in \HW(k)^\infty$. Then 
$q'$ is a $\HW(k)^\infty$-overapproximation of CQ $q$, if 
$q \subseteq q'$ and there is no $q''\in \HW(k)^\infty$ such that
 $q \subseteq q'' \subset q'$. (Here, $\subseteq$ is still defined with respect to
 finite databases only). 
In $\HW(k)^\infty$,  
we can provide each CQ $q$ an overapproximation. 

\begin{theorem} \label{theo:overapp-infty} Fix $k \geq 1$. 
For every CQ $q$ there is a CQ in 
$\HW(k)^\infty$ that is a $\HW(k)^\infty$-overapproximation of $q$.  
\end{theorem}  

\begin{proof}
We show that the CQ $q'(\bar x')$ -- as given in Lemma \ref{lemma:games-infty} -- is a 
$\HW(k)^\infty$-overapproximation of $q(\bar x)$. 
Notice first that $(q',\bar x') \to (q,\bar x)$ (by choosing $(\D, \bar a)$ as $(q, \bar x)$ in Lemma \ref{lemma:games-infty}), and hence  
$q \subseteq q'$. This is because  
the condition that $(q',\bar x') \to (q,\bar x)$ implies $q \subseteq q'$, 
expressed in Equation \eqref{eq:cont}, holds even for countably infinite CQs. 
We also have that $(q,\bar x) \to_{k} (q',\bar x')$ (by choosing $(\D, \bar a)$ as $(q', \bar x')$ in Lemma~\ref{lemma:games-infty}). 
Proposition \ref{prop:games-tw-infty} then tells us that for each $q''(\bar x'')$ in $\HW(k)^\infty$, 
if $(q'',\bar x'') \to (q,\bar x)$ then $(q'',\bar x'') \to (q',\bar x')$.
But then 
$q \subseteq q''$ implies $q' \subseteq q''$. This is because the 
condition that $q \subseteq q''$ implies $(q'',\bar x'') \to (q,\bar x)$,     
expressed in 
Equation \eqref{eq:cont}, continues to hold
as long as $q$ (but not necessarily $q'$) is finite. 
Thus, $q'$ is a $\HW(k)^\infty$-overapproximation of $q$. \qed
\end{proof} 

Despite the non-computable nature of $\HW(k)^\infty$-overapproximations, we get 
from Proposition \ref{prop:games-poly} and the proof of Theorem \ref{theo:overapp-infty} that 
they 
can be evaluated efficiently.  

\begin{corollary} \label{coro:eval-overapp-infty} 
Fix $k \geq 1$. 
Checking whether $\bar a \in q'(\D)$, given a CQ $q$ with $\HW(k)^\infty$-overapproximation $q'$,  
a 
database $\D$, and a tuple $\bar a$ in $\D$, boils down to checking if $(q,\bar x) \to_k (\D,\bar a)$, and 
thus it can be solved in polynomial time. 
\end{corollary}

Recall from Section \ref{preliminaries} that the ``approximation'' of the notion of homomorphism 
provided by the existential $k$-cover game suffices for evaluating those CQs $q(\bar x)$ whose core is in $\HW(k)$. That is, for every database $\D$ and tuple $\bar a$ of elements in $\D$, we have that 
$\bar a \in q(\D) \Leftrightarrow (q,\bar x) \to (\D,\bar a) \Leftrightarrow (q,\bar x) \to_k (\D,\bar a)$. 
For other CQs the existential $k$-cover game always provides an 
``overestimation" of the exact result. Interestingly, Corollary \ref{coro:eval-overapp-infty}  
establishes that such an 
overestimation is not completely arbitrary, as it is the one defined by the CQ in 
$\HW(k)^\infty$ that better approximates $q$ over the set of all databases.

\section{Existence of Overapproximations} \label{decidability}

CQs always have $\HW(k)^\infty$-overapproximations, but not necessarily 
finite ones. Here we study when a 
CQ $q$ has a finite overapproximation. We start with the case $k = 1$, 
which we show to be decidable in {\sc 2Exptime} (we do not know if this is optimal). 
For $k > 1$ we
leave the decidability open, but provide some explanation about where the difficulty lies.

\subsection{The acyclic case}  
\label{sec:exist-acyclic}

We start with the case of $\HW(1)$-overapproximations. Recall that $\HW(1)$ is an important class, as it consists precisely of 
the well-known {\em acyclic} CQs. Our main result is the following: 

\begin{theorem}\label{theo:decidability-acyclic} 
The following statements hold: 
\begin{enumerate}[label=(\alph*)]
\item There is a {\sc 2Exptime} algorithm that checks if a CQ $q$ has a 
$\HW(1)$-overapproximation and, if one exists, 
it computes one in triple-exponential time. 
\item If the maximum arity of the schema is fixed, there is an {\sc Exptime} algorithm that does this and 
computes a $\HW(1)$-overapproximation of $q$ in double-exponential time. 
\end{enumerate} 
\end{theorem} 

The general idea behind the proof is as follows. 
From a CQ $q$ we build 
a {\em two-way alternating tree automaton} \cite{CGKV88}, or 2ATA,   
$\mathcal A_q$, such that the language $L(\mathcal A_q)$ of trees accepted by $\A_q$ is nonempty if and only if 
$q$ has a $\HW(1)$-overapproximation. Intuitively, $\A_q$  
accepts those trees that encode a $\HW(1)$-overapproximation $q'$ of $q$. 
The emptiness problems for languages defined by 2ATAs can be solved in 
exponential time in the number of states \cite{CGKV88}.   
As our automaton $\mathcal A_q$ 
will have exponentially many states, its emptiness can be tested in double-exponential time. 
In addition, if the maximum arity of the schema is fixed, the number of states in $\mathcal A_q$ 
is polynomial, and hence emptiness can be tested in exponential time.

Before describing the details of the construction, let us shortly recapitulate two-way alternating tree automata. We closely follow the presentation from \cite{CGKV88}. The input of an 2ATA is a ranked tree over an alphabet $\Sigma$. In each computation step, a 2ATA is in one of finitely many states and visits a node $v$ of the input tree. Depending on the state and the label of the node, it can recursively start a bunch of processes; each of them starting from some state and either one of $v$'s neighbors or from $v$ itself. Whether the computation step is successful depends on a boolean combination of the outcomes of these processes.

Formally, a 2ATA is a tuple $(\Sigma, S, S_0, \delta, F)$ where $\Sigma$ is the tree alphabet, $S$ is a finite set of states, $S_0 \subseteq S$ is a set of initial states, $F \subseteq S$ is a set of accepting states, and $\delta$ is a transition function defined on $S \times \Sigma$ such that if $\sigma \in \Sigma$ has arity $\ell$ then $\delta(s, \sigma)$ is a propositional formula with variables from $S \times [\ell]$. Here, $[\ell]$ denotes the set $\{-1, 0, 1, \ldots, \ell\}$ of directions the automaton can take, where $-1$ denotes moving to the parent node, 0 denotes staying in the same node, and $j > 0$ denotes moving to the $j$th child. A proposition $(s', i) \in \delta(s, \sigma)$ represents that the automaton transitions into state $s'$ and moves to the node represented by $i$. As an example, if $\delta(s, \sigma) =  (r, 2) \wedge  \big((p, 1) \vee (q, -1)\big)$ then, when being in state $s$ and reading $\sigma$, the automaton starts two processes. One of them starting in state $r$ in the second child; the other starting in state $p$ in the first child or in state $q$ in the parent node. Thus, in particular, using the propositional formula of a transition, a 2ATA can universally or existentially choose a next state. A run of the 2ATA starts in the root of an input tree and in a state from $S_0$. Starting from there, a computation tree is spanned by applying the transitions. The input tree is accepted if the automaton has a computation tree whose leaves are accepting and the propositional formulas of transitions taken by the automaton are satisfied. We refer to~\cite{CGKV88} for the details of the semantics. 


We now show how the problem of existence of $\HW(1)$-overapproximations  can be reduced 
to the emptiness problem for 2ATA.

\begin{proposition} \label{prop:tree-automata}
There exists an algorithm that takes as input a CQ $q$  and returns a 2ATA $\A_q$ such that $q$ has a $\HW(1)$-overapproximation iff
\mbox{$L(\A_q) \neq \emptyset$}. Furthermore, the algorithm needs double-exponential time and $\A_q$ has exponentially many states. Furthermore: 
\begin{itemize} 
\item 
From every tree $T$ in $L(\A_q)$ one can construct in polynomial time a $\HW(1)$-overapproximation of $q$. 
\item 
If the maximum arity of the schema is fixed, then the algorithm needs exponential time and $\A_q$ has polynomially many states.
\end{itemize} 
\end{proposition} 

\begin{proof}
For simplicity we assume that $q$ is Boolean; towards the end of the proof we explain how the construction can be adapted for non-Boolean queries. 

Before describing the construction of $\A_q$, we explain how input trees for $\A_q$ encode CQs in $\HW(1)$. Suppose that the maximum arity of $q$ is $r$. Then an \emph{encoding} of a $\HW(1)$-tree decomposition is a tree whose nodes are labeled with (a) variables from the set $\{u_1, \ldots, u_{2r}\}$ and (b) atoms over these variables whose arity is at most $r$; the only condition being that all variables of a node are covered by one of the atoms. 

A CQ $q'$ from $\HW(1)$ with tree decomposition $(T_{q'}, \chi)$ of width one 
can be encoded as follows. Even though $q'$ can have unbounded many variables,  in each node of $T_{q'}$ at most $r$ variables appear, where $r$ is the maximum arity of an atom in~$q'$. 
Thus, by reusing variables, $(T_{q'}, \chi)$ can be encoded by using $2r$ variables: 
the encoding $\Enc(T_{q'}, \chi)$ of $(T_{q'}, \chi)$ is obtained by traversing the nodes of $(T_{q'}, \chi)$ in a top-down fashion. 
Fresh variables in a node $v$, i.e. variables not used by its parent node, are encoded by fresh variables from $\{u_1, \ldots, u_{2r}\}$. 
On the other hand, an encoding of a $\HW(1)$-tree decomposition can be decoded in a top-down manner into a $\HW(1)$-tree decomposition 
by assigning a fresh variable name to each \emph{new} occurrence of a variable $u_i$, that is, an occurrence of $u_i$ that does not appear in the parent node.  
Observe that decoding $\Enc(T_{q'}, \chi)$ yields the decomposition of a query identical to $q'$ up to renaming of variables. 


The 2ATA $\A_q$ needs to verify that the CQ $q'$ encoded by $T' = \Enc(T_{q'}, \chi)$ is a $\HW(1)$-overapproximation of $q$. By Theorem \ref{theo:char}, we need to check: (1) $q' \to_1 q$, 
and (2) $q \to_{1} q'$. The 2ATA $\A_q$ will be defined as the intersection of 2ATAs $\A_1$ and $\A_2$, that check 
conditions (1) and (2), respectively.  



Condition (1) is equivalent to $q' \to q$ (since $q' \in \HW(1)$). The 2ATA $\A_1$ can guess and verify a homomorphism from $q'$ to $q$. More precisely, it assumes that $T'$ is annotated by an intended homomorphism $h : q' \to q$, that is, each variable $x'$ in a node of $T'$ is annotated by a variable $x$ in $q$. The automaton then checks that this annotation encodes a homomorphism, i.e. that (a) all connected occurrences of $x'$ are annotated by the same variable of $q$, and (b) for each atom $R(\bar x')$ labeling a node of $T'$, the image of $R(\bar x')$ defined by the annotation  is in $q$. For (a), when processing a node $t'$ of $T'$, the automaton stores the partial homomorphism for the variables of $t'$ and tests that it is consistent with the partial homomorphism of each neighboring node of $t'$. In particular, $\A_1$ requires no alternation and  has at most exponentially many states. 
If the maximum arity of the schema is fixed then only polynomially many states are needed, as then the stored partial homomorphisms are over constantly many elements and thus can be stored in $O(\log |q|)$ bits.

We now describe how the automaton $\mathcal A_2$ works. 
First, as mentioned in Section~\ref{sec:compact}, $q \to_1 q'$ can be characterized as Duplicator having a 
compact winning strategy, which in turn can be characterized as follows~\cite{CD05}. 
Duplicator has a compact winning strategy on $q$ and $q'$, i.e. $q \to_1 q'$, iff there is a non-empty family ${\cal F}$ of partial homomorphisms from $q$ to $q'$ such that:
(a) the domain of each $f\in \cal F$ is a $1$-union of $q$, 
and (b) if $U$ and $U'$ are $1$-unions of~$q$, then each $f\in \cal F$ with domain $U$ can be {\em extended} to $U'$, 
i.e., there is $f'\in \cal F$ with domain $U'$ such that $f(x)=f'(x)$ for every $x\in U\cap U'$.

The 2ATA $\A_2$ assumes an annotation of $T' = \Enc(T_{q'}, \chi)$ that encodes the intended strategy $\cal F$. 
This annotation labels each node $t'$ of $T'$ by the set of partial mappings from 
$q$ to $q'$ whose domain is a $1$-union of $q$, and
whose range is contained in the variables from $\{u_1,\dots,u_{2r}\}$ labeling~$t'$.
It can be easily checked from the labelings of $T'$ if each mapping in this annotation is a partial homomorphism. 

To check condition (2), the 2ATA $\A_2$ makes a universal transition for each pair $(U,U')$ of $1$-unions and 
each partial mapping $g$ with domain $U$ annotating a node $t$ of $T'$.  
Then it checks the existence of a node $t'$ in $T'$ that is annotated with a mapping $g'$ that extends 
$g$ to $U'$.
The latter means that, for each $x\in U\cap U'$, both $g(x)$ and $g'(x)$ are the same variable of $q'$, that is, 
$g(x)$ and $g(x')$ are connected occurrences of the same variable in $\{u_1,\dots,u_{2r}\}$.
Thus to check the consistency of $g$ and $g'$, the automaton can store the variables in $\{g(x)\mid x\in U\cap U'\}$, 
and check that these are present in the label of each node guessed before reaching~$t'$. 
As this is a polynomial amount of information, $\A_2$ can be implemented using exponentially many states. 
Again, if the maximum arity of the schema is fixed then only polynomially many states suffice. 

  The construction above can be easily extended from Boolean to non-Boolean queries $(q, \bar x)$ and $(q', \bar x')$.  In this case, the encoding $T' = \Enc(T_{q'}, \chi)$ of $q'$ includes atoms that may contain free variables. The automaton $\mathcal A_1$ additionally checks that whenever a node in $T'$ is annotated by an atom $R(\bar y')$ then there is an atom $R(\bar y)$ in $q$ such that $h'(\bar y') = \bar y$ where $h'$ is the extension of the intended homomorphism $h$ that also maps $\bar x'$ to $\bar x$. The automaton $\mathcal A_2$ does an analogous check for the partial homomorphisms.\qed
\end{proof}

It is easy to see how Theorem \ref{theo:decidability-acyclic} follows from Proposition 
\ref{prop:tree-automata}. Checking if a CQ $q$ has a $\HW(1)$-overapproximation amounts to checking if 
$L(\A_q) \neq \emptyset$. The latter can be done in exponential time in the number of states of $\A_q$ \cite{CGKV88}, 
and thus in double-exponential time in the size of $q$. If $L(\A_q) \neq \emptyset$, one can construct a tree $T \in L(\A_q)$ 
in double-exponential time in the size of $\A_q$, and thus in triple-exponential time in the size of $q$. From $T$ one then gets 
in polynomial time (i.e., in {\sc 3EXPTIME} in the size of $q$) a $\HW(1)$-overapproximation of $q$. 

If the maximum arity is fixed, the 2ATA $\A_q$ has polynomially many states and, therefore, $L(\A_q) \neq \emptyset$ can be checked in single-exponential time. If $L(\A_q) \neq \emptyset$, one can then construct a tree $T \in L(\A_q)$ 
in double-exponential time in the size of $q$. From $T$ one then gets 
in polynomial time (i.e., in {\sc 2EXPTIME} in the size of $q$) a $\HW(1)$-overapproximation of $q$.

\subsection{The case of Boolean CQs over binary schemas}  
For Boolean CQs over schemas
of maximum arity two the existence and computation of $\HW(1)$-overapproximations 
can be solved in polynomial time. This is of practical importance since data models such as 
{\em graph databases} \cite{Bar13} and description logic {\em ABoxes} \cite{DL-handbook}
can be represented using schemas of this kind. It is worth noticing 
that in this context 
$\HW(1)$ coincides with the class of CQs of treewidth one \cite{DKV02}.  
 
 \begin{theorem}\label{theo:treewdith-one} 
 There is a {\sc Ptime}  
 algorithm that checks if a Boolean
CQ $q$ over a schema of maximum arity two 
has a $\HW(1)$-overapproximation $q'$, and computes such a $q'$ if it exists. 
 \end{theorem}

 We devote the rest of this section to prove Theorem~\ref{theo:treewdith-one}.  
 Let $q$ be a Boolean CQ. We define the \emph{Gaifman graph} $\G(q)$ of $q$ to be 
 the undirected graph whose nodes are the variables of $q$ and the edges are the pairs $\{z,z'\}$ 
 such that $z\neq z'$ and $z$ and $z'$ appear together in some atom of $q$. 
 A \emph{connected component} of $q$ is a Boolean CQ associated with a connected component $C=(V(C),E(C))$ of $\G(q)$, 
 i.e., one whose set of variables is $V(C)$ and contains precisely all the atoms in $q$ induced by variables in $V(C)$. 
 The Boolean CQ $q$ is \emph{connected} if it has only one connected component, that is, if $\G(q)$ is connected.

When the maximum arity is two, we have that a Boolean CQ $q$ is in $\HW(1)$ iff $\G(q)$ is an \emph{acyclic} (undirected) graph. 
In particular, if $q$ is connected then $q\in\HW(1)$ iff $\G(q)$ is a tree. 

To prove the theorem, we first show how the problem can be solved in polynomial time for connected 
Boolean CQs, 
and then explain how to reduce in polynomial time 
the problem for general Boolean CQs to connected ones.

\subsubsection{A polynomial time algorithm for connected Boolean CQs}
\label{sec:connected-case}


We start with the following observation:

\begin{lemma}
\label{claim:overconnected}
Let us assume that $q$ is a connected Boolean CQ that 
has an $\HW(1)$-overapproximation. 
Then it is the case that 
$q$ has a connected $\HW(1)$-overapproximation. 
\end{lemma}

\begin{proof}
Let $q'$ be an $\HW(1)$-overapproximation of $q$. 
Without loss of generality, we can assume that $q'$ is a core. 
By contradiction, suppose that $q'$ is not connected. 
By Theorem \ref{theo:char}, we have that $q\to_1 q'$ and $q'\to_1 q$. 
Since $q'\in \HW(1)$, the latter is equivalent to $q'\to q$. 
We claim that there is a connected component $q'_0$ of $q'$ such that $q\to_1 q'_0$. 
Recall from Section \ref{sec:compact} that $q\to_1 q'$ can be witnessed by a compact winning strategy $\H$ of the Duplicator.  
We make the Spoiler play in an arbitrary non-empty $1$-union $S_0$ of $q$.
Note that in this case, a $1$-union is either a singleton or an edge of $\G(q)$. 
The Duplicator can respond, following $\H$, with a partial homomorphism $h_0$ from $q$ to $q'$ with domain $S_0$. 
The elements in $h_0(S_0)$ must belong to some connected component of $q'$, say $q'_0$. 
Starting from this configuration of the game, 
we assume that Spoiler plays in a connected manner, i.e., if $S$ and $S'$ are two consecutive moves for Spoiler then $S\cap S'\neq \emptyset$. 
Then the Duplicator can play indefinitely by following $\H$. By the way Spoiler plays, all responses of Duplicator must fall in $q_0'$. 
Also, since $q$ is connected, each $1$-union of $q$ is eventually played by the Spoiler. 
This implies that $q\to_1 q'_0$. 

We have on the other hand that $q'_0\to q$, and hence $q'\equiv q'_0$. 
Since $q'$ is not connected, $q'_0$ has fewer atoms than $q'$. 
This contradicts the fact that $q'$ is a core. 
We conclude that $q'$ must be connected.  
\qed
\end{proof}

The high-level idea of the construction is to show that whenever a connected Boolean CQ $q$ has an $\HW(1)$-overapproximation $q'$, 
then we can assume that $q'$ is  
a ``subquery" of $q$, or of a slight modification of $q$. 
This will allow us to design a polynomial time algorithm that greedily looks for an $\HW(1)$-overapproximation of $q$.
Note that Lemma \ref{claim:overconnected} tells us that $q'\in \HW(1)$ can be assumed to be a connected Boolean CQ. 
In order to show that $q'$ is a ``subquery" of $q$, 
we first show a key lemma (see Lemma \ref{lemma:twohomomorphisms} below) 
about the structure of the \emph{endomorphisms} of a connected core in $\HW(1)$. 
In particular, we prove that besides the identity mapping, there can be only one extra endomorphism of a very particular form.  
Recall that an endomorphism is a homomorphism from the Boolean CQ to itself. 
For a core, any endomorphism $h$ is actually an \emph{isomorphism}, i.e., a bijection such that $h^{-1}$ is a homomorphism.

We first need to establish the following auxiliary lemma. 

\begin{lemma}\label{lemma:corehomomorphisms} 
Let $q$ be a connected Boolean CQ in $\HW(1)$ that is a core. Let $u$ and $v$ be variables of $q$. 
Then there is at most one endomorphism of $q$ mapping $u$ to $v$.
\end{lemma}

\begin{proof}
By contradiction, assume that there are two distinct endomorphisms $h_1$ and $h_2$ of $q$ with $h_1(u) = h_2(u) = v$. 
Recall that, since $q$ is a core, $h_1$ and $h_2$ are isomorphisms. 
By hypothesis, $\G(q)$ is a tree and we root it at $u$. 
Since $h_1\neq h_2$, there must be a variable $w\neq u$ of $q$ with $h_1(w) \neq h_2(w)$, 
which we choose to have minimal distance to $u$ in the tree $\G(q)$. 
As $h_2$ is a bijection, there is a unique $w'$ such that $h_2(w')=h_1(w)$. 
Note that $w'\neq w$. 
We claim that $w$ and $w'$ have the same parent in $\G(q)$. 
Let $z$ be the parent of $w$. 
Since $h_1$ is an isomorphism, $h_1(w)$ and $h_1(z)$ are adjacent in $\G(q)$. 
As $h_1(w)=h_2(w')$ and, by minimality of $w$, $h_1(z)=h_2(z)$, we have that $h_2(w')$ and $h_2(z)$ are also adjacent. 
Since $h_2$ is an isomorphism, it follows that $w'$ and $z$ are adjacent in $\G(q)$. 
However, $w'$ cannot be the parent of $z$ since $h_1(w')\neq h_2(w')$. 
Therefore $z$ is the parent of $w'$.

We define a mapping $h$ from $q$ to itself 
such that $h(t)=h_2(t)$ if $t$ is a variable that 
belongs to the subtree of $\G(q)$ rooted at $w'$; 
otherwise, $h(t)=h_1(t)$. 
Note that $h$ is an endomorphism of $q$. 
Indeed, the only atoms that in principle are not satisfied by $h$ are those mentioning $w'$ and $z$. 
However, since $h_2$ is an endomorphism and $h_1(z)=h_2(z)$, these atoms are actually satisfied. 
Finally, observe that
$$h(w) \, = \, h_1(w) \, = \, h_2(w') \, = \, h(w'),$$ and then 
$h$ is not injective. This is a contradiction to $q$ being a core.
\qed
\end{proof}

Let $q\in \HW(1)$ be a Boolean connected core. 
We say that an endomorphism $h$ of $q$ is a \emph{swapping endomorphism} for $u$ and $v$ if 
$h(u)=v$ and $\{u,v\}$ is an edge in $\G(q)$, i.e., $u$ and $v$ are adjacent variables in $\G(q)$. 
Using the fact that $\G(q)$ is a tree and $h$ is an isomorphism, 
it can be seen that $h(v)=u$ must also hold in this case (and hence the name ``swapping"). 
Moreover, if such $h$ exists for $u$ and $v$, by Lemma~\ref{lemma:corehomomorphisms}, it must be unique and then 
we can speak about \emph{the} swapping endomorphism for $u$ and $v$. We then have the following:

\begin{lemma}\label{lemma:twohomomorphisms}
Let $q$ be a Boolean connected CQ in $\HW(1)$ that is a core. 
Then $q$ has at most one endomorphism besides the identity mapping. 
If such endomorphism exists, it is the swapping endomorphism for some $u$ and $v$.
\end{lemma}

\begin{proof}
Fix a simple path $P=w_0,w_1,\dots,w_m$ in $\G(q)$ of maximal length (recall that in a simple path all vertices are distinct). 
Suppose that there exists an endomorphism $h$ of $q$ different from the identity. 
As $q$ is a core, $h$ is an isomorphism. 
Then the path $P' = h(w_0),h(w_1),\dots,h(w_m)$ is a simple path of the same length. 
Furthermore, $P$ and $P'$ share a vertex. 
Indeed, if this is not the case, since $\G(q)$ is connected, 
one can pick $w$ in $P$ and $w'$ in $P'$ such that $w$ and $w'$ are connected by a simple path $P''$ which is 
vertex-disjoint  from $P$ and $P'$ (except for $w$ and $w'$), and construct a longer path than $P$.

We claim that $m$ is odd (recall that $m$ is the size of $P$). 
By contradiction, suppose $m$ is even and let $u = w_{m/2}$ be the middle vertex of $P$. 
It must be the case that $u$ is also the middle vertex of $P'$, as otherwise $\G(q)$ would contain a path longer than $P$. 
In particular, $h(u) = u$ and then $h$ must be the identity mapping by Lemma \ref{lemma:corehomomorphisms}; a contradiction. 

Let $u = w_{\lfloor m/2 \rfloor}$ and $v = w_{\lceil m/2 \rceil}$ be the middle vertices of $P$. 
Again, by maximality of $P$, we have that $u$ and $v$ are also the middle vertices of $P'$, i.e., $\{u,v\}=\{h(u), h(v)\}$. 
Since $h$ is not the identity and by Lemma \ref{lemma:corehomomorphisms}, it follows that $h(u)=v$ (and hence, $h(v)=u$). 
Since $u$ and $v$ are adjacent in $\G(q)$, we conclude that $h$ must be the swapping endomorphism for $u$ and~$v$.
\qed
\end{proof}

Let $q\in \HW(1)$ be a connected Boolean CQ. 
Let $u$ and $v$ be variables adjacent in $\G(q)$.
Since $\G(q)$ is a tree, 
if we remove from $q$ all the atoms 
that mention $u$ and $v$ simultaneously, we obtain two connected Boolean CQs, one containing $u$ and the 
other containing $v$. We denote these CQs by $t_u^q$ and $t_v^q$, respectively. 

We need to introduce some notation. 
Suppose that $q$ and $q'$ are Boolean CQs and $X$ and $X'$ are subsets of the variables of $q$ and $q'$, respectively. 
We denote by $(q,X)\rightarrow_1 (q',X')$ the fact that 
the Duplicator has a winning strategy in the existential $1$-cover game on $q$ and $q'$ with the property that whenever 
the Spoiler places a pebble on an element of $X$ in $q$, 
then the Duplicator responds with some element of $X'$ in $q'$. 
It can be seen that checking whether $(q,X)\rightarrow_1 (q',X')$ can still be done in polynomial time. 

The following lemma formalizes the idea that an $\HW(1)$-overapproximation can be assumed to be essentially a subquery 
of the original query.  

 \begin{lemma}\label{lemma-games}
 Suppose $q$ is a Boolean CQ and suppose $q'$ is a connected core that is a $\HW(1)$-overapproximation of $q$. 
 Then we have the following:
 \begin{enumerate}
 \item If the only endomorphism of $q'$ is the identity mapping, then any homomorphism from $q'$ to $q$ is injective. 
 \item If $q'$ has the swapping endomorphism for some $u'$ and $v'$, then for any homomorphism $h$ from $q'$ to $q$, we have that
 \begin{enumerate}
 \item $(q,\{h(u'),h(v')\})\rightarrow_1 (q',\{u',v'\})$, and
 \item $h(z')\neq h(z'')$ for all pairs of variables $z',z''$, except maybe for $z'\neq u'$ in $t_{u'}^{q'}$ and 
 $z''\neq v'$ in $t_{v'}^{q'}$. 
 \end{enumerate}
 \end{enumerate}
  \end{lemma}

 \begin{proof}
 Suppose the only endomorphism of $q'$ is the identity and let $h$ be a homomorphism from $q'$ to $q$. 
 Towards a contradiction, suppose $h(z')=h(z'')=b$ for distinct variables $z'$ and $z''$ in $q'$. 
 Let $\H$ be a winning strategy of Duplicator witnessing $q\rightarrow_1 q'$. 
 We choose any variable $b'$ in $q'$ such that $b'$ is a possible response of Duplicator according to $\H$, 
  when Spoiler starts playing on $b$ in $q$.
Suppose that $b'\neq z'$ (the case $b'\neq z''$ is analogous). 
Then by composing $h$ with $\H$, 
we obtain a winning strategy for Duplicator in the game on $q'$ and $q'$ such that $b'$ is a possible response of Duplicator
 when Spoiler starts playing on $z'$. 
 Since $q'\in \HW(1)$, we can define an endomorphism $g$ of $q'$ that maps $z'$ to $b'$. 
 Then $g$ is an endomorphism different from the identity, which is a contradiction.
 
 Suppose now that $q'$ has the swapping endomorphism for some $u'$ and $v'$, and let $h$ be a homomorphism from $q'$ to $q$.
 First, assume by contradiction that Duplicator's strategy witnessing $q\rightarrow_1 q'$ is such that 
 for $h(u')$ (the case for $h(v')$ is analogous), Duplicator responds with $z'\not\in \{u',v'\}$. By composing $h$ 
 with this strategy, and using the fact that $q'\in \HW(1)$, 
 it follows that there is an endomorphism $g$ of $q'$ that maps $u'$ to $z'$. This endomorphism is different 
 from the identity and from the swapping endomorphism for $u'$ and $v'$, 
 which contradicts Lemma \ref{lemma:twohomomorphisms}. 
 Finally, suppose towards a contradiction that item (2.b) does not hold for some pair $z'\neq z''$ with $h(z')=h(z'')$. 
We have two cases. 
 \begin{itemize} 
 \item[(i)] $z',z''\in t^{q'}_{u'}$ or $z',z''\in t^{q'}_{v'}$; or 
 \item[(ii)] $z'\in t^{q'}_{u'}$ and $z''\in t^{q'}_{v'}$, and either $z'=u'$ or $z''=v'$. 
 \end{itemize} 
In any case, we can again compose $h$ with the strategy witnessing $q\rightarrow_1 q'$ and use the fact that $q'\in \HW(1)$, 
to derive an endomorphism of $q'$ that is neither the identity nor the swapping endomorphism for $u'$ and~$v'$, which is a contradiction. \qed
\end{proof}

Let $q$ be a Boolean CQ and $u$ and $v$ be adjacent variables in $\G(q)$. 
We define a CQ $q_u \# q_v$ as follows. 
Denote by $q \setminus v$ the CQ obtained from $q$ by removing all atoms that contain $v$. 
Let $q_u$ be the query constructed from $q \setminus v$ by replacing each variable $z$ in $q \setminus v$ by a fresh variable $z_u$. 
Similarly, let $q_v$ be the CQ where each variable $z$ in $q \setminus u$ is replaced by a fresh variable $z_v$. 
The CQ $q_u \# q_v$ contains all the atoms of $q_u$ and $q_v$, 
and additionally, all atoms $R(u_u, v_v)$ or $R(v_v, u_u)$ whenever $R(u,v)$ or $R(v,u)$ is an atom in $q$, respectively. 
Note that by mapping variables $z_u$ and $z_v$ to $z$, we have that $q_u \# q_v\to q$.

Now we are ready to present our algorithm. 
Observe that Lemma \ref{lemma-games} implies that whenever $q'$ is an $\HW(1)$-overapproximation of $q$, 
we can assume that either $q'$ is a subquery of $q$ (item (1)), or a subquery of $q_u \# q_v$ for some $u$ and $v$ (item (2)). 
The algorithm then greedily searches through the subqueries of $q$ and $q_u \# q_v$, for all $u$ and $v$, 
to find an $\HW(1)$-overapproximation of $q$.

\paragraph{{\em {\bf The algorithm.}}}
Let $q$ be a connected Boolean CQ.
The algorithm first checks whether a subquery of $q$ is an $\HW(1)$-overapproximation. This is Step $1$. 
 In Step $2$, the algorithm checks whether a subquery of $q_u \# q_v$ is an $\HW(1)$-overapproximation, for some $u$ and $v$ in $q$. 
If neither step succeed then the algorithm rejects.
For a Boolean CQ $p$ and an atom $e$ of $p$, we denote by $p\setminus e$ the Boolean CQ obtained from $p$ by removing $e$. 
Step $1$ is as follows:

  \begin{enumerate}
   \item Set $q_0$ to be $q$. 
   \item While $q_i \notin \HW(1)$, search for an atom $e$ such that \mbox{$q_i \rightarrow_1 q_i \setminus e$}. 
   If there is no such atom then continue with Step $2$. Otherwise, set $q_{i+1}$ to be~$q_i\setminus e$.
   \item If $q_i \in \HW(1)$, for some $i$, then accept and output $q_i$.
     \end{enumerate}

For Step $2$, let $\cal P$ be an enumeration of the pairs $(u,v)$ such that $u,v$ are adjacent in $\G(q)$ and $q\rightarrow_1 q_u\# q_v$. 
Step~$2$ is as follows:

  \begin{enumerate}
  \item Let $(u,v)$ be the first pair in $\cal P$.
   \item Set $q_0$ to be $q_{u}\# q_v$. 
   \item While $q_i \notin \HW(1)$, search for an atom $e$ that does not mention 
   $u_u$ and $v_v$ simultaneously such that $(q_i, \{u_u,v_v\}) \rightarrow_1 (q_i \setminus e, \{u_u,v_v\})$. 
   If there is no such atom, let $(u,v)$ be the next pair in $\cal P$ and repeat from item $2$. 
   Otherwise, set $q_{i+1}$ to be $q_i\setminus e$.
   \item If $q_i \in \HW(1)$, for some $i$, then accept and output $q_i$.
  \end{enumerate}
  
\medskip
  
Notice that the described algorithm can be implemented in polynomial time. 
Below we argue that it is correct. 

Suppose first that the algorithm, on input $q$, accepts with output $q^*$. By construction~\mbox{$q^*\in \HW(1)$}. 
 Assume first that the algorithm accepts in the $m$-th iteration of Step~$1$, and thus $q^*=q_m$. 
 By construction, for each $0\leq i<m$, we have that $q_i\rightarrow_1 q_{i+1}$ and~\mbox{$q_{i+1}\rightarrow_1 q_i$}.
 In particular, $q\rightarrow_1 q^*$ and $q^*\rightarrow_1 q$, and thus $q^*$ is a $\HW(1)$-overapproximation of $q$.
 Suppose now that the algorithm accepts in Step~$2$ for a pair $(u,v)\in \cal P$, in the $m$-th iteration. 
 Again we have that $q_i\rightarrow_1 q_{i+1}$ and $q_{i+1}\rightarrow_1 q_i$, for each $0\leq i<m$, and thus 
 $$\text{$q_u\#q_v\rightarrow_1q^* \quad$ and $\quad q^*\rightarrow_1 q_u\#q_v$.}$$ Since $(u,v)\in \cal P$, it follows that $q\rightarrow_1 q_u\#q_v$, and 
 then $q\rightarrow_1 q^*$. Using the fact that $q_u\#q_v\to q$, we have that $q^*\rightarrow_1 q$. 
 Hence, $q^*$ is a $\HW(1)$-overapproximation~of~$q$.

 It remains to show that if $q$ has an $\HW(1)$-overapproximation $q'$ then the algorithm accepts. 
 Since $q$ is connected, by Lemma \ref{claim:overconnected}, we can assume that $q'$ also is. Moreover, we can assume without loss of generality that $q'$ is a core.
 By Lemma \ref{lemma:twohomomorphisms}, we have two cases:
 \begin{itemize}
 \item[(1)] the only endomorphism of $q'$ is the identity, or 
 \item[(2)] $q'$ has two endomorphisms, namely, the identity and the swapping endomorphism for some variables $u'$ and $v'$.
\end{itemize} 

First suppose case (1) applies. We show that the algorithm accepts in Step~$1$. 
By definition, $q_{i}\rightarrow_1 q_{i+1}$ and $q_{i+1}\rightarrow_1 q_i$ (actually $q_{i+1}\to q_i$), 
for each $0\leq i\leq m$, where $m$ is the number of iteration in Step~$1$. 
It follows that $$\text{$q_0=q\rightarrow_1 q_m \quad $ and $\quad q_m\rightarrow_1 q$.}$$ 
Since the relation $\rightarrow_1$ composes, $q'$ is a $\HW(1)$-overapproximation of $q_m$ 
and by using Lemma \ref{lemma-games}, $q'$ is a subquery of $q_m$.
Now for the sake of contradiction assume that the algorithm does not accept in Step~$1$. 
Then $q_m\not\in \HW(1)$ and there is no edge $e$ in $q_m$ such that $q_m\rightarrow_1 q_m\setminus e$. 
Since $q'$ is $\HW(1)$-overapproximation of $q_m$, we have that $q_m\rightarrow_1 q'$ and, since $q'\in \HW(1)$, 
$q'$ is a proper subquery of $q_m$. 
It follows that there is an edge $e$ in $q_m$ such that $q_m\rightarrow_1 q_m\setminus e$, which is a contradiction.

Suppose case (2) holds. In this case the algorithm accepts in Step~$2$. 
Let $h$ be a homomorphism from $q'$ to $q$, and let $u=h(u')$ and $v=h(v')$.
By Lemma \ref{lemma-games}, $u\neq v$ and then $u$ and $v$ are adjacent. 
Also, by Lemma \ref{lemma-games}, $q'$ is a subquery of $q_u\#q_v$. 
Since $q\rightarrow_1 q'$, it follows that $q\rightarrow_1 q_u\#q_v$, and then $(u,v)\in \cal P$. 
We claim that the algorithm accepts when $(u,v)$ is chosen from $\cal P$.
First, note that $q'$ is a $\HW(1)$-overapproximation of $q_m$.
Indeed, by definition, 
$$q_m\to q_u\#q_v, \quad q_u\#q_v\to q, \quad  
\text{and} \quad q\rightarrow_1 q'.$$ It follows that $q_m\rightarrow_1 q'$. On the other hand, 
we have that $$(q',(u',v'))\to (q_u\#q_v, (u_u,v_v)),$$ as $q'$ is a subquery of $q_u\#q_v$, 
and 
$$(q_u\#q_v, \{u_u,v_v\})\rightarrow_1 (q_m,\{u_u,v_v\}).$$  
It follows that $(q',\{u',v'\})\rightarrow_1 (q_m,\{u_u,v_v\})$, 
which implies that $(q',(u',v'))\to (q_m,(u_u,v_v))$ via a homomorphism $g$. 
Then $q'$ is a $\HW(1)$-overapproximation of $q_m$. 
By applying Lemma \ref{lemma-games} to $q_m,q'$ and $g$, 
we obtain that $$(q_m,\{u_u,v_v\})\rightarrow_1 (q',\{u',v'\}),$$ and 
$g$ satisfies item (2.b). Observe that $g(z')\neq g(z'')$ for all $z'\neq u'$ in $t_{u'}^{q'}$ and 
 $z''\neq v'$ in $t_{v'}^{q'}$, since $\{u_u,v_v\}$ is a \emph{bridge} of $\G(q_m)$, i.e., its removal disconnects $\G(q_m)$.
 We conclude that $g$ is injective and then $q'$ is a subquery of $q_m$.

Towards a contradiction, assume that the algorithm do not accept when $(u,v)$ is chosen from $\cal P$. 
Then $q_m\not\in \HW(1)$ and there is no edge $e$ that does not mention both $u_u,v_v$ 
such that $(q_m, \{u_u,v_v\})\rightarrow_1 (q_m\setminus e, \{u_u,v_v\})$. 
Since $$\text{$(q_m, \{u_u,v_v\})\rightarrow_1 (q',\{u',v'\}) \quad$ \text{and}  
$\quad (q', (u',v'))\to (q_m, (u_u,v_v))$}$$ via the injective homomorphism $g$ and $q'\in \HW(1)$,  
it follows that there is an edge $e$ that does not mention both $u_u,v_v$ 
such that $(q_m, \{u_u,v_v\})\rightarrow_1 (q_m\setminus e, \{u_u,v_v\})$. 
This is a contradiction.

\subsubsection{Reduction to the connected case}

Now we consider the non-connected case. 
Given a Boolean CQ $q$ with connected components $q_1,\dots,q_m$, the algorithm proceeds as follows:

\begin{enumerate}
\item Start by simplifying $q$: Compute a minimal subset of Boolean CQs $\cal Q$ in $\{q_1,\dots,q_m\}$ such that 
for each $1\leq i\leq m$, there is a $p\in \cal Q$ with $q_i\rightarrow_1 p$.
\item Check whether each $p\in \cal Q$ has an $\HW(1)$-overapproximation $p'$ using the algorithm described for connected Boolean CQs in Section \ref{sec:connected-case}. 
If this is the case then accept and output the disjoint conjunction $\bigwedge_{p\in \cal Q} p'$.
\end{enumerate}

Clearly, the algorithm can be implemented in polynomial time. For the correctness, 
suppose first that the algorithm accepts and outputs $q'=\bigwedge_{p\in \cal Q} p'$. 
Then $$q' \, \to \, \bigwedge_{p\in \cal Q}p \, \to \, q.$$ 
We also have that $q\rightarrow_1  \bigwedge_{p\in \cal Q}p$ (by definition of $\cal Q$), 
and $\bigwedge_{p\in \cal Q}p\rightarrow_1 q'$. 
This implies that $q'$ is a $\HW(1)$-overapproximation of $q$. 

Suppose now that $q$ has an $\HW(1)$-overapproximation $q'$. 
Since both $q\rightarrow_1 \bigwedge_{p\in \cal Q}p$ and $\bigwedge_{p\in \cal Q}p\rightarrow_1 q$ hold, 
it is the case that $q'$ is also an $\HW(1)$-overapproximation of $\bigwedge_{p\in \cal Q}p$.
By the minimality of $\cal Q$, 
we have that $p\not\rightarrow_1 \hat p$, for each pair of distinct CQs $p,\hat p\in \cal Q$. 
Let $p$ be a CQ in $\cal Q$. Since $p\rightarrow_1 q'$ and $p$ is connected, it follows that $p\rightarrow_1 p^*$, 
where $p^*$ is a connected component of $q'$ 
(note that this follows by using the same argument we use in the proof of Lemma \ref{claim:overconnected} to show that $q\to_1 q'_0$, for some component $q'_0$). 
Also, since $q'\to \bigwedge_{p\in \cal Q}p$, 
there is $p_0\in \cal Q$ such that $p^*\to p_0$. In particular,~\mbox{$p\rightarrow_1 p_0$}. 
It follows that $p_0=p$, and then $p^*$ is an $\HW(1)$-overapproximation of~$p$.
We conclude that each $p\in \cal Q$ has an $\HW(1)$-overpproximation, and thus the algorithm accepts.
This finishes the proof of Theorem \ref{theo:treewdith-one}.

\paragraph{Remark}
The restriction to Boolean CQs in Theorem \ref{theo:treewdith-one} 
 is used in an essential way in our proof, and it is not clear whether this extends to non-Boolean CQs. 
The issue is that we do not have an analog of Lemma \ref{claim:overconnected} for the non-Boolean case, and hence, it is not clear how to translate Lemma \ref{lemma:twohomomorphisms} and Lemma \ref{lemma-games}
(or some modifications thereof) into a polynomial time algorithm. 

\subsection{Size of overapproximations} 
Over binary schemas $\HW(1)$-overapproximations are 
of polynomial size. This is optimal as over schemas of arity three there is an 
exponential lower bound for the size of $\HW(1)$-overapproximations:  

\begin{proposition} \label{prop:size} 
There is a schema $\sigma$ with a single ternary relation symbol
 and a family $(q_n)_{n \geq 1}$ of Boolean CQs over~$\sigma$, such that 
 (1) $q_n$ is of size $O(n)$, and (2) the size of every $\HW(1)$-overapproximation of $q_n$ 
 is $\Omega(2^n)$.  
 \end{proposition} 
 
 
  \begin{proof}
  The CQ $q_n$ contains the atoms:  
  \begin{itemize} 
  \item 
  $R(x_0, x^1_{1}, x^2_{1})$, $R(x_0, x^2_{1}, x^1_{1})$, as well as 
  \item $R(x^j_i, x^1_{i+1}, x^2_{i+1})$ and $R (x^j_i, x^2_{i+1}, x^1_{i+1})$, for each $1 \leq i \leq n-1$ and $j \in \{1, 2\}$.
     \end{itemize}
   Consider now the CQ $q'_n$ with the atoms:  
     \begin{itemize}
       \item $R(y_0, y^1_{1}, y^2_{1})$, $R(y_0, y^2_{1}, y^1_{1})$, and
       \item $R(y^w_{|w|}, y^{w1}_{|w|+1}, y^{w2}_{|w|+1})$ and $R(y^w_{|w|}, y^{w2}_{|w|+1}, y^{w1}_{|w|+1})$, for each word $w$ over  $\{1, 2\}$ of length $1 \leq |w| \leq n-1$.
     \end{itemize}
Figure \ref{figure:exampleAcyclicBlowupA} depicts the CQs $q_3$ and $q'_3$.
  
  \begin{figure}
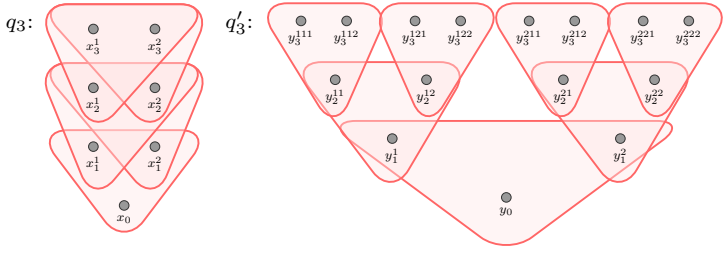

  \centering  
    $q_3$:\hspace{-1mm}\raisebox{-0.9\height}{\centering \scalebox{0.6}{
      \picExampleAcyclicBlowupA
    }}
    $q'_3$:\hspace{-2mm}\raisebox{-0.9\height}{ \scalebox{0.6}{
      \picExampleAcyclicBlowupB
    }}
    \vspace{-3mm}
    \caption{Illustration of the CQs $q_3$ and $q'_3$ from Proposition \ref{prop:size}. Each triple of variables represents 
    two atoms in the query; e.g., $\{y_0,y_1^1,y_1^2\}$ represents atoms $R(y_0, y^1_{1}, y^2_{1})$ and $R(y_0, y^2_{1}, y^1_{1})$ in $q'_3$.}
    \label{figure:exampleAcyclicBlowupA}
  \end{figure}

  Clearly, the mapping $h: q'_n \rightarrow q_n$ defined as 
  $h(y_0) = x_0$ and $$h(y^{wj}_{|w|+1}) \, = \, x^j_{|w|+1},$$ for each word $w$ over $\{1, 2\}$ of length $0 \leq 
  |w| \leq n-1$ and $j \in \{1,2\}$, 
  is a homomorphism. 
  We now show that $q_n\rightarrow_1 q'_n$ by building a compact winning strategy for Duplicator 
  (see the proof of Proposition \ref{prop:tree-automata}) which 
  basically ``inverts'' the homomorphism $h$. It contains: 
(a) Partial homomorphisms $(x_0, x^1_{1}, x^2_{1}) \to (y_0, y^1_{1}, y^2_{1})$ and 
    $(x_0, x^1_{1}, x^2_{1}) \to (y_0, y^2_{1}, y^1_{1})$, and (b)   
   for each word $w$ over $\{1, 2\}$ of length $1\leq |w| \leq n-1$ and $j\in \{1,2\}$, the 
    partial homomorphisms: 
    \begin{itemize}
    \item $(x^j_{|w|}, x^1_{|w|+1}, x^2_{|w|+1}) \to (y^w_{|w|}, y^{w1}_{|w|+1}, y^{w2}_{|w|+1})$, and
    \item $(x^j_{|w|}, x^1_{|w|+1}, x^2_{|w|+1}) \to (y^w_{|w|}, y^{w2}_{|w|+1}, y^{w1}_{|w|+1})$.    
  \end{itemize}
  It can be seen that this is a winning strategy for Duplicator.
  
Observe that the size of $q'_n$ is $\Omega(2^n)$. 
A straightforward case-by-case analysis shows that $q'_n$ is a core, 
i.e., there is no homomorphism from $q_n'$ to a {\em proper} subset of its atoms. 
We claim that $q'_n$ is the smallest $\HW(1)$-overapproximation of $q_n$, from which the proposition follows. 
Assume, towards a contradiction, that $q'$ is a $\HW(1)$-overapproximation of $q_n$ with fewer atoms than $q'_n$. Then, by Corollary \ref{coro:over-unique}, 
we have that $q'_n\equiv q'$. Composing the homomorphisms $h_1: q'_n \to q'$ and $h_2: q' \to q'_n$ yields a homomorphism from $q'_n$ to a proper subset of the atoms of $q'_n$. This is a contradiction since $q'_n$ is a core.
\qed
\end{proof}

\subsection{Beyond acyclicity} 
\label{sec:beyond-acyclic}

Theorem \ref{theo:char} characterizes when a CQ has a 
$\HW(k)$-overapproximation.  We provide an alternative  characterization 
in terms of a {\em boundedness} condition for the existential cover game. 
This helps understanding where lies the difficulty of determining the 
decidability status of the problem of existence of $\HW(k)$-overapproximations, for $k > 1$. 

Recall that we write $(\D,\bar a) \to_{k}^c (\D',\bar b)$, for $k \geq 1$ and $c \geq 1$, if Duplicator 
has a winning strategy {\em in the first $c$ rounds} of the (compact) existential $k$-cover 
game on $(\D,\bar a)$ and $(\D',\bar b)$. 
The next result establishes that a CQ $q$ has a $\HW(k)$-overapproximation iff the existential $k$-cover 
game played from $q$ is ``bounded'', i.e., if there is a constant $c \geq 1$ that bounds the number of rounds this game
needs to be played in order to determine if
Duplicator wins. 

\begin{theorem}  \label{theo:boundedness} Fix $k \geq 1$. The CQ $q(\bar x)$ has a $\HW(k)$-overapproximation iff 
there is an integer $c \geq 1$ such that  
$$(q,\bar x) \to_{k} (\D,\bar a) \ \ \Longleftrightarrow \ \ (q,\bar x) \to_{k}^c (\D,\bar a),$$ 
for each database $\D$ and tuple $\bar a$ of elements in $\D$. 
\end{theorem} 

\begin{proof} 
Assume first that there is an integer $c \geq 1$ such that  
$(q,\bar x) \to_{k} (\D,\bar a)$ iff $(q,\bar x) \to_{k}^c (\D,\bar a)$, for each database $\D$ and tuple $\bar a \in  \D$. 
Therefore, 
$(q,\bar x) \to_{k} (\D,\bar a)$ iff $(q_c,\bar x_c) \to (\D,\bar a)$,
for each database $\D$ and $\bar a \in  \D$, 
where $q_c(\bar x_c)$ is the CQ in $\HW(k)$ which is defined 
in Lemma \ref{lemma:tedious}. 
But then, since $q_c \in \HW(k)$, we obtain from Equation \eqref{eq:tw} 
that for each database $\D$ and tuple $\bar a \in  \D$: 
\begin{equation} \label{eq:bound} 
(q,\bar x) \to_{k} (\D,\bar a) \ \ \Longleftrightarrow \ \ (q_c,\bar x_c) \to_k (\D,\bar a).
\end{equation}  
We show next that $q_c$ is the $\HW(k)$-overapproximation of 
$q$. From Theorem \ref{theo:char}, we need to show that $(q_c,\bar x_c) \to_k (q,\bar x)$ and $(q,\bar x) \to_k (q_c,\bar x_c)$.  
Both conditions follow from Equation \eqref{eq:bound}: the former by choosing $(\D,\bar a)$ to be $(q,\bar x)$, 
and the latter by choosing $(\D,\bar a)$ to be $(q_c,\bar x_c)$. 
%
 
Assume on the other hand that $q(\bar x)$ has a $\HW(k)$-overapproximation $q^*(\bar x^*)$.
From Theorem \ref{theo:games-overapp} we obtain that for every (finite) 
database $\D$ and tuple $\bar a$ of elements in $\D$: 
 $$(q^*,\bar x^*) \to (\D,\bar a) \ \ \Longleftrightarrow \ \ (q,\bar x) \to_{k} (\D,\bar a).$$ 
As a first step we prove that this continues to hold for countably infinite databases.
We do so by refining the proof of Theorem \ref{theo:games-overapp}. 

Let $\D$ be a countably infinite database. Assume first that $(q,\bar x) \rightarrow_{k} (\D,\bar a)$. 
Due to the fact that $q^*$ is a $\HW(k)$-overapproximation of $q$, we have that 
$(q^*,\bar x^*) \to (q,\bar x)$. 
Proposition \ref{prop:games-tw-infty} then implies 
that $(q^*,\bar x^*) \to (\D,\bar a)$ since $q^*$ is in $\HW(k)$.   
Assume, on the other hand, that $(q^*,\bar x^*) \to (\D,\bar a)$. 
Suppose, for the sake of contradiction, that 
$(q,\bar x) \not\to_{k} (\D,\bar a)$. Proposition \ref{prop:games-tw-infty} then establishes that 
there is a CQ $q^{\#}(\bar x^{\#})$ in 
$\HW(k)^\infty$ such that 
\begin{equation} 
\label{eq:forcont} 
\text{$(q^{\#},\bar x^{\#}) \to (q,\bar x)$, \ but $ \ (q^{\#},\bar x^{\#}) \not\to (\D,\bar a)$.}
\end{equation} 
The first fact in Equation \eqref{eq:forcont} implies that $q \subseteq q^{\#}$, as this direction of Equation \ref{eq:cont} 
continues being true as long as the CQ $q$ in the left-hand side of the containment relation is finite. Therefore,  
$(q^{\#},\bar x^{\#}) \to (q^*,\bar x^*)$ 
since Corollary \ref{coro:over-unique} continues being true for CQs in $\HW(k)^\infty$. 
We conclude that $(q^{\#},\bar x^{\#}) \to (\D,\bar a)$ since $(q^*,\bar x^*) \to (\D,\bar a)$. 
This contradicts the second fact in Equation \eqref{eq:forcont}. 

Therefore, if $q'(\bar x')$ is the CQ given by Lemma \ref{lemma:games-infty} for CQ $q(\bar x)$, 
then for every countable database $\D$ and tuple $\bar a$ in $\D$ it is the case that 
$$(q',\bar x') \to (\D,\bar a) \ \ \Longleftrightarrow \ \ (q,\bar x) \to_{k} (\D,\bar a) \ \ \Longleftrightarrow \ \ (q^*,\bar x^*) \to (\D,\bar a).$$
In particular, $(q',\bar x') \to (q^*,\bar x^*)$ 
and $(q^*,\bar x^*) \to (q',\bar x')$ by choosing $(\D,\bar a)$ as $(q^*,\bar x^*)$ and $(q',\bar x')$, respectively, 
in the previous equation.  
By composing $(q',\bar x') \to (q^*,\bar x^*)$ 
with $(q^*,\bar x^*) \to (q',\bar x')$, we 
obtain that $(q',\bar x') \to (q'_{\rm fin},\bar x')$, 
where $q'_{\rm fin}$ is a CQ with finitely many atoms, all of which belong to $q'$.  

Recall from the proof of Lemma \ref{lemma:games-infty} that the atoms of the CQ 
$q'$ are obtained as the union of the $Q_c$s, for $c \geq 0$, where $Q_c$ is 
the set of atoms in the CQ $q_c$, i.e., 
the one that describes the first $c$ rounds of the existential $k$-cover game played from $q$.   
Therefore, there must be an integer $c \geq 0$ such that $(q',\bar x') \to (q_c,\bar x_c)$.
For the same reason, it holds that 
$(q'_c,\bar x_c) \to (q',\bar x')$. We conclude that for every database $\D$ and tuple $\bar a$ of elements 
in $\D$: 
$$(q',\bar x') \to (\D,\bar a) \ \ \Longleftrightarrow \ \ (q_c,\bar x_c) \to (\D,\bar a).$$ 
In other words,  over every database $\D$ and tuple $\bar a$ of elements in $\D$,  it is the case that
$$(q,\bar x) \to_{k} (\D,\bar a) \ \ \Longleftrightarrow \ \ (q,\bar x) \to_{k}^c (\D,\bar a).$$
This finishes the proof of the theorem. 
\qed
\end{proof} 

Boundedness conditions 
are a difficult area of study, with a delicate decidability boundary; e.g., boundedness 
is decidable for Datalog programs if all intensional predicates are monadic
\cite{CGKV88}, but  undecidable if binary intensional predicates are allowed \cite{GMSV93}. 
For {\em least fixed point logic} (LFP), undecidability results for boundedness abound with the exception of a few 
restricted fragments \cite{Otto2006,BOW14}. 
Although the existence of winning Duplicator 
strategies in existential pebble games is expressible in LFP \cite{KV-jcss}, 
no result obtained in such context seems to be directly applicable to 
determining the decidability status of the boundedness condition in Theorem \ref{theo:boundedness}. 

\section{Overapproximations under Constraints}
\label{constraints} 

It has been observed that 
semantic information about the data -- in particular, in the form of constraints --  
enriches the quality of approximations \cite{BGP16}. This is exemplified next. 

\begin{example} \label{ex:constraints} 
As mentioned before, the Boolean CQ 
$$q \ = \ \exists x \exists y \exists z \big(E(x,y) \wedge E(y,z) \wedge E(z,x)\big),$$ 
shown in Figure ~\ref{fig:fig2}, has no $\HW(1)$-overapproximation. 
On the other hand, if we know that the data satisfies the constraint 
$$\forall x \forall y \forall z \big(E(x,y) \wedge E(y,z) \, \rightarrow \, E(z,x)\big),$$ 
then $q$ becomes equivalent to the CQ 
$q' = \exists x \exists y \exists  z (E(x,y) \wedge E(y,z))$, which is in $\HW(1)$.  
\qed
\end{example} 

In this section we study the notion of $\HW(k)$-overapproximation under constraints. 
We consider the two most important classes of database constraints; namely:

\begin{enumerate}
\item {\em Tuple-generating dependencies} (tgds), i.e., expressions of
  the form \begin{equation} 
  \label{eq:tgd} \forall \bar x \forall \bar y \, \big(\phi(\bar x,\bar y)
  \, \rightarrow \, \exists \bar z  \psi(\bar x,\bar z)\big),\end{equation}  
  where $\phi$ and
  $\psi$ are conjunctions of atoms. 
  Notice that the constraint in Example \ref{ex:constraints} is a tgd.

  Tgds subsume the central class of {\em inclusion dependencies} (IDs) \cite{Fagin81}. 
  For example, assuming that $R$ and $P$ are binary relations, the ID $R[1] \subseteq P[2]$, which states that the set of values occurring in the first attribute of $R$ is a subset of the set of values in the second attribute of $P$, is expressed via the tgd $\forall x \forall y (R(x,y) \rightarrow \exists z\, P(z,x))$. 
  
  There is a particular subclass of tgds that is expressive enough to subsume IDs and has received considerable attention in the literature; 
  namely, the class of {\em guarded} tgds \cite{CaGK08a}. 
A tgd is {\em guarded} if its
body $\phi(\bar x,\bar y)$ contains an atom, called the {\em guard},
that mentions all the variables in $(\bar x \cup \bar y)$. Notice that tgds that represent IDs are trivially guarded, as their body consists of a single atom.

\item {\em Equality-generating dependencies} (egds), i.e., expressions of
  the form 
  \begin{equation} \label{eq:egd} 
  \forall \bar x \, \big(\phi(\bar x) \, \rightarrow \, y = z \big),
  \end{equation}  
  where
  $\phi$ is a CQ and $y,z$ are variables in $\bar
  x$. Egds subsume the important classes of {\em keys} and {\em functional dependencies} (FDs). 
  
  For example, assuming that $R$ is a ternary relation, the FD $R : \{1\} \to \{3\}$, 
  i.e., the first attribute of $R$ functionally determines the third attribute of $R$, is expressed via the egd 
  $$\forall x \forall y \forall z \forall y' \forall z' \, (R(x,y,z) \wedge R(x,y',z') \, \rightarrow \, z = z').$$ 
  Notice that FDs that have more than one attribute in the right-hand side, can be expressed via a {\em set} of egds.

%
\end{enumerate}



A (potentially infinite) 
database $\D$ {\em satisfies} a tgd of the form given in Equation \eqref{eq:tgd}
if the following holds: for each tuple $(\bar a,\bar b)$ of elements such that $(\phi, \bar x,\bar y)\to(\D,\bar a,\bar b)$, there is a tuple $\bar c$ of elements
such that $(\psi,\bar x,\bar z)\to (\D,\bar a,\bar c)$.
Analogously, $\D$ satisfies an egd as in Equation \eqref{eq:egd} if, for each tuple $\bar a$ of elements such that $(\phi,\bar x)\to (\D,\bar a)$ via a homomorphism $h$, 
it is the case that $h(y)=h(z)$. Finally, $\D$ satisfies a set $\Sigma$ of constraints if it satisfies every tgd and egd in $\Sigma$.

\paragraph{{\em {\bf CQ containment under constraints.}}}
The right notion of containment, under constraints, is measured over those databases that satisfy the constraints only (as we know that our datasets satisfy such constraints). 
Formally, let $q,q'$ be CQs and $\Sigma$ a set of constraints. Then {\em $q$ is contained in $q'$ under $\Sigma$}, denoted $q \subseteq_\Sigma q'$, if and only if $q(\D) \subseteq q'(\D)$ for each database $\D$ that satisfies $\Sigma$. 
It is worth remarking that, as before, containment is defined over finite databases only.  
The notion of equivalence is defined analogously, and we write $q \equiv_{\Sigma} q'$.

The {\em chase procedure} is a canonical tool for reasoning about CQ 
containment under constraints~\cite{MaMS79}.  
We start by defining a single chase step for tgds.
Let $q$ be a CQ and $\tau$ a tgd of the form $\forall \bar x \forall \bar y (\phi(\bar x,\bar y) \rightarrow \exists \bar z \, \psi(\bar x,\bar z))$. 
We say that $\tau$ is \emph{applicable} with respect to $q$, if there exists a tuple $(\bar a,\bar b)$ of elements in $q$ such that 
$(\phi, \bar x,\bar y) \to (q,\bar a,\bar b)$. In this case, {\em the result of applying $\tau$ over $q$ with $(\bar a,\bar b)$} is the CQ $q'$ 
that extends $q$ with every atom in $\psi(\bar a,\bar c)$, where $\bar c$ is the tuple obtained by simultaneously replacing each  
$z \in \bar z$ with a fresh element not occurring in $q$. For such a single chase step we write
$q \xrightarrow{\tau,(\bar a,\bar b)} q'$. 

Let us now assume that $q$ is a CQ and $\Sigma$ a set of tgds. A {\em chase sequence for $q$ under $\Sigma$} is a sequence
\[
q_0 \xrightarrow{\tau_0,(\bar a_0,\bar b_0)} q_1 \xrightarrow{\tau_1,(\bar a_1,\bar b_1)} q_2 \dots
\]
of chase steps such that:
\begin{enumerate}
\item $q_0 = q$.

\item For each $i \geq 0$ we have that
$\tau_i$ is a tgd in $\Sigma$.

\item $q' \models \Sigma$, where $q'$ is the CQ formed by the union of the atoms in the $q_i$s, for $i \geq 0$.
\end{enumerate}
The (potentially infinite) CQ $q'$ is the {\em result} of this chase sequence, which always exists.

Although the result of a chase sequence is not unique (up to isomorphism), each such result is equally useful for our purposes since it can be homomorphically embedded into every other result.
%
%
This is a consequence of the fact that the result $q'$ of a chase sequence for $q$ under $\Sigma$ is {\em universal}, i.e., 
for every (potentially infinite) CQ $q''$ such that $q \subseteq q''$ and $q'' \models \Sigma$, there is a homomorphism from $q'$ to $q''$~\cite{FKMP05,DeNR08}.
%
Henceforth, we write $chase_\Sigma(q)$ for the result of an arbitrary chase sequence for $q$ under $\Sigma$. 

As for tgds, the chase is a useful tool when reasoning with egds. Let us first define a single  chase step for egds.
Consider a CQ $q$ and an egd $\epsilon$ of the form
$\forall \bar x (\phi(\bar x) \rightarrow x_i = x_j)$.
We say that $\epsilon$ is \emph{applicable} with respect to $q$, if there exists a homomorphism $h$ that witnesses $\phi \to q$ 
for which it holds that $h(x_i) \neq h(x_j)$. In this case, {\em the result of applying $\epsilon$ over $q$ with $h$} is the CQ $q'$ that obtained from $q$ by identifying 
$h(x_i)$ and $h(x_j)$ everywhere.
We can define the notion of the chase sequence for a CQ $q$ under a set $\Sigma$ of egds analogously as we did for tgds. 
Notice that such a sequence is finite and unique; thus, we refer to {\em the} chase for $q$ under $\Sigma$, denoted $chase_\Sigma(q)$.
Observe that the chase sequence that leads to $chase_\Sigma(q)$ gives rise to a homomorphism $h_{q,\Sigma} : q \to chase_\Sigma(q)$ 
such that $h_{q,\Sigma}(q) = chase_\Sigma(q)$.
 It is well-known that an extended notion of containment under constraints -- which is defined over 
both finite and infinite databases -- can be characterized in terms of the notion of homomorphism and the chase procedure. Formally, let $q,q'$ be CQs and $\Sigma$ a set of constraints. We write 
$q \subseteq^\infty_\Sigma q'$ iff $q(\D) \subseteq q'(\D)$ for each countable 
database $\D$ that satisfies $\Sigma$. T


\begin{lemma}\label{lemma:cq-equiv-tgds}
Let $q(\bar x), q'(\bar x')$ be CQs. 
\begin{enumerate} 
\item If $\Sigma$ is a set of tgds, then 
\begin{equation*}
\label{eq:chase-tgds}
q \subseteq^\infty_\Sigma q' \, \, \Longleftrightarrow \, \, (q',\bar x') \to (chase_\Sigma(q),\bar x). 
\end{equation*}
\item If $\Sigma$ is a set of egds, then 
\begin{equation*}
\label{eq:chase-egds}
q \subseteq^\infty_\Sigma q' \, \, \Longleftrightarrow \, \, 
(q',\bar x') \to (chase_\Sigma(q),h_{q,\Sigma}(\bar x)). 
\end{equation*}
\end{enumerate} 
\end{lemma}

The sets $\Sigma$ of constraints for which the notions $\subseteq_\Sigma$ and $\subseteq_\Sigma^\infty$
coincide are called {\em finitely controllable}. It is easy to see that any set $\Sigma$ 
of egds is finitely controllable, as $chase_\Sigma(q)$ is always finite in such a case. On the other hand, 
arbitrary sets of tgds are not necessarily finitely controllable. An important exception corresponds to the case when $\Sigma$ is a set of guarded tgds. In fact, a deep result shows that sets of guarded tgds are finitely controllable, i.e., for any CQs $q,q'$ and set $\Sigma$ of guarded tgds, it holds that $q \subseteq_\Sigma q' \, \Leftrightarrow \, q \subseteq_\Sigma^\infty q'$ \cite{BGO13}.

\subsection{$\HW(k)$-overapproximations under constraints}


In the absence of constraints, $\HW(k)$-overapproximations may not exist. 
However, as discussed in Section \ref{overapproximations}, by considering the class $\HW(k)^\infty$ of CQs with countable many atoms and generalized hypertreewidth bounded by $k$, we can provide 
each CQ with a $\HW(k)^\infty$-overapproximation. We show that this good behavior generalizes to the 
case when constraints expressed as egds or guarded tgds are present. 

Formally, given CQs $q$ and $q'$ such that $q' \in \HW(k)^\infty$, 
we say that $q'$ is a {\em $\HW(k)^\infty$-overapproximation of $q$ under $\Sigma$}, if (i) $q\subseteq_\Sigma q'$, 
and (ii) there is no $q''\in \HW(k)^\infty$ 
such that $q\subseteq_\Sigma q''\subsetneq_\Sigma q'$.

\begin{theorem}
\label{theo:cons-exist}
Fix $k\geq 1$. For every CQ $q(\bar x)$ and set $\Sigma$ consisting 
exclusively of egds or guarded tgds, there is a CQ $q'(\bar x') \in \HW(k)^\infty$ that is a 
$\HW(k)^\infty$-overapproximation of $q$ under $\Sigma$.  
\end{theorem}

\begin{proof} We only consider the case when $\Sigma$ is a set of guarded 
tgds, as for egds the proof is even simpler. 
As in Lemma \ref{lemma:games-infty}, we can show 
that there is a CQ $q'(\bar x')$ such that for every countable database $\D$ and tuple
 $\bar a$ in $\D$:  
\begin{equation} 
\label{eq:inf}
\bar a \in q(\D) \ \ \Longleftrightarrow \ \ (q',\bar x') \to (\D,\bar a) \ \ \Longleftrightarrow \ \ 
(chase_\Sigma(q),\bar x) \to_k (\D,\bar a).
\end{equation}  
Notice that $q'$ is not guaranteed to exist a priori, since Lemma \ref{lemma:games-infty} 
is stated only for a finite CQ $q$, while $chase_\Sigma(q)$ 
may be infinite. However, the same proof of Lemma \ref{lemma:games-infty} 
applies to the infinite case. 
We show next that $q'$ is the $\HW(k)^\infty$-overapproximation of $q$ under $\Sigma$. 

First, $(q',\bar x') \to (chase_\Sigma(q),\bar x)$ by choosing $(\D, \bar a)$ as 
$(chase_\Sigma(q), \bar x)$ in Equation \eqref{eq:inf}, and hence  
$q \subseteq^\infty_\Sigma q'$ as this direction of 
Lemma \ref{lemma:cq-equiv-tgds} holds even for CQs with countable many atoms. 
Therefore, also $q \subseteq_\Sigma q'$. 
On the other hand, $(chase_\Sigma(q),\bar x) \to_{k} (q',\bar x')$ 
by choosing $(\D, \bar a)$ as $(q', \bar x')$ in Equation \eqref{eq:inf}. 
Proposition \ref{prop:games-tw-infty} then tells us that for each $q''(\bar x'')$ in $\HW(k)^\infty$, 
if $(q'',\bar x'') \to (chase_\Sigma(q),\bar x)$ then $(q'',\bar x'') \to (q',\bar x')$.
But then 
$q \subseteq_\Sigma q''$ implies $q' \subseteq q''$. 
This is because $q \subseteq_\Sigma q''$ implies $(q'',\bar x'') \to (chase_\Sigma(q),\bar x)$ by finite controllability of $\Sigma$. Indeed, in this case $q \subseteq_\Sigma q''$ implies 
$q \subseteq^\infty_\Sigma q''$, and thus $q(\D) \subseteq q''(\D)$ over every countable database 
$\D$ that satisfies $\Sigma$. Since $chase_\Sigma(q)$ is one such a database and $\bar x \in q(chase_\Sigma(q))$, we have that $\bar x \in q''(chase_\Sigma(q))$. By definition then,
$(q'',\bar x'') \to (chase_\Sigma(q),\bar x)$.  


In summary, $q \subseteq_\Sigma q'$ and for each CQ $q''$ in $\HW(k)^\infty$ we have that 
$$q \subseteq_\Sigma q'' \ \ \Longrightarrow \ \ q' \subseteq q''.$$    
This implies that $q'$ is the $\HW(k)^\infty$-overapproximation of $q$ under $\Sigma$. \qed
\end{proof} 

The only property of the sets of guarded tgds that we used in the previous proof is finite controllability. 
Since finite controllability does not hold for general tgds, 
we cannot extend Theorem \ref{theo:cons-exist} to arbitrary sets of constraints. 
On the other hand, if we change the notion of $\HW(k)^\infty$-overapproximation under $\Sigma$ to be defined 
in terms of $\subseteq_\Sigma^\infty$, we can mimic the proof of 
Theorem \ref{theo:cons-exist} and provide each CQ with a $\HW(k)^\infty$-overapproximation under $\Sigma$, 
where $\Sigma$ is an arbitrary set of tgds. While considering countable databases 
is not standard in databases, it is a common choice for the semantics of 
containment and related problems 
in the area of ontology-mediated query answering; cf., \cite{CaGK08a,BGO13}.  

As a corollary to the proof of Theorem \ref{theo:cons-exist} we obtain that the $\HW(k)^\infty$-overapproximation of a CQ 
under $\Sigma$, 
where $\Sigma$ is a set of egds or guarded tgds, can be evaluated 
by applying the existential $k$-cover game on $chase_\Sigma(q)$. 

\begin{corollary} 
\label{coro:chase-inf-eval} 
Fix $k\geq 1$. Consider a CQ $q(\bar x)$ and a set $\Sigma$ consisting 
exclusively of egds or guarded tgds, such that the 
$\HW(k)^\infty$-overapproximation of $q$ under $\Sigma$ 
is $q'(\bar x')$. Then for every database 
$\D$ and tuple $\bar a$ in $\D$ it is the case that 
$$\bar a \in q'(\D) \ \ \Longleftrightarrow \ \ (q',\bar x') \to (\D,\bar a) \ \ 
\Longleftrightarrow \ \ (chase_\Sigma(q),\bar x) \to_k (\D,\bar a).$$  
\end{corollary} 


\paragraph{{\em {\bf Evaluating overapproximations under constraints.}}}
While in the absence of constraints the $\HW(k)^\infty$-overapproximation of $q$ can be evaluted in 
polynomial time by applying the existential $k$-cover game on $q$,
 the situation is more complex in the presence 
of constraints. In fact, as stated in Corollary \ref{coro:chase-inf-eval} to evaluate 
the $\HW(k)^\infty$-overapproximation of $q$ under $\Sigma$ by using the existential $k$-cover game, we first
 need to compute the result of the chase on $q$. For arbitrary egds this might take exponential time, as 
checking whether $chase_\Sigma(q) = q$, when $\Sigma$ consists of a single egd $\epsilon$, is an \np-complete problem (since we need to detect whether the CQ that defines the body of $\epsilon$ is 
applicable on $q$). For sets of guarded tgds the result of the chase might be infinite, and thus it is not even computable. 

We show that, in spite of the previous observation, if 
$\Sigma$ is a set of guarded tgds then the 
$\HW(k)^\infty$-overapproximation of $q$ under $\Sigma$ can be evaluated in polynomial time. This is because when a database $\D$ satisfies $\Sigma$, 
applying the existential $k$-cover game on $q$ or $chase_\Sigma(q)$ is the same.    


\begin{theorem}
Fix $k\geq 1$. Given a CQ $q$, a set of $\Sigma$ of guarded tgds, a database $\D$ satisfying $\Sigma$, and a tuple $\bar a$, 
checking whether $\bar a\in q'(\D)$, where 
$q'$ is the $\HW(k)^\infty$-overapproximation of $q$ under $\Sigma$, can be solved in 
polynomial time by simply verifying if $(q,\bar x) \to_k (\D,\bar a)$.  
\end{theorem}

\begin{proof} 
As stated in \cite{BGP16}, if $\Sigma$ is a set of guarded tgds and $\D$ satisfies $\Sigma$, then  
$(chase_\Sigma(q),\bar x) \to_1 (\D,\bar a) \, 
\Longleftrightarrow \, (q,\bar x) \to_1 (\D,\bar a)$ for every CQ $q$.    
A slight modification of this proof shows that this property extends to any $k > 1$. That is, for every 
$k \geq 1$ it is the case that 
$$(chase_\Sigma(q),\bar x) \to_k (\D,\bar a) \ \ 
\Longleftrightarrow \ \ (q,\bar x) \to_k (\D,\bar a).$$
The result then follows from this equivalence and Corollary \ref{coro:chase-inf-eval}.   
\qed
\end{proof} 

For egds, on the other hand, we obtain that the problem 
can be solved in time $|D|^{O(1)} \cdot f(|q|)$, for a
 computable function $f : \mathbb{N} \to \mathbb{N}$. In 
 parameterized complexity terms, this means that the problem
is {\em fixed-parameter tractable} (FPT), with the parameter being 
the size of the CQ. This is a positive result, as no FPT for general CQ evaluation 
is believed to exist \cite{PY99}.    

\begin{theorem}
Fix $k\geq 1$. Given a CQ $q$, a set of $\Sigma$ of egds, 
a database $\D$ satisfying $\Sigma$ and a tuple $\bar a$, 
checking whether $\bar a\in q'(\D)$, where 
$q'$ is the $\HW(k)^\infty$-overapproximation of $q$ under $\Sigma$, 
can be solved by an FPT 
algorithm.   
\end{theorem}

\begin{proof} 
The algorithm first computes $chase_\Sigma(q)$ in exponential time in the size of $q$, and then 
checks whether $(chase_\Sigma(q),\bar x) \to_k (\D,\bar a)$ in polynomial time in the combined size
of $chase_\Sigma(q)$ and $\D$. But the size of $chase_\Sigma(q)$ is bounded by that of $q$, and thus the 
whole procedure can be carried out in time 
$$2^{p(|q|)} \, + \, (|D|+|q|)^{O(1)},$$
for $p : \mathbb{N} \to \mathbb{N}$ a polynomial. Hence, the algorithm is FPT. \qed
\end{proof} 

Notice than in case that $\Sigma$ consists exclusively of FDs, then the problem can be solved in 
polynomial time. This is because in such a case $chase_\Sigma(q)$ can be computed in polynomial time. 

\begin{corollary}
Fix $k\geq 1$. Given a CQ $q$, a set of $\Sigma$ of FDs, 
a database $\D$ satisfying $\Sigma$, and a tuple $\bar a$, 
checking whether $\bar a\in q'(\D)$, where 
$q'$ is the $\HW(k)^\infty$-overapproximation of $q$ under $\Sigma$, can be solved in polynomial time.   
\end{corollary}

\section{Beyond Under- and Overapproximations: $\Delta$-Approximations}\label{beyond} 

We now turn to $\Delta$-approximations. 
Recall that a $\HW(k)$-$\Delta$-approximation of $q$ is a maximal element in $\HW(k)$ with respect to the partial 
order $\sqsubseteq_q$, where $q'\sqsubseteq_q q''$, for CQs $q',q''\in \HW(k)$, iff $\Delta(q(\D),q''(\D))\subseteq \Delta(q(\D),q'(\D))$ for all databases $\D$. 
It is worth noticing that $\HW(k)$-$\Delta$-approximations generalize over- and underapproximations. 

\begin{proposition}
\label{prop:general1}
Fix $k\geq 1$. Let $q,q'$ be CQs such that $q'\in \HW(k)$. 
If $q\subseteq q'$ (resp., $q'\subseteq q$), then $q'$ is a $\HW(k)$-$\Delta$-approximation of $q$ if and only if $q'$ is an 
$\HW(k)$-overapproximation (resp., $\HW(k)$-underapproximation) of $q$.
\end{proposition}

\begin{proof}
We only prove it for the case when $q \subseteq q'$. The proof for the case when $q' \subseteq q$ is analogous. 

  Suppose first that $q'$ is a $\HW(k)$-$\Delta$-approximation of $q$. Assume, towards a contradiction, that 
  there is a query $q''$ such that $q \subseteq q'' \subset q'$. Then (i) $q(\D) \subseteq q''(\D) \subseteq q'(\D)$ for each database $\D$, and (2) 
  there is a database $\D^*$ such that $q'(\D^*)  \not\subseteq q''(\D^*)$. In particular, $\Delta(q(\D),q''(\D))\subseteq \Delta(q(\D),q'(\D))$ for each database $\D$, while 
  $\Delta(q(\D^*),q'(\D^*))\not\subseteq \Delta(q(\D^*),q''(\D^*))$. This is a contradiction to 
  $q'$ being a $\HW(k)$-$\Delta$-approximation of $q$.
  
  Now suppose that $q'$ is a $\HW(k)$-overapproximation of $q$. Assume, towards a contradiction, that there is a query $q''$ such that (i) $\Delta(q(\D),q''(\D))\subseteq \Delta(q(\D),q'(\D))$ for each database $\D$, while (ii) for some database $\D^*$ it is the case that 
  $\Delta(q(\D^*),q'(\D^*))\not\subseteq \Delta(q(\D^*),q''(\D^*))$. 
  Then $q \subseteq q'' \subseteq q'$ as $\Delta(q(\D),q'(\D))$ may only contain tuples in $q'(\D) \setminus 
  q(\D)$. Also, $q'(\D^*) \not\subseteq q''(\D^*)$. This is a contradiction to $q'$ being  a $\HW(k)$-overapproximation of $q$. \qed
  %
\end{proof}

For this reason, we concentrate on 
the study of $\HW(k)$-$\Delta$-approximations that are neither $\HW(k)$-under- nor $\HW(k)$-overapproximations. 
Evaluating such $\Delta$-approximations can give us useful information that complements the one provided by under- and overapproximations. 
But, do these  $\HW(k)$-$\Delta$-approximations exist at all, i.e., 
are there $\HW(k)$-$\Delta$-approximations that are neither $\HW(k)$-under- nor $\HW(k)$-overapproximations? 
In the rest of this section, we settle this question and study complexity questions associated with such $\HW(k)$-$\Delta$-approximations.

\subsection{Incomparable $\HW(k)$-$\Delta$-approximations}

Let $q$ be a CQ. 
 In view of Proposition \ref{prop:general1}, 
the $\HW(k)$-$\Delta$-approximations $q'$ of $q$  
that are neither $\HW(k)$-under nor $\HW(k)$-overapproximations must be {\em incomparable} with $q$ in terms of containment; i.e., 
both $q\not\subseteq q'$ and $q'\not\subseteq q$ must hold. Incomparable $\HW(k)$-$\Delta$-approximations do not necessarily exist, even when approximating in the set of infinite CQs $\HW(k)^\infty$. A trivial example is any CQ $q$ in $\HW(k)$, as its only $\HW(k)$-$\Delta$-approximation (up to equivalence) is $q$ itself. 
On the other hand, the following characterization will help us to find CQs that do have
 incomparable $\HW(k)$-$\Delta$-approximations.

\begin{theorem}
\label{theo:char-incomp}
Fix $k\geq 1$. Let $q(\bar x),q'(\bar x')$ be CQs such that $q'\in \HW(k)$. 
Then $q'$ is an incomparable $\HW(k)$-$\Delta$-approximation of $q$ 
iff $(q,\bar x)\to_k (q',\bar x')$, and both $q \not\subseteq q'$ and $q' \not\subseteq q$ hold.
\end{theorem}

\begin{proof}
Suppose that $q'$ is an incomparable $\HW(k)$-$\Delta$-approximation of $q$ and assume, by contradiction, 
that $(q,\bar x)\not\to_k (q',\bar x')$. By Proposition \ref{prop:games-tw}, there is a $q''(\bar x'')\in \HW(k)$ such that $q\subseteq q''$ and $q'\not\subseteq q''$. 
Recall that $(q''\wedge q')(\bar z)$ denotes the disjoint conjunction of $q''$ and $q'$ (see Section \ref{sec:uniqueness} for the precise definition). 
We show that $q'\sqsubset_q (q''\wedge q')$, which is a contradiction as $(q''\wedge q') \in \HW(k)$. 
Assume that $\bar a\in\Delta(q(\D),(q''\wedge q')(\D))$, for some $\D$ and $\bar a \in \D$. If $\bar a\not\in q(\D)$, then $\bar a\in (q''\wedge q')(\D)\subseteq q'(\D)$, 
and thus, $\bar a\in \Delta(q(\D), q'(\D))$. Otherwise, 
$\bar a\in q(\D)$ and $\bar a\not\in (q''\wedge q')(\D)$.  Since $q\subseteq q''$, we have $\bar a\not\in q'(\D)$, and then $\bar a\in \Delta(q(\D), q'(\D))$. Hence $q'\sqsubseteq_q (q''\wedge q')$. 
Now, since $q'\not\subseteq q''$, there is a database $\D^*$ such that $q'(\D^*)\not\subseteq q''(\D^*)$, i.e., 
$\bar a\in q'(\D^*)$ but $\bar a\not\in q''(\D^*)$, for some tuple $\bar a$ in $\D^*$. 
In particular $\bar a\in \Delta(q(\D^*),q'(\D^*))$ and $\bar a\not\in\Delta(q(\D^*),(q''\wedge q')(\D^*))$, and thus $(q''\wedge q') 
\not\sqsubseteq_q  q'$. 

For the converse, we need the following lemma whose proof can be found in the appendix. 

\begin{lemma}
\label{lemma:aux}
Fix $k\geq 1$. Let $q(\bar x),q'(\bar x'),q''(\bar x'')$ be CQs such that $q''\in \HW(k)$. 
Let $(q'\wedge q)(\bar z)$ be the disjoint conjunction of $q'$ and $q$. 
Suppose that $(q,\bar x)\to_k(q',\bar x')$. Then $(q'',\bar x'')\to (q'\wedge q, \bar z)$ implies $(q'',\bar x'')\to (q', \bar x')$. 
\end{lemma}

Assume now 
that $q\nsubseteq q'$, $q'\nsubseteq q$, and $(q,\bar x)\to_k (q',\bar x')$. By contradiction, 
suppose that there is a CQ $q''\in \HW(k)$ such that $q'\sqsubset_q q''$. We show that $q'\equiv q''$, which is a contradiction. 
Recall that $\D_{(q'\wedge q)}$ denotes the canonical database of $(q'\wedge q)(\bar z)$. 
Clearly, $\bar z\in q(\D_{(q'\wedge q)})$ and $\bar z\in q'(\D_{(q'\wedge q)})$. It follows that $\bar z\not\in \Delta(q(\D_{(q'\wedge q)}), q'(\D_{(q'\wedge q)}))$, 
and by hypothesis, $\bar z\not\in \Delta(q(\D_{(q'\wedge q)}), q''(\D_{(q'\wedge q)}))$. Hence, $\bar z\in q''(\D_{(q'\wedge q)})$. 
By Lemma \ref{lemma:aux}, we have  $(q'',\bar x'')\to (q',\bar x')$, that is, $q'\subseteq q''$. 
For $q''\subseteq q'$, note that $\bar x''\not\in q(\D_{q''})$; otherwise, $q''\subseteq q$ would hold, implying that $q'\subseteq q$, which is a contradiction. 
Since $\bar x''\in q''(\D_{q''})$, we have $\bar x''\in \Delta(q(\D_{q''}),q''(\D_{q''}))$. This implies that 
$\bar x''\in \Delta(q(\D_{q''}),q'(\D_{q''}))$, and then $\bar x''\in q'(\D_{q''})$, i.e., $q''\subseteq q'$. Hence, $q'\equiv q''$. \qed
\end{proof}

\begin{example}
\label{ex:general} 
Consider again the CQ 
$$q \ = \ \exists x\exists y\exists z \big(E(x,y)\wedge E(y,z) \wedge E(z,x)\big)$$ 
from Figure~\ref{fig:fig2}. 
It is easy to prove that 
$q'=\exists x E(x,x)$ is the unique $\HW(1)$-underapproximation of $q$. 
Also, as mentioned in Section \ref{sec:uniqueness}, $q$ has no $\HW(1)$-overapproximations. 
Is it the case, on the other hand, that $q$ has incomparable $\HW(1)$-$\Delta$-approximations? By applying Theorem \ref{theo:char-incomp}, we can give a positive answer to this question. In fact,  
the CQ $$q''  \ = \ \exists x\exists y \big(E(x,y)\wedge E(y,x)\big)$$ is an incomparable $\HW(1)$-$\Delta$-approximation of $q$.  \qed
\end{example}

Therefore, as Example \ref{ex:general} shows, incomparable $\HW(k)$-$\Delta$-approximations may exist for some CQs. 
However, in contrast with overapproximations, they are not unique in general (see the appendix for details). 

\begin{proposition} 
\label{prop:non-unique} 
There is a CQ with infinitely many (non-equivalent) incomparable $\HW(1)$-$\Delta$-approximations.
In fact, this even holds for the CQ $q$ shown in Figure \ref{fig:fig1}.
\end{proposition}

\paragraph{{\em {\bf Identification, existence and evaluation.}}}
A direct consequence of Theorem \ref{theo:char-incomp} is that the identification problem, i.e.,  
checking if $q' \in \HW(k)$ is an incomparable $\HW(k)$-$\Delta$-approximation of a CQ $q$, is in coNP. 
It suffices to check that $q \not\subseteq q'$ and $q' \not\subseteq q$ -- which are in coNP -- and 
$(q,\bar x) \to_k (q',\bar x')$ -- which is in {\sc Ptime} from Proposition \ref{prop:games-poly}.  
We show next that this bound  is optimal (the proof is in the appendix). 

\begin{proposition}
\label{prop:incomp-ident} 
Fix $k\geq 1$. 
Checking if a given CQ $q'\in \HW(k)$ is an incomparable $\HW(k)$-$\Delta$-approximation of a given CQ $q$, is {\sc coNP}-complete. 
\end{proposition}

As in the case of $\HW(k)$-overapproximations, we do not know how to check existence of incomparable $\HW(k)$-$\Delta$-approximations, for $k>1$. 
Nevertheless, for $k =1$ we can exploit the automata techniques developed in Section \ref{sec:exist-acyclic} and obtain 
an analogous decidability result. 

\begin{proposition}
\label{prop:existence-incomp} 
There is a {\sc 2Exptime} algorithm that checks if a CQ $q$ has an incomparable $\HW(1)$-$\Delta$-approximation and, if one exists, 
it computes one in triple exponential time.  
The bounds become {\sc Exptime} and {\sc 2Exptime}, respectively, if the maximum arity of the 
schema is fixed. 
\end{proposition}

Now we study evaluation. Recall that, unlike $\HW(k)$-overapproximations, 
incomparable $\HW(k)$-$\Delta$-approximations of a CQ $q$ are not unique. In fact, there can be infinitely many 
(see Proposition \ref{prop:non-unique}). Thus, it is reasonable to start by trying to evaluate 
{\em at least one} of them. It would be desirable, in addition, if 
the one we evaluate depends only on $q$ (i.e., it is independent 
of the underlying database $\D$). 
Proposition \ref{prop:existence-incomp} allows us to do so as follows. 
Given a CQ $q$ with at least one incomparable $\HW(1)$-$\Delta$-approximation, we can compute in {\sc 3Exptime} 
one such an incomparable $\HW(1)$-$\Delta$-approximation $q'$. We can then evaluate $q'$ over a database $\D$ 
in time $O(|\D| \cdot |q'|)$ \cite{Yan81}, which is 
$O(|\D| \cdot f(|q|))$, for $f$ a triple-exponential function. This means that 
the evaluation of such a $q'$ over $\D$ is fixed-parameter tractable, i.e., it can be solved by an algorithm that depends 
polynomially on the size of the large database $\D$, but more loosely on the size of the small CQ $q$. 
(Recall that this is a desirable property for evaluation, which does not hold in general for the class of all CQs \cite{PY99}). 

\begin{theorem}
There is a fixed-parameter tractable algorithm that, given a CQ $q$ that has incomparable $\HW(1)$-$\Delta$-approximations, a database $\D$, 
and a tuple $\bar a$, checks whether $\bar a\in q'(\D)$, for some incomparable $\HW(1)$-$\Delta$-approximation $q'$ of $q$ 
that depends only on $q$. 
\end{theorem}

It is worth noticing that the automata techniques are essential for proving this result, and thus for 
evaluating incomparable $\HW(1)$-$\Delta$-approximations. This is in stark contrast with $\HW(k)$-overapproximations, which 
can be evaluated in polynomial time by simply checking if $(q,\bar x) \to_k (\D,\bar a)$. It is not 
at all clear whether such techniques can be extended to allow for the efficient evaluation of 
incomparable $\HW(k)$-$\Delta$-approximations.

\paragraph{{\em {\bf The infinite case.}}}
All the previous results continue to apply for the class of 
infinite CQs in $\HW(k)^\infty$. 
The following example shows that, as in the case of $\HW(k)$-overapproximations, 
considering $\HW(k)^\infty$  helps us to obtain better incomparable $\HW(k)$-$\Delta$-approximations.

\begin{figure}
\centering 

 \scalebox{0.85}{
  \begin{tikzpicture}[
      xscale=0.9,
      yscale=0.5,      
    ]
     \node (0) at (-0.5,2){$q$:};
         \node (0) at (2.1,2.6){$P_1$};
                  \node (0) at (2.1,-0.6){$P_2$};
      \node (0) at (-1,-0.55){};
     
           \node[mnode, minimum size=5pt] (1) at (-0.5,1){};
      \node[mnode, minimum size=2pt] (v2) at (1.2,2){};
       \node[mnode, minimum size=2pt] (v4) at (2.9,2){};
             \node[mnode, minimum size=2pt] (u2) at (1.2,0){};
       \node[mnode, minimum size=2pt] (u4) at (2.9,0){};
       
                 \node[mnode, minimum size=2pt] (v1) at (0.35,1.5){};
                  \node[mnode, minimum size=2pt] (u1) at (0.35,0.5){};
                    \node[mnode, minimum size=2pt] (v5) at (3.75,1.5){};
                     \node[mnode, minimum size=2pt] (u5) at (3.75,0.5){};
                           \node[mnode, minimum size=2pt] (v3) at (2.05,2){};
                     \node[mnode, minimum size=2pt] (u3) at (2.05,0){};

               \node[mnode, minimum size=5pt] (2) at (4.6,1){};

     \draw [dEdge] (1) to (v1);
       \draw [dEdge] (v1) to (v2);
         \draw [dEdge] (v2) to (v3);
           \draw [dEdge] (v4) to (v3);
             \draw [dEdge] (v4) to (v5);
               \draw [dEdge] (v5) to (2);
               
         \draw [dEdge] (1) to (u1);
                \draw [dEdge] (u1) to (u2);
                    \draw [dEdge] (u3) to (u2);
                        \draw [dEdge] (u3) to (u4);
                            \draw [dEdge] (u4) to (u5);
                                \draw [dEdge] (u5) to (2);
                
    \begin{scope}[shift={(7,0.8)}]

      \node (0) at (-0.7,1){$q^*$:};
       \node (0) at (8.3,1){$\dots$};
      \node (0) at (-1,-0.6){};
      
      \node[mnode, minimum size=5pt] (1) at (0,1){};
      \node[mnode, minimum size=2pt] (2) at (1.5,1){};
       \node[mnode, minimum size=2pt] (3) at (3,1){};
        \node[mnode, minimum size=2pt] (4) at (4.5,1){};
         \node[mnode, minimum size=2pt] (5) at (6,1){};
               \node[mnode, minimum size=2pt] (6) at (7.5,1){};
               
                 \draw [dEdge] (1) to node[above]{$P_1$} (2);
                 \draw [dEdge] (3)  to node[above]{$P_2$} (2);
                 \draw [dEdge] (3)  to node[above]{$P_1$} (4);
                 \draw [dEdge] (5)  to node[above]{$P_2$} (4);
                 \draw [dEdge] (5)  to node[above]{$P_1$} (6);
    \end{scope}       

    \begin{scope}[shift={(7,-0.8)}]
      \node (0) at (-0.7,1){$q'$:};
      \node (0) at (-1,-0.6){};
      
      \node[mnode, minimum size=5pt] (1) at (0,1){};
      \node[mnode, minimum size=2pt] (2) at (1,1){};
       \node[mnode, minimum size=2pt] (3) at (2,1){};
        \node[mnode, minimum size=2pt] (4) at (3,1){};
         \node[mnode, minimum size=2pt] (5) at (4,1){};
               \node[mnode, minimum size=2pt] (6) at (5,1){};
         \node[mnode, minimum size=2pt] (7) at (6,1){};
               \node[mnode, minimum size=5pt] (8) at (7,1){};
               
                 \draw [dEdge] (1) to (2);
                 \draw [dEdge] (2) to (3);
                 \draw [dEdge] (3) to (4);
                 \draw [dEdge] (5) to (4);
                 \draw [dEdge] (5) to (6);
                 \draw [dEdge] (6) to (7);
                 \draw [dEdge] (7) to (8);
        \end{scope}
      \end{tikzpicture}
  }
  \vspace{-2mm}

\caption{The CQ $q\in \HW(2)$ from Example \ref{ex:incomp}. The CQ $(q^*\wedge q')$ is an incomparable $\HW(1)^\infty$-$\Delta$-approximation of $q$. 
On the other hand, $q$ has no incomparable $\HW(1)$-$\Delta$-approximations.}
\label{fig:fig-incomp}
\end{figure}
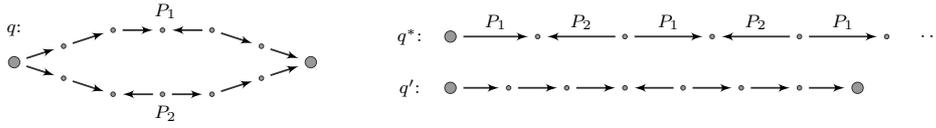

\begin{example}
\label{ex:incomp}
Consider the CQ $q$ that asks for the existence of the two oriented paths $P_1$ and $P_2$, as shown in Figure \ref{fig:fig-incomp}. 
Theorem \ref{theo:char-incomp} can be used to show that $q$ has no incomparable $\HW(1)$-$\Delta$-approximation. 
However, $q$ has an incomparable $\HW(1)^\infty$-$\Delta$-approximation. 
In fact, let $q^*$ be the CQ which is depicted in Figure \ref{fig:fig-incomp} 
(a $P_1$-labeled edge represents a copy of the oriented path $P_1$, similarly for $P_2$). 
It is easy to see that $q^*$ is the $\HW(1)^\infty$-overapproximation of $q$. 
Also, let $q'$ be an arbitrary CQ in $\HW(1)$ which is incomparable with $q$ (one such 
$q'$ is shown in Figure \ref{fig:fig-incomp}). 
Applying the extension of Theorem \ref{theo:char-incomp} to the class $\HW(k)^\infty$, 
we can prove that $(q^*\wedge q')$ is an incomparable $\HW(1)^\infty$-$\Delta$-approximation of $q$. 
\qed
\end{example}

Example \ref{ex:incomp} also illustrates the following  fact: 
If there is a CQ $q'\in \HW(k)$ which is incomparable with $q$, then $(q^* \wedge q')$ is an incomparable $\HW(k)^\infty$-$\Delta$-approximation of $q$, 
where $q^*$ is the $\HW(k)^\infty$-overapproximation of $q$. 
Given a database $\D$ and a tuple $\bar a$ in $\D$, we can check whether $\bar a$ belongs to the evaluation of 
such a $\Delta$-approximation $(q^*\wedge q')$ over $\D$ as follows. First we compute $q'$, and then we check both 
$(q,\bar x)\to_k (\D,\bar a)$ and $\bar a\in q'(\D)$. In other words, we evaluate $(q^*\wedge q')$ 
via the existential $k$-cover game -- as for the $\HW(k)^\infty$-overapproximation --, and then 
use the incomparable CQ $q'$ to filter out some tuples in the answer. 
Interestingly, we can easily exploit automata techniques and compute such an incomparable $q'$ (in case one exists). 
Thus we have the following: 

\begin{theorem}
Fix $k\geq 1$. There is a fixed-parameter tractable algorithm that given a CQ $q$ that has an incomparable $q'$ in $\HW(k)$, 
a database $\D$, and a tuple $\bar a$ in $\D$, decides whether $\bar a\in \hat q(\D)$, for some 
incomparable $\HW(k)^\infty$-$\Delta$-approximation $\hat q$ of $q$ that depends only on $q$. 
\end{theorem}

\section{Final Remarks}\label{section:conclusion}

Several problems remain open:  
is the existence of $\HW(k)$-overapproximations decidable for $k > 1$? What is the precise complexity 
of checking for the existence of $\HW(1)$-overapproximations? In particular, can we improve the {\sc 2Exptime} upper bound from Theorem \ref{theo:decidability-acyclic}?   
What is an optimal upper bound on the size of 
$\HW(1)$-overapproximations? Can we extend to non-Boolean CQs the result that states
 the tractability of checking for the existence of $\HW(1)$-overapproximations over binary schemas?

In the future we plan to study how our notions of approximation 
can be combined with other techniques to obtain quantitative guarantees. One possibility is to 
exploit semantic information about the data -- e.g., in the form of integrity constraints -- 
in order to ensure that certain bounds on the size
of the result of the approximation hold. Another possibility is to try to obtain probabilistic guarantees for approximations based on reasonable 
assumptions about the distribution of the data.


\bibliographystyle{spmpsci}      
\bibliography{overapp}   


\section*{Appendix}

\begin{proof}[Theorem \ref{theo:non-existence}]
Fix $k > 1$. The CQ $q$ is defined over graphs, i.e., over a schema with a single binary relation 
symbol $E$, and 
consists of $k+1$ variables $v_1,\dots,v_{k+1}$. 
For every $1 \leq i < j \leq k+1$ we add either the atom (i.e., edge) 
$E(v_i,v_j)$ or $E(v_j,v_i)$ to $q$ in such a way that the subgraph of $G$ induced by 
$\{v_1,v_2,v_3\}$ is a directed cycle 
and a certain condition ($\dagger$), defined below, holds. 
We start introducing some terminology. 

Let $G$ be a directed graph on nodes $v_1,\dots,v_{k+1}$ that contains, for 
each $1 \leq i < j \leq k+1$, either the edge $E(v_i,v_j)$ or $E(v_j,v_i)$. 
For a $B \subseteq \{v_1,\dots,v_{k+1}\}$ of size $1 \leq \l \leq k-1$ and a node 
$v \in \{v_1,\dots,v_{k+1}\} \setminus B$, we define 
${\sf conn}(v,B)$  as the tuple 
$(e_1,\dots,e_{k+1}) \in \{-1,1,\#\}^{k+1}$ such that for each $1 \leq p \leq k+1$: 
$$e_p \ = \ \begin{cases}
\#, & \text{if $v_p \not\in B$,} \\ 
1,  & \text{if $v_p \in B$ and the edge $E(v,v_p)$ is in $G$,} \\ 
-1, & \text{otherwise, i.e., $v_p \in B$ and $E(v_p,v)$ is in $G$.}
\end{cases} 
$$  
In simple terms, ${\sf conn}(v,B)$ specifies how $v$ connects with the nodes in $B$.  


Our condition ($\dagger$) then establishes the following: 
\begin{enumerate}
\item[($\dagger$)] 
For each $B \subseteq \{v_1,\dots,v_{k+1}\}$ of size $2 \leq \l \leq k-1$ and each node 
$v$ in $\{v_1,\dots,v_{k+1}\} \setminus B$, there is a node $v' \in \{v_1,\dots,v_{k+1}\} \setminus B$ such that 
$${\sf conn}(v,B) \quad \neq \quad {\sf conn}(v',B).$$
That is, for each such $B$ and $v$ we will always be able to find another $v'$ outside $B$ that connects 
to the nodes in $B$ in a different way than $v$. 
%
\end{enumerate}

\begin{figure}
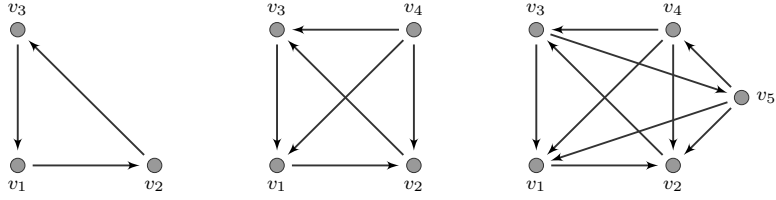

\centering 
\scalebox{0.9}{
\picCrossConditionTwo \hspace{10mm}
\picCrossConditionThree \hspace{10mm}
\picCrossConditionFour}
\caption{Directed graphs that satisfy condition ($\dagger$) for $k = 2,3,4$, respectively.}
\label{fig:fig3}
\end{figure}

\begin{example} \label{ex:dagger} The graphs in Figure \ref{fig:fig3} satisfy this condition for $k = 2,3,4$, respectively.
Notice that the directed cycle on nodes $\{v_1,v_2,v_3\}$, shown in the left-hand side, satisfies condition ($\dagger$) trivially. 
\qed
\end{example} 

The next lemma establishes that for each $k > 1$ there is always a graph that satisfies this condition.

\begin{lemma} \label{lemma:good-graphs}  
For each $k>1$, there is a directed graph $G$ on nodes $v_1,\dots,v_{k+1}$ such that the following hold:
\begin{enumerate}
\item For each $1 \leq i < j \leq k+1$, either the edge $E(v_i,v_j)$ or $E(v_j,v_i)$ is in $G$;   
\item the subgraph of $G$ induced by $\{v_1,v_2,v_3\}$ is a directed cycle; and 
\item $G$ 
satisfies condition {\em ($\dagger$)}. 
\end{enumerate} 
\end{lemma} 

\begin{proof}[Lemma \ref{lemma:good-graphs}]
For $k = 2$ this is given by the graph in Example \ref{ex:dagger}. 
For $k \geq 3$ we prove by induction a stronger claim: 
There is a directed graph $G$ on nodes $v_1,\dots,v_{k+1}$ such that:
\begin{enumerate}
\item $G$ contains either the edge $E(v_i,v_j)$ or $E(v_j,v_i)$ for 
each $1 \leq i < j \leq k+1$.
\item The subgraph of $G$ induced by $\{v_1,v_2,v_3\}$ is a directed cycle. 
\item $G$ contains the edges 
$E(v_1,v_2)$ and $E(v_4,v_3)$. 
\item $G$ satisfies condition ($\dagger$). 
\end{enumerate} 

The basis case $k = 3$ is given again by the graph in Example \ref{ex:dagger}. 
For the inductive case, assume by induction hypothesis that there is a directed graph 
$G$ on nodes $v_1,\dots,v_{k+1}$ that satisfies the claim above. 
A new graph $G'$ is then created from $G$ by adding a new node $v_{k+2}$ 
and connecting it to the nodes in $\{v_1,\dots,v_{k+1}\}$ as follows: 
For each $1 \leq i \leq k$, if $E(v_i,v_{i+1})$ is in $G$ 
then we add the edge $E(v_{k+2},v_i)$ to $G'$, otherwise we add the edge $E(v_{i},v_{k+2})$. 
Moreover, if $E(v_{k+1},v_{1})$ is in $G$ 
then we add the edge $E(v_{k+2},v_{k+1})$ to $G'$, otherwise we add the edge $E(v_{k+1},v_{k+2})$. 
Notice that $G$ coincides with the subgraph of $G'$ that is induced by nodes $v_1,\dots,v_{k+1}$. 
Moreover, by construction $G'$ satisfies the first three conditions of the claim. We prove next that it also 
satisfies condition ($\dagger$). 
 
Take an arbitrary $B \subseteq \{v_1,\dots,v_{k+2}\}$ of size $2 \leq \l \leq k$ and a node $v$ outside $B$. 
 We prove that the condition holds by cases: 

 \begin{itemize}
 
 \item $v_{k+2} \not\in B$, $v \in \{v_1,\dots,v_{k+1}\}$, 
 and $2 \leq \l \leq k-1$: By inductive hypothesis there is a node $v' \in \{v_1,\dots,v_{k+1}\} \setminus B$ such that
 ${\sf conn}(v, B) \neq {\sf conn}(v', B)$. 
 
 \item $v_{k+2} \not\in B$, $v \in \{v_1,\dots,v_{k+1}\}$, 
 and $\l = k$: We set $v' := v_{k+2}$ and claim that the predecessor 
 $u$ of $v$ in $\{v_1,\dots,v_{k+1}\}$ distinguishes $v$ and $v'$. Here, the ``predecessor'' of $v_{i}$ is  $v_{i-1}$ if 
 $2 \leq i \leq k+1$, and the predecessor of $v_{1}$ is $v_{k+1}$ (note that $u \in B$ as $\l = k$). 
 By construction of $G'$, we have that $E(u,v) \in G'$ if and only if $E(v',u) \in G'$. 
 We conclude that ${\sf conn}(v, B) \neq {\sf conn}(v', B)$. 
 
  \item $v_{k+2} \not\in B$ and $v = v_{k+2}$: There must exist some node $v'$ in $\{v_1,\dots,v_{k+1}\}$ that
  does not belong to $B$ but its predecessor $u$ in $\{v_1,\dots,v_{k+1}\}$ does. Then by construction of $G'$, 
  we have that 
 $E(u,v') \in G'$ if and only if $E(v,u) \in G'$. We conclude that 
 ${\sf conn}(v, B) \neq {\sf conn}(v', B)$.
 
 \item $v_{k+2} \in B$ and $\l \geq 3$: Then $B' = B \setminus \{v_{k+2}\}$ is of size $2 \leq \l-1 \leq k-1$. 
 By induction hypothesis, for every node $v$ outside 
 $B'$ there is another node $v' \in \{v_1,\dots,v_{k+1}\} \setminus B'$ such that ${\sf conn}(v, B') \neq {\sf conn}(v', B')$. This implies that 
 ${\sf conn}(v, B) \neq {\sf conn}(v', B)$.
 
 \item $v_{k+2} \in B$ and $\l = 2$: Then $B = \{v_{k+2},u\}$ for some $u \in \{v_1,\dots,v_{k+1}\}$. 
 Suppose first that $u \in \{v_1,v_2,v_3\}$. Since the subgraph induced by $\{v_1,v_2,v_3\}$ in $G$ defines a directed cycle, 
 it is the case that $E(u,z)$ holds if and only if $E(w,u)$ holds, where $\{u,w,z\} = \{v_1,v_2,v_3\}$. Therefore, for each $v \in \{v_1,\dots,v_{k+1}\} \setminus B$ 
 there is a node $v' \in \{z,w\}$ such that ${\sf conn}(v, \{u\}) \neq {\sf conn}(v', \{u\})$. It follows that 
 ${\sf conn}(v, B) \neq {\sf conn}(v', B)$. Suppose now that $u \not\in \{v_1,v_2,v_3\}$. It suffices to exhibit two nodes 
 $v'$ and $v''$ outside $B$ such that $E(v', v_{k+2})$ and $E(v_{k+2}, v'')$. By induction hypothesis the edges $E(v_1,v_2)$ and $E(v_4,v_3)$ are in $G'$. Therefore, 
 $v_{k+2}$ is connected via edges $E(v_3, v_{k+2})$ and $E(v_{k+2}, v_1)$ in $G'$.
%
 \end{itemize} 
 
This concludes the proof of the lemma. \qed
\end{proof} 
 
 Fix $k \geq 1$. 
We then take as $q$ any Boolean 
CQ whose canonical database is a graph $G$ on nodes 
$v_1,\dots,v_{2k+1}$ that satisfies the conditions stated in 
Lemma \ref{lemma:good-graphs}. That is, 
 (1) for 
each $1 \leq i < j \leq 2k+1$, either the edge $E(v_i,v_j)$ or $E(v_j,v_i)$ is in $G$, (2) 
the subgraph of $G$ induced by $\{v_1,v_2,v_3\}$ is a directed cycle, and (3) $G$ 
satisfies condition ($\dagger$).
It is easy to see that $q$ is in $\HW(k+1) \setminus \HW(k)$ as its underlying undirected graph is a clique 
on $2k+1$ elements. 
In fact, these elements can be covered with $(k+1)$ edges, but not with $k$. 

We claim that $q$ has no $\HW(\l)$-overapproximation for any $1 \leq \l \leq k$. 
The proofs for the cases when $\l = 1$ and $\l > 1$ are slightly different. We start with the latter, i.e., 
when $1 < \l \leq k$. The proof for every such an $\l$ is analogous, and thus we concentrate on proving 
the claim for $\l = k > 1$. According to Theorem \ref{theo:boundedness},
we need to prove that there is no constant $c \geq 0$ such that for every database $\D$ it holds that   
$$q \to_{k} \D \quad \Longleftrightarrow \quad q \to_{k}^c \D.$$ 
It is sufficient to show then that for each integer $c \geq 0$ there is a database $\D$ 
such that    
$$q \to_{k}^c \D \ \ \text{ but } \ \ q \not\to_{k}^{c+1} \D.$$
Or, equivalently, that 
for each integer $c \geq 0$ there is a database $\D$ 
such that    
$$q_c \to \D \ \ \text{ but } \ \ q_{c+1} \not\to \D,$$
where $q_c$, for $c \geq 0$, is the CQ which is defined in Lemma \ref{lemma:tedious}, i.e., 
for every $\D$ it is the case that $q \to_k^c \D$ iff $q_c \to \D$. 
In view of Equation \eqref{eq:cont}, this boils down to proving that 
\begin{equation} \label{eq:induction}  
q_{c+1} \not\to q_{c}, \ \ \ \text{for each $c \geq 0$.}
\end{equation}  

We prove Equation \eqref{eq:induction} by induction. 
The claim clearly holds for $c = 0$, as by definition 
$q_0$ is empty while $q_1$ is not.  Let us assume now that the claim holds for $c \geq 0$. 
That is, $q_{c+1} \not\to q_c$. This means, in particular, that the core of 
$q_{c+1}$ is not contained in $q_c$. That is, this core contains 
at least 
one node $w$ in $q_{c+1}$ that does not belong to $q_c$. 

%
%
By the way $q$ is defined, any  
$k$-union of $q$ must be of the form $S \subseteq \{v_1,\dots,v_{2k+1}\}$ with $|S| = 2k$. 
Let us consider now $(T_{c+1},\beta_{c+1})$ as defined in the proof of Lemma 
\ref{lemma:tedious}. Since $w \not\in q_c$, it must be the case that 
there is a unique node $t$ of $T_{c+1}$ such that 
$w \in \beta_{c+1}(t)$.   
Moreover, this $t$ must be a leaf of $T_{c+1}$. 
Suppose that 
$\phi_t(w) = v$, for $v \in \{v_1,\dots,v_{2k+1}\}$, where $\phi_t$ is as defined in the proof of Lemma 
\ref{lemma:tedious}, i.e., $\phi_t$ is a bijection between $\beta_{c+1}(t)$ and 
the $k$-union $S \subseteq \{v_1,\dots,v_{2k+1}\}$ of $q$ such that $\lambda_{c+1}(t) = S$. 
Notice, by definition, that if the parent of $t$ in $T_{c+1}$ is $t'$, then either 
$\lambda_{c+1}(t') = \emptyset$ -- which holds precisely when $t'$ is the root of $T_{c+1}$ --, or
$\lambda_{c+1}(t') = S'$, where $S'$ is the subset of $\{v_1,\dots,v_{2k+1}\}$ which contains all 
elements save for $v$. That is, in the latter case we have that $S'$ is obtained from $S$ by replacing some 
element $v'$ in $\{v_1,\dots,v_{2k+1}\}$, with $v' \neq v$, by $v$ itself.

From Proposition \ref{prop:cores}, 
we can assume that  
the homomorphism that maps $q_{c+1}$ to its core is a retraction, i.e., 
it is the identity on the nodes of this core, in particular, on $w$. 
On the other hand, $w$ is linked in $q_{c+1}$ exclusively with the remaining 
nodes that appear in $\beta_{c+1}(t)$.
 Moreover, the graph induced by the nodes in $\lambda_{c+1}(t)$ is   
a clique on $2k$ elements, and thus all the elements in $\beta_{c+1}(t)$  
must belong to the core of $q_{c+1}$. 

Recall that $\phi_t(w) = v$. 
%
Take an arbitrary node $v'' \in S$ that is not $v$. Notice 
that neither $v'' = v'$ as $v'' \in S$, while $v' \not\in S$.  
By definition, 
$T_{c+2}$ contains a leaf $t''$ whose parent is $t$ such that $\lambda_{c+2}(t'') = S''$, 
where $S''$ is the subset of $\{v_1,\dots,v_{2k+1}\}$ which is obtained from $S$
by replacing $v''$ with the unique node in 
$\{v_1,\dots,v_{2k+1}\} \setminus S$, namely $v'$.  
Let us assume that $\phi_{t''}(v') = w''$. Notice that $w''$ appears in no other node in 
$(T_{c+2},\beta_{c+2})$. 

Assume now, for the sake of contradiction, that $q_{c+2} \to q_{c+1}$. Then 
the core of $q_{c+2}$ is the same than the core of $q_{c+1}$. Let $C$ be this core. 
Henceforth, 
from Proposition \ref{prop:cores} there 
is a retraction $h$ from $q_{c+2}$ to $C$.
Since all elements in $\beta_{c+2}(t) = \beta_{c+1}(t)$ are in $C$, the homomorphism $h$ must be the identity on them.  
But then $h$ maps $w'$ to the unique element in $q_{c+1}$ that is linked to exactly 
the same nodes than $w'$ in 
$q_{c+2}$; namely, $\phi_t(v'') = w''$.   

Suppose that $v'$ and $v''$ represent the nodes $v_i$ and $v_j$ in $\{v_1,\dots,v_{2k+1}\}$, respectively. 
By assumption, $i \neq j$. 
But this implies then that in the canonical database $G$ of $q$ we have that  
$${\sf conn}(v_i, B) \ = \ {\sf conn}(v_j, B),$$
where $B = \{v_1,\dots,v_{2k+1}\} \setminus \{v_i,v_j\}$. This is 
a contradiction since $B$ is of size $2k - 1 > 1$ 
and $G$ satisfies condition ($\dagger$). 
This concludes our proof that $q$ has no $\HW(k)$-overapproximation (and, analogously, that 
it has no $\HW(\l)$-overapproximation for any $1 < \l \leq k$). 

We prove next that $q$ neither has a $\HW(1)$-overapproximation. 
Let us assume, for the sake of contradiction, that $q$ 
has a $\HW(1)$-overapproximation $q'$. It is an easy observation that the directed graphs in 
$\HW(1)$ are precisely those whose underlying undirected graph is acyclic.
Notice also that $q'$ has no directed cycles of length two (i.e., atoms of the form 
$E(u,v)$ and $E(v,u)$); otherwise, since $q'\to q$, 
we would have that $q$ also has such a cycle (which we know it does not). 
Using the fact that $q' \in \HW(1)$ and has no directed cycles of length two, it is not difficult to show (see e.g. \cite{HN}) that there is a sufficiently large integer $n\geq 1$ 
such that, if $\vec{P}_n$ is the directed path on $n$ vertices, then 
$$q' \to \vec{P}_n \ \ \text{ but } \ \ \vec{P}_n \not\to q'.$$ 

This implies that if $q''$ is the Boolean 
CQ which is naturally defined by $\vec{P}_n$, then $q'' \subsetneq q'$. Moreover, $\vec P_n \to G$. 
This is due to the fact that $G$ contains a directed cycle on $\{v_1,v_2,v_3\}$. We conclude that
$$q \, \subseteq \, q'' \subsetneq q',$$
and, therefore, that $q'$ is not a $\HW(1)$-overapproximation of $q$. This is a contradiction. 
We then conclude the proof of Theorem \ref{theo:non-existence}. \qed 
\end{proof} 

\medskip

\begin{proof}[Lemma \ref{lemma:aux}]
Before proving the lemma, we need some terminology and claims. 
Let $\D$ be a database and $(A_1,\dots,A_n)$ be a tuple of pairwise-disjoint subsets of elements of $\D$, where $n\geq 0$. 
In addition, let $\D'$ be a database and $(a_1,\dots,a_n)$ a tuple of elements in $\D'$. 
Then we write $(\D,(A_1,\dots,A_n))\to (\D',(a_1,\dots,a_n))$ iff there is a homomorphism $h$ from $\D$ to $\D'$  such that, 
for each $i\in \{1,\dots,n\}$ and $a\in A_i$,  it is the case that $h(a)=a_i$. 

For such a pair $(\D,(A_1,\dots,A_n))$, with $n\geq 0$, we define its generalized hypertreewidth in the natural way. 
The intuition is that we see  $(\D,(A_1,\dots,A_n))$ as a ``query'', where $A_1\cup\cdots \cup A_n$ are the ``free variables'' and the 
rest of the elements are the ``existential variables''. Formally, a \emph{tree decomposition} of $(\D,(A_1,\dots,A_n))$ is a pair $(T,\chi)$, where $T$ is a tree and $\chi$  is a mapping that assigns a subset of the elements in $\D\setminus (A_1\cup\cdots \cup A_n)$ to each node $t \in T$, such that 
the following statements hold:
%
\begin{enumerate}
  \item For each atom $R(\bar a)$ in $\D$, it is the case that $\bar a\cap (\D\setminus (A_1\cup\cdots \cup A_n))$ is contained in $\chi(t)$, for some $t \in T$.
  \item For each element $a$ in $\D\setminus (A_1\cup\cdots \cup A_n)$, the set of nodes $t \in T$ for which $a$ occurs in $\chi(t)$ is connected. 
\end{enumerate} 

The {\em width} of  node $t$ in $(T,\chi)$ is the minimal number $\ell$ for which there are $\ell$ atoms in $\D$ covering $\chi(t)$, i.e., 
atoms $R(\bar a_1),\dots,R(\bar a_\ell)$ in $\D$ 
such that $\chi(t)\subseteq\bigcup_{1\leq i \leq \ell} \bar a_i$
 The width of $(T,\chi)$ is the maximal width of the nodes of $T$.  
The \emph{generalized hypertreewidth} 
of $(\D,(A_1,\dots,A_n))$ is the minimum width of its tree decompositions.

By mimicking the proof of the forward implication of Proposition \ref{prop:games-tw}, 
we can show the following:

\begin{lemma}
\label{lemma:aux1}
Fix $k\geq 1$. Let $q(\bar x),q'(\bar x')$ be CQs, where $\bar x=(x_1,\dots,x_n)$ and $\bar x'=(x_1',\dots,x_n')$, for $n\geq 0$. 
Suppose that $(q,\bar x)\to_k (q',\bar x')$. Then, for each database $\D$ and tuple $(A_1,\dots,A_n)$ of subsets of $\D$ such that 
 $(\D,(A_1,\dots,A_n))$ has generalized hypertreewidth at most $k$, it is the case that  
 \begin{multline*}
 (\D,(A_1,\dots,A_n))\to (q,(x_1,\dots,x_n)) \quad \Longrightarrow \\ 
 (\D,(A_1,\dots,A_n))\to (q',(x_1',\dots,x_n')).\end{multline*} 
\end{lemma}

\begin{proof}
Let $\H$ be a winning strategy for Duplicator witnessing the fact that $(q,\bar x)\to_k (q',\bar x')$. 
Let us assume that $(\D,(A_1,\dots,A_n))$ has generalized hypertreewidth at most $k$, and that 
$(\D,(A_1,\dots,A_n))\to (q,(x_1,\dots,x_n))$ is witnessed via a homomorphism $h$. 
Then we can compose $h$ with the strategy $\H$ to define a homomorphism $g$ witnessing $(\D,(A_1,\dots,A_n))\to (q',(x_1',\dots,x_n'))$. 
The mapping $g$ is defined in a top-down fashion over the tree decomposition $(T,\chi)$ of width at most $k$ of $(\D,(A_1,\dots,A_n))$. 
One starts at the root $r$ of $T$, and forces Spoiler to play his pebbles over the set $h(\chi(r))$. 
If Duplicator responds according to $\H$ with a partial homomorphism $f_r$, 
we then let $g(a)=f_r(h(a))$, for each $a\in \chi(r)$. 
We then move to each child of $r$ and so on, 
until all leaves are reached and $g$ is defined over all elements in $\D\setminus (A_1\cup\cdots\cup A_n)$. 
Since Duplicator responds to Spoiler's moves with consistent partial homomorphisms, we have that 
$g$ is actually a well-defined homomorphism 
from $(\D,(A_1,\dots,A_n))$ to $(q',(x_1',\dots,x_n'))$.\qed
\end{proof}

Now we are ready to show our lemma. 
Suppose that $(q,\bar x)\to_k(q',\bar x')$, where $\bar x=(x_1,\dots,x_n)$ and $\bar x'=(x_1',\dots,x_n')$, for some $n\geq 0$. 
Assume that $(q'',\bar x'')\to (q'\wedge q, \bar z)$ via a homomorphism $h$, for $q''(\bar x'')\in \HW(k)$, and 
suppose that $\bar x''=(x_1'',\dots,x_n'')$ and $\bar z=(z_1,\dots,z_n)$. 
For each $i\in \{1,\dots,n\}$, we define $V_i$ to be the set of variables $x$ in $q''$ such that $h(x)=z_i$. 
In particular, $x_i''\in V_i$,  for each  $i\in \{1,\dots,n\}$. 
We define $V$ to be the set of variables $x$ in $q''$ such that $h(x)=y$, where $y$ is an existentially quantified variable of $q$. 
Similarly, we define $V'$ with respect to the existentially quantified variables of $q'$. 
Note that the sets $V,V',V_1,\dots,V_n$ form a partition of the variables of $q''$. 

Recall that $\D_{q''}$ be the canonical database of $q''$. Since $q''\in \HW(k)$, we know that 
$$\big(\D_{q''}, (\{x_1''\},\dots,\{x_n''\})\big)$$ has generalized hypertreewidth at most $k$, as defined above. 
Let $\D_V$ be the database induced in $\D_{q''}$ by the set of variables $V\cup V_1\cup \dots\cup V_n$, i.e., the set of atoms $R(\bar t)\in \D_{q''}$ 
such that each element in $\bar t$ is in $V\cup V_1\cup\dots\cup V_n$. 
We now show that $$\big(\D_V,(V_1,\dots,V_n)\big)$$ has also generalized hypertreewidth at most $k$. 
Indeed, let $(T,\chi)$ be the tree decomposition of $(\D_{q''}$ $,$ $(\{x_1''\},\dots,\{x_n''\}))$ of width at most $k$. 
Define $\chi'$ such that for each $t\in T$, we have that $\chi'(t)=\chi(t)\cap V$. 
We claim that $(T,\chi')$ is a tree decomposition of $(\D_{V}$ $,$ $(V_1,\dots,V_n))$ of width at most $k$. 

In fact, since $(T,\chi)$ is a tree decomposition, we have that, for each $a\in V$, it is the case that the set $\{t\in T\mid a\in \chi'(t)\}$ is connected; and 
for each atom $R(\bar a)\in \D_V$, there is a node $t\in T$ such that $\bar a\cap V\subseteq \chi'(t)$. 
To see that the width of $(T,\chi')$ is bounded by $k$, let $t$ be a node in $T$. Since 
the width of $(T,\chi)$ is at most $k$, there are $\ell$ atoms $R(\bar a_1),\dots,R(\bar a_\ell)$ in $\D_{q''}$, 
with $\ell\leq k$, such that $\chi(t)\subseteq\bigcup_{1\leq i \leq \ell} \bar a_i$. Let $R(\bar a_{i_1}),\dots,R(\bar a_{i_p})$, where 
$1\leq i_1<\cdots < i_p\leq \ell$ and $p\leq \ell$, be the atoms in $\{R(\bar a_1),\dots,R(\bar a_\ell)\}$ that contain an element in $\chi'(t)$. 
Since $\chi'(t)\subseteq \chi(t)$, it is the case that $\chi'(t)\subseteq\bigcup_{1\leq j \leq p} \bar a_{i_j}$.  
It suffices to show that each $R(\bar a_{i_j})$ is actually an atom in $\D_V$, for $1\leq j\leq p$. Towards a contradiction, 
suppose that this is not the case. Then, there is an atom in $\D_{q''}$ that contains simultaneously one variable in $\chi'(t)\subseteq V$ and one variable in $V'$. 
By the definitions of $V'$ and $V$, and the fact that $h$ is a homomorphism, 
it follows that there is an atom in $(q'\wedge q)(\bar z)$ that mentions simultaneously one existentially quantified variable from $q'$ and one from $q$; 
this contradicts the definition of $(q'\wedge q)(\bar z)$. 
We conclude that the generalized hypertreewidth of $(\D_V,(V_1,\dots,V_n))$ is at most $k$. 

Recall that $h$ is our initial homomorphism from $(q'',\bar x'')$ to $(q'\wedge q, \bar z)$.  
Let $h_V$ be the restriction of $h$ to the set $V\cup V_1\cup\cdots\cup V_n$. 
By construction, $$\big(\D_V,(V_1,\dots,V_n)\big) \, \to \, \big(q,(x_1,\dots,x_n)\big)$$ via homomorphism 
$h_V$. 
We can then apply Lemma \ref{lemma:aux1} and obtain that 
$$\big(\D_V,(V_1,\dots,V_n)\big) \, \to \, \big(q',(x_1',\dots,x_n')\big)$$ via a homomorphism $h'$. 
We define our required homomorphism $g$ from $(q'',\bar x'')$ to $(q',\bar x')$ as follows: 
if $a\in V\cup V_1\cup\cdots\cup V_n$, then $g(a)=h'(a)$; 
otherwise, if $a\in V'$, then $g(a)=h(a)$.  
To see that $g$ is a homomorphism, it suffices to consider an atom $R(\bar a)\in \D_{q''}$ such that $\bar a$ contains an element in $V'$ and one element not in $V'$, 
and show that $R(g(\bar a))\in \D_{q'}$. 
Let $A$ be the set of elements in $\bar a$ that are not in $V'$. As mentioned above, 
there are no atoms in $\D_{q''}$ mentioning elements in $V'$ and $V$ simultaneously, thus $A\subseteq V_1\cup\cdots\cup V_n$. 
In particular, $h(a)=h'(a)$, for each $a\in A$. It follows that $R(g(\bar a))=R(h(\bar a))$, from which we conclude that 
$R(g(\bar a))\in \D_{q'}$. \qed
\end{proof}

\medskip

\begin{proof}[Proposition \ref{prop:non-unique}]
Consider the Boolean CQ $q$ from Figure \ref{fig:fig1}, defined as 
$$q\, = \, \exists x\exists y\exists z \big(P_a(x,y)\wedge P_a(y,x) \wedge P_a(y,z) \wedge P_a(z,y) \wedge P_b(z,x) \wedge P_b(x,z)\big),$$
and the CQ $q'$ from the same figure defined by 
\begin{multline*} 
q'\ , = \, 
\exists x\exists y_1\exists y_2\exists z \big(P_a(x,y_1)\wedge P_a(y_1,x) \wedge P_a(y_2,z) \\ 
\wedge P_a(z,y_2) \wedge P_b(z,x) \wedge P_b(x,z)\big).
\end{multline*} 
For each $n\geq 1$, we define the CQ 
\begin{multline*} 
q_n \, = \, 
\exists x_1\cdots \exists x_{n+1} \big( 
P_a(x_1,x_2)\wedge \cdots \wedge P_a(x_n,x_{n+1})\wedge \\ 
P_b(x_1,x_1)\wedge P_b(x_{n+1},x_{n+1}\big).
\end{multline*} 

Observe that $q'\wedge q_n\in \HW(1)$, for each $n\geq 1$. 
We now show that, for each $n\geq 1$, $q'\wedge q_n$ is an incomparable $\HW(1)$-$\Delta$-approximation of $q$. 
As mentioned in Example \ref{ex:example2}, we have that $q\to_1 q'$. In particular $q\to_1 (q'\wedge q_n)$. 
Clearly, $q\not\to (q'\wedge q_n)$. Also, $q_n\not\to q$ since variables $x_1$ and $x_{n+1}$ of $q_n$ cannot be mapped to any variable in $q$ via a homomorphism. 
Therefore, $(q'\wedge q_n)\not\to q$. 
By Theorem \ref{theo:char-incomp}, it follows that $q'\wedge q_n$ is an incomparable $\HW(1)$-$\Delta$-approximation of $q$.

Now we show that the CQs $\{q'\wedge q_n\}_{n\geq 1}$ form a family of non-equivalent CQs. 
First note that $q_n\not\to q'$, for each $n\geq 1$. Also, observe that $q_i\to q_j$ iff $i=j$, for $i,j\geq 1$. 
It follows that for each $i,j\geq 1$, such that $i\neq j$, it is the case that $(q'\wedge q_i)\not\to (q'\wedge q_j)$ and $(q'\wedge q_j)\not\to (q'\wedge q_i)$. 
In particular, $\{q'\wedge q_n\}_{n\geq 1}$ is a family of non-equivalent CQs. \qed
\end{proof}

\medskip

\begin{proof}[Proposition \ref{prop:incomp-ident}]
As already mentioned, the {\sc coNP} upper bound follows directly from Theorem \ref{theo:char-incomp}. 
For the lower bound, we consider the {\sc Non-Hom}$(H)$ problem, for a fixed directed graph $H$, which asks, 
 given a directed graph $G$, whether $G\not\to H$. 
 Let us assume that, for each $k\geq 1$, there is a directed graph $H_k$ such that: 
 
 \begin{enumerate}
 \item $H_k \in \HW(k)$, or more formally, the Boolean CQ $q_{H_k}$ whose canonical database is $H_k$ belongs to $\HW(k)$. 
 \item {\sc Non-Hom}$(H_k)$ is {\sc coNP}-complete even when the input directed graph $G$ satisfies that $H_k\not\to G$. 
 \end{enumerate}

We later explain how to obtain these graphs $H_k$'s. 
Now we reduce from the restricted version of {\sc Non-Hom}$(H_k)$ given by item (2) above. 
Let $G$ be a directed graph such that $H_k\not\to G$. 
We first check in polynomial time whether $G\to_k H_k$. 
If $G\not\to_k H_k$, we output a fixed pair $q_0,q_0'$ such that $q_0'\in \HW(k)$ and $q_0'$ is an incomparable $\HW(k)$-$\Delta$-approximation of $q_0$. 
In case that $G\to_k H_k$, we output the pair $q_{G}, q_{H_k}$, where $q_G$ and $q_{H_k}$ are Boolean CQs whose canonical databases
 are precisely $G$ and $H_k$, respectively. Since $q_{H_k}\in \HW(k)$ by item (1) above, the reduction is well-defined.

Suppose first that $G\not\to H_k$. 
If $G\not\to_k H_k$, then we are done, since $q_0'$ is an incomparable $\HW(k)$-$\Delta$-approximation of $q_0$. 
Otherwise, if $G\to_k H_k$, since $G\not\to H_k$ and $H_k\not\to G$ (item (2) above), Theorem \ref{theo:char-incomp} implies that 
$q_{H_k}$ is an incomparable $\HW(k)$-$\Delta$-approximation of $q_G$. 
On the other hand, assume that $G\to H_k$. 
In particular, we have that $G\to_k H_k$, and then, in this case, the reduction outputs 
 the pair $q_{G}, q_{H_k}$. Since $G\to H_k$, 
 we conclude that $q_{H_k}$ is not an incomparable $\HW(k)$-$\Delta$-approximation of $q_G$. 
 
 It remains to define the directed graph $H_k$. 
 If $k\geq 2$, it suffices to consider the clique on $2k$ vertices, that is, the directed graph $\vec K_{2k}$ 
 whose vertex set is $\{1,\dots, 2k\}$ and whose edges are $\{(i,j)\mid i\neq j, \text{ for } i,j\in\{1,\dots,2k\}\}$. 
 We have that $\vec K_{2k}\in \HW(k)$, and thus item (1) above is satisfied. 
 Also, we can reduce from the non-$2k$-colorability problem by replacing each undirected edge $\{u,v\}$ of a given undirected graph $G$, 
 by a directed edge in an arbitrary direction, e.g., from $u$ to $v$. 
 Clearly, this is a reduction from non-$2k$-colorability  to {\sc Non-Hom}$(\vec K_{2k})$. 
 Also note that the output $f(G)$ of the reduction satisfies that $\vec K_{2k}\not\to f(G)$, as 
 $f(G)$ has no directed loops nor directed cycles of length $2$. Therefore, item (2) above is satisfied. 
 For $k=1$, it is known from \cite{HNZ96} that there is an {\em oriented tree} $T$ (i.e., a directed graph whose underlying undirected graph is a tree and has no directed cycles of length $1$ (loops) and $2$) 
 such that {\sc Non-Hom}$(T)$ is {\sc coNP}-complete. Since $T$ is an oriented tree then it belongs to $\HW(1)$, and then item (1) is satisfied. 
 Also, by inspecting the reduction in \cite{HNZ96}, we have that item (2) also holds. \qed
 \end{proof}

\end{document}

%% file: pictures.tex


\newcommand{\bgcolor}{black!5}
\newcommand{\substructurefillcolor}{blue!40}
\newcommand{\substructureufillcolor}{blue!60!red!40}
\newcommand{\substructuredrawcolor}{blue!80}
\newcommand{\substructureudrawcolor}{blue!60!red!80}

\newcommand{\structurefillcolor}{blue!20}
\newcommand{\structureufillcolor}{blue!60!red!20}
\newcommand{\structuredrawcolor}{blue!40}
\newcommand{\structureudrawcolor}{blue!60!red!40}

\newcommand{\mnodedrawcolor}{black!80}
\newcommand{\mnodefillcolor}{black!40}

\newcommand{\mnc}{\mnodefillcolor}
\newcommand{\mgrey}{\mnodefillcolor}
\newcommand{\mblue}{blue}
\newcommand{\mred}{red}
\newcommand{\myellow}{yellow}
\newcommand{\mviolet}{violet}

\newcommand{\mopacity}{1.0}
\newcommand{\opac}{0.2}

\usetikzlibrary{arrows, shapes, snakes}

\tikzstyle{mnode}=[
  circle,
  fill=\mnodefillcolor, 
  draw=\mnodedrawcolor,
  minimum size=6pt, 
  inner sep=0pt
]

\tikzstyle{mnodeinvisible}=[
  minimum size=6pt, 
  inner sep=0pt
]

\tikzstyle{invisible}=[
  minimum size=0pt, 
  inner sep=0pt
]

\tikzstyle{invisiblel}=[
  minimum size=10pt, 
  inner sep=0pt
]

\tikzstyle{invisibleEdge}=[
  transparent
]

\tikzstyle{nameNode}=[
  font=\scriptsize
]

\tikzstyle{namingNode}=[
  font=\normalsize
]

\tikzstyle{mEdge}=[
  -latex', 
  thick, 
  shorten >=3pt, 
  shorten <=3pt,
  draw=black!80,
]

\tikzstyle{dDashedEdge}=[
  -latex', 
  thick, 
  shorten >=3pt, 
  shorten <=3pt,
  draw=black!80,
  dashed
]

\tikzstyle{dEdge}=[
  -latex', 
  thick, 
  shorten >=3pt, 
  shorten <=3pt,
  draw=black!80,
]

\tikzstyle{dhEdge}=[
  -latex', 
  thick, 
  shorten >=3pt, 
  shorten <=3pt,
  draw=black!80,
]
\tikzstyle{uEdge}=[
  thick, 
  shorten >=3pt, 
  shorten <=3pt,
  draw=black!80,
]
\tikzstyle{uhEdge}=[
  thick, 
  shorten >=3pt, 
  shorten <=3pt,
  draw=black!80,
]

\tikzstyle{cEdge}=[
  ultra thick, 
  shorten >=3pt, 
  shorten <=3pt,
  draw=black!80,
]

\tikzstyle{dotsEdge}=[
  very thick, 
  loosely dotted, 
  shorten >=7pt, 
  shorten <=7pt
]

\tikzstyle{snakeEdge}=[
  ->, 
  decorate, 
  decoration={snake,amplitude=.4mm,segment length=2.5mm,post length=0.5mm},
]

\tikzstyle{snakeEdgea}=[
  ->, 
  decorate, 
  decoration={snake,amplitude=.4mm,segment length=3mm,post length=0.5mm}
]

\newcommand{\redEdge}{\tikz{
  \node (tmpa) at (-0.4,0)[mnode]{};
  \node (tmpb) at (0.4,0)[mnode]{};
  \draw [cEdge, draw=red] (tmpa) to (tmpb);
}}

\newcommand{\yellowEdge}{\tikz{
  \node (tmpa) at (-0.4,0)[mnode]{};
  \node (tmpb) at (0.4,0)[mnode]{};
  \draw [cEdge, draw=yellow] (tmpa) to (tmpb);
}}

\newcommand{\blueEdge}{\tikz{
  \node (tmpa) at (-0.4,0)[mnode]{};
  \node (tmpb) at (0.4,0)[mnode]{};
  \draw [cEdge, draw=blue] (tmpa) to (tmpb);
}}

\newcommand{\violetEdge}{\tikz{
  \node (tmpa) at (-0.4,0)[mnode]{};
  \node (tmpb) at (0.4,0)[mnode]{};
  \draw [cEdge, draw=violet] (tmpa) to (tmpb);
}}

\newcommand{\textnode}[1]{%
      \begin{minipage}{100pt}%
    #1
      \end{minipage}%
}

\newcommand{\classNode}[2]{
    \begin{scope}[shift={#1}]
      \node[text depth=0.7cm, anchor=north] (tmp) at (0,0) {
  \begin{minipage}{170pt}%
    #2
  \end{minipage}%
      };
    \end{scope}
}

\tikzstyle{class rectangle}=[
  draw=black,
  inner sep=0.2cm,
  rounded corners=5pt,
  thick
]

\tikzstyle{mline}=[
  draw=black,
  inner sep=0.2cm,
  rounded corners=5pt,
  thick
]

\tikzstyle{mainclass rectangle}=[
  draw=blue,
  inner sep=0.2cm,
  rounded corners=5pt,
  very thick
]

\newcommand{\redpoint}{\tikz{\node (tmp) at (-1.5,1.6)[mnode, fill=red]{};}}
\newcommand{\yellowpoint}{\tikz{\node (tmp) at (-1.5,1.6)[mnode, fill=yellow]{};}}
\newcommand{\bluepoint}{\tikz{\node (tmp) at (-1.5,1.6)[mnode, fill=blue]{};}}
\newcommand{\violetpoint}{\tikz{\node (tmp) at (-1.5,1.6)[mnode, fill=violet]{};}}

\tikzstyle{substructure}=[
  fill=\substructurefillcolor,
  draw=\substructuredrawcolor,
  inner sep=0.2cm,
  rounded corners=5pt
]


\pgfdeclarelayer{background}
\pgfdeclarelayer{substructure}
\pgfdeclarelayer{edges}
\pgfdeclarelayer{foreground}
\pgfsetlayers{background,substructure,edges,main,foreground}

\tikzstyle{background rectangle}=[
  fill=black!10,
  draw=black!20,
  line width = 5pt,
  inner sep=0.4cm,
  outer sep=0.4cm,
  rounded corners=5pt
]


\tikzstyle{relationR}=[
  draw=red!60,
  fill=red!10,
  fill opacity=0.4,
  very thick,
  inner sep=0.2cm,
  rounded corners=20pt
]

\tikzstyle{relationS}=[
  draw=blue!60,
   fill=blue!10,
   fill opacity=0.4,
  very thick,
  inner sep=0.2cm,
  rounded corners=20pt
]

\newcommand{\picExampleAcyclicBlowupA}{
  \begin{tikzpicture}[
        xscale=0.9,
        yscale=1.3,
      ]
      
      \node[mnode] (r0) at (1.5,0)[label={below:$x_0$}]{};

      \node[mnode] (r12) at (0.75,1)[label={below:$x^1_1$}]{};
      \node[mnode] (r11) at (2.25,1)[label={below:$x^2_1$}]{};

      \node[mnode] (r22) at (0.75,2)[label={below:$x^1_2$}]{};
      \node[mnode] (r21) at (2.25,2)[label={below:$x^2_2$}]{};

      \node[mnode] (r32) at (0.75,3)[label={below:$x^1_3$}]{};
      \node[mnode] (r31) at (2.25,3)[label={below:$x^2_3$}]{};

     \begin{pgfonlayer}{substructure}

      \def\f{1.4};

      \draw[relationR] ($(r0)+\f*(0,-0.5)$) -- ($(r11)+\f*(1.0,0.2)$) -- ($(r12)+\f*(-1.0,0.2)$) -- cycle;

      \draw[relationR] ($(r11)+\f*(0,-0.7)$) -- ($(r21)+\f*(1.0,0.3)$) -- ($(r22)+\f*(-1.0,0.3)$) -- cycle;
      \draw[relationR] ($(r12)+\f*(0,-0.7)$) -- ($(r21)+\f*(1.0,0.3)$) -- ($(r22)+\f*(-1.0,0.3)$) -- cycle;

      \draw[relationR] ($(r21)+\f*(0,-0.60)$) -- ($(r31)+\f*(1.0,0.3)$) -- ($(r32)+\f*(-1.0,0.3)$) -- cycle;
      \draw[relationR] ($(r22)+\f*(0,-0.60)$) -- ($(r31)+\f*(1.0,0.3)$) -- ($(r32)+\f*(-1.0,0.3)$) -- cycle;

      \end{pgfonlayer}
      
  \end{tikzpicture}
}

\newcommand{\picExampleAcyclicBlowupB}{
  \begin{tikzpicture}[
        xscale=1.0,
        yscale=1.3,
      ]
       
      \node[mnode] (r0) at (0,0)[label={below:$y_0$}]{};

      \begin{scope}[shift={(-2.5,0)}]
        \node[mnode] (r1a1) at (0,1)[label={below:$y^1_1$}]{};

        \node[mnode] (r2a11) at (-1.25,2)[label={below:$y^{11}_2$}]{};
        \node[mnode] (r2a12) at (.75,2)[label={below:$y^{12}_2$}]{};

        \node[mnode] (r3a111) at (-2.0,3)[label={below:$y^{111}_3$}]{};
        \node[mnode] (r3a112) at (-1.0,3)[label={below:$y^{112}_3$}]{};
        \node[mnode] (r3a121) at (0.5,3)[label={below:$y^{121}_3$}]{};
        \node[mnode] (r3a122) at (1.5,3)[label={below:$y^{122}_3$}]{};
      \end{scope}

      \begin{scope}[shift={(2.5,0)}]
        \node[mnode] (r1a2) at (0,1)[label={below:$y^2_1$}]{};

        \node[mnode] (r2a21) at (-1.25,2)[label={below:$y^{21}_2$}]{};
        \node[mnode] (r2a22) at (.75,2)[label={below:$y^{22}_2$}]{};

        \node[mnode] (r3a211) at (-2,3)[label={below:$y^{211}_3$}]{};
        \node[mnode] (r3a212) at (-1.0,3)[label={below:$y^{212}_3$}]{};
        \node[mnode] (r3a221) at (0.5,3)[label={below:$y^{221}_3$}]{};
        \node[mnode] (r3a222) at (1.5,3)[label={below:$y^{222}_3$}]{};
      \end{scope}

     \begin{pgfonlayer}{substructure}

      \def\f{1.0};
      \draw[relationR] ($(r0)+\f*(0,-1.0)$) -- ($(r1a2)+\f*(1.5,0.3)$) -- ($(r1a1)+\f*(-1.5,0.3)$) -- cycle;

      \draw[relationR] ($(r1a1)+\f*(0,-1.0)$) -- ($(r2a12)+\f*(1.0,0.3)$) -- ($(r2a11)+\f*(-1.0,0.3)$) -- cycle;
      \draw[relationR] ($(r1a2)+\f*(0,-1.0)$) -- ($(r2a22)+\f*(1.0,0.3)$) -- ($(r2a21)+\f*(-1.0,0.3)$) -- cycle;

      \draw[relationR] ($(r2a11)+\f*(0,-0.9)$) -- ($(r3a112)+\f*(1.0,0.3)$) -- ($(r3a111)+\f*(-1.0,0.3)$) -- cycle;
      \draw[relationR] ($(r2a12)+\f*(0,-0.9)$) -- ($(r3a122)+\f*(1.0,0.3)$) -- ($(r3a121)+\f*(-1.0,0.3)$) -- cycle;
      \draw[relationR] ($(r2a21)+\f*(0,-0.9)$) -- ($(r3a212)+\f*(1.0,0.3)$) -- ($(r3a211)+\f*(-1.0,0.3)$) -- cycle;
      \draw[relationR] ($(r2a22)+\f*(0,-0.9)$) -- ($(r3a222)+\f*(1.0,0.3)$) -- ($(r3a221)+\f*(-1.0,0.3)$) -- cycle;

      \end{pgfonlayer}

  \end{tikzpicture}
}

\newcommand{\picExampleAcyclicBlowupAold}{
  \begin{tikzpicture}[
        xscale=1.5,
        yscale=2.0,
      ]
      
      \node[mnode] (r0) at (1,0)[label={below:$r_0$}]{};

      \node[mnode] (s11) at (0,1)[label={below:$s^1_1$}]{};
      \node[mnode] (s12) at (1,1)[label={below:$s^2_1$}]{};
      \node[mnode] (s13) at (2,1)[label={below:$s^3_1$}]{};

      \node[mnode] (r13) at (0,2)[label={below:$r^3_1$}]{};
      \node[mnode] (r12) at (1,2)[label={below:$r^2_1$}]{};
      \node[mnode] (r11) at (2,2)[label={below:$r^1_1$}]{};

      \node[mnode] (s21) at (0,3)[label={below:$s^1_2$}]{};
      \node[mnode] (s22) at (1,3)[label={below:$s^2_2$}]{};
      \node[mnode] (s23) at (2,3)[label={below:$s^3_2$}]{};

      \node[mnode] (r23) at (0,4)[label={below:$r^3_2$}]{};
      \node[mnode] (r22) at (1,4)[label={below:$r^2_2$}]{};
      \node[mnode] (r21) at (2,4)[label={below:$r^1_2$}]{};

     \begin{pgfonlayer}{substructure}

      \def\f{0.9};

      \draw[relationR] ($(r0)+\f*(0,-0.5)$) -- ($(s13)+\f*(1.0,0.3)$) -- ($(s11)+\f*(-1.0,0.3)$) -- cycle;

      \draw[relationR] ($(r13)+\f*(-0.1,-0.7)$) -- ($(s21)+\f*(-1.0,0.3)$) -- ($(s23)+\f*(1.0,0.3)$) -- cycle;
      \draw[relationR] ($(r12)+\f*(0,-0.7)$) -- ($(s21)+\f*(-1.0,0.3)$) -- ($(s23)+\f*(1.0,0.3)$) -- cycle;
      \draw[relationR] ($(r11)+\f*(+0.1,-0.7)$) -- ($(s21)+\f*(-1.0,0.3)$) -- ($(s23)+\f*(1.0,0.3)$) -- cycle;

      \draw[relationS] ($(s11)+\f*(-0.5,-0.5)$) -- ($(r13)+\f*(-0.2,0.7)$) -- ($(s12)+\f*(0.7,-0.5)$) -- cycle;
      \draw[relationS] ($(s12)+\f*(-0.7,-0.5)$) -- ($(r11)+\f*(0.2,0.7)$) -- ($(s13)+\f*(0.5,-0.5)$) -- cycle;
      \draw[relationS] ($(s11)+\f*(-0.7,-0.5)$) -- ($(r12)+\f*(0.0,0.7)$) -- ($(s13)+\f*(0.7,-0.5)$) -- ($(s12)+\f*(0.0,0.3)$) -- cycle;

      \draw[relationS] ($(s21)+\f*(-0.5,-0.5)$) -- ($(r23)+\f*(-0.2,0.7)$) -- ($(s22)+\f*(0.7,-0.5)$) -- cycle;
      \draw[relationS] ($(s22)+\f*(-0.7,-0.5)$) -- ($(r21)+\f*(0.2,0.7)$) -- ($(s23)+\f*(0.5,-0.5)$) -- cycle;
      \draw[relationS] ($(s21)+\f*(-0.7,-0.5)$) -- ($(r22)+\f*(0.0,0.7)$) -- ($(s23)+\f*(0.7,-0.5)$) -- ($(s22)+\f*(0.0,0.3)$) -- cycle;
      \end{pgfonlayer}
      
  \end{tikzpicture}
}

\newcommand{\picExampleAcyclicBlowupBold}{
  \begin{tikzpicture}[
        xscale=1.5,
        yscale=2.0,
      ]
       
      \node[mnode] (r0) at (1,0)[label={below:$r_0$}]{};

      \node[mnode] (s1a1) at (0,1)[label={below:$s^1_1$}]{};
      \node[mnode] (s1a2) at (1,1)[label={below:$s^2_1$}]{};
      \node[mnode] (s1a3) at (2,1)[label={below:$s^3_1$}]{};

      \begin{scope}[shift={(-4,0)}]
        \node[mnode] (r1a3) at (2,2)[label={below:$r^3_1$}]{};

        \node[mnode] (s2a31) at (0,3)[label={below:$s^{31}_2$}]{};
        \node[mnode] (s2a32) at (1,3)[label={below:$s^{32}_2$}]{};
        \node[mnode] (s2a33) at (2,3)[label={below:$s^{33}_2$}]{};

        \node[mnode] (r2a33) at (0,4)[label={below:$r^{33}_2$}]{};
        \node[mnode] (r2a32) at (1,4)[label={below:$r^{32}_2$}]{};
        \node[mnode] (r2a31) at (2,4)[label={below:$r^{31}_2$}]{};
      \end{scope}

      \node[mnode] (r1a2) at (1,2)[label={below:$r^2_1$}]{};

      \node[mnode] (s2a21) at (0,3)[label={below:$s^{21}_2$}]{};
      \node[mnode] (s2a22) at (1,3)[label={below:$s^{22}_2$}]{};
      \node[mnode] (s2a23) at (2,3)[label={below:$s^{23}_2$}]{};

      \node[mnode] (r2a23) at (0,4)[label={below:$r^{23}_2$}]{};
      \node[mnode] (r2a22) at (1,4)[label={below:$r^{22}_2$}]{};
      \node[mnode] (r2a21) at (2,4)[label={below:$r^{21}_2$}]{};

      \begin{scope}[shift={(+4,0)}]
        \node[mnode] (r1a1) at (0,2)[label={below:$r^1_1$}]{};

        \node[mnode] (s2a11) at (0,3)[label={below:$s^{11}_2$}]{};
        \node[mnode] (s2a12) at (1,3)[label={below:$s^{12}_2$}]{};
        \node[mnode] (s2a13) at (2,3)[label={below:$s^{13}_2$}]{};

        \node[mnode] (r2a13) at (0,4)[label={below:$r^{13}_2$}]{};
        \node[mnode] (r2a12) at (1,4)[label={below:$r^{12}_2$}]{};
        \node[mnode] (r2a11) at (2,4)[label={below:$r^{11}_2$}]{};
      \end{scope}


     \begin{pgfonlayer}{substructure}

      \def\f{0.9};
      \draw[relationR] ($(r0)+\f*(0,-0.5)$) -- ($(s13)+\f*(1.0,0.3)$) -- ($(s11)+\f*(-1.0,0.3)$) -- cycle;

      \draw[relationS] ($(s1a3)+\f*(0.5,-0.5)$) -- ($(r1a1)+\f*(0.9,0.4)$) -- ($(s1a2)+\f*(-1.1,-0.0)$) -- cycle;
      \draw[relationS] ($(s1a1)+\f*(-0.5,-0.5)$) -- ($(r1a3)+\f*(-0.9,0.4)$) -- ($(s1a2)+\f*(1.1,-0.0)$) -- cycle;
       \draw[relationS] ($(s1a1)+\f*(-0.7,-0.5)$) -- ($(r1a2)+\f*(0.0,0.7)$) -- ($(s1a3)+\f*(0.7,-0.5)$) -- ($(s1a2)+\f*(0.0,0.3)$) -- cycle;
      
      \draw[relationR] ($(r1a3)+\f*(0.3,-0.7)$) -- ($(s2a31)+\f*(-1.0,0.3)$) -- ($(s2a33)+\f*(1.0,0.3)$) -- cycle;

      \draw[relationS] ($(s2a31)+\f*(-0.5,-0.5)$) -- ($(r2a33)+\f*(-0.2,0.7)$) -- ($(s2a32)+\f*(0.7,-0.5)$) -- cycle;
      \draw[relationS] ($(s2a32)+\f*(-0.7,-0.5)$) -- ($(r2a31)+\f*(0.2,0.7)$) -- ($(s2a33)+\f*(0.5,-0.5)$) -- cycle;
      \draw[relationS] ($(s2a31)+\f*(-0.7,-0.5)$) -- ($(r2a32)+\f*(0.0,0.7)$) -- ($(s2a33)+\f*(0.7,-0.5)$) -- ($(s2a32)+\f*(0.0,0.3)$) -- cycle;

       \draw[relationR] ($(r1a2)+\f*(0,-0.7)$) -- ($(s2a21)+\f*(-1.0,0.3)$) -- ($(s2a23)+\f*(1.0,0.3)$) -- cycle;

      \draw[relationS] ($(s2a21)+\f*(-0.5,-0.5)$) -- ($(r2a23)+\f*(-0.2,0.7)$) -- ($(s2a22)+\f*(0.7,-0.5)$) -- cycle;
      \draw[relationS] ($(s2a22)+\f*(-0.7,-0.5)$) -- ($(r2a21)+\f*(0.2,0.7)$) -- ($(s2a23)+\f*(0.5,-0.5)$) -- cycle;
      \draw[relationS] ($(s2a21)+\f*(-0.7,-0.5)$) -- ($(r2a22)+\f*(0.0,0.7)$) -- ($(s2a23)+\f*(0.7,-0.5)$) -- ($(s2a22)+\f*(0.0,0.3)$) -- cycle;

      \draw[relationR] ($(r1a1)+\f*(-0.3,-0.7)$) -- ($(s2a11)+\f*(-1.0,0.3)$) -- ($(s2a13)+\f*(1.0,0.3)$) -- cycle;

      \draw[relationS] ($(s2a11)+\f*(-0.5,-0.5)$) -- ($(r2a13)+\f*(-0.2,0.7)$) -- ($(s2a12)+\f*(0.7,-0.5)$) -- cycle;
      \draw[relationS] ($(s2a12)+\f*(-0.7,-0.5)$) -- ($(r2a11)+\f*(0.2,0.7)$) -- ($(s2a13)+\f*(0.5,-0.5)$) -- cycle;
      \draw[relationS] ($(s2a11)+\f*(-0.7,-0.5)$) -- ($(r2a12)+\f*(0.0,0.7)$) -- ($(s2a13)+\f*(0.7,-0.5)$) -- ($(s2a12)+\f*(0.0,0.3)$) -- cycle;

      \end{pgfonlayer}
      
  \end{tikzpicture}
}


\newcommand{\picCrossConditionTwo}{
  \begin{tikzpicture}[
        xscale=1.0,
        yscale=1.0,
      ]
      
      \node[mnode] (1) at (0,0)[label={below:$v_1$}]{};
      \node[mnode] (2) at (2,0)[label={below:$v_2$}]{};
      \node[mnode] (3) at (0,2)[label={above:$v_3$}]{};

      \draw [dEdge] (1) to (2);
      \draw [dEdge] (2) to (3);
      \draw [dEdge] (3) to (1);
  \end{tikzpicture}
}

\newcommand{\picCrossConditionThree}{
  \begin{tikzpicture}[
      xscale=1.0,
      yscale=1.0,
    ]
      \node[mnode] (1) at (0,0)[label={below:$v_1$}]{};
      \node[mnode] (2) at (2,0)[label={below:$v_2$}]{};
      \node[mnode] (3) at (0,2)[label={above:$v_3$}]{};
      \node[mnode] (4) at (2,2)[label={above:$v_4$}]{};

      \draw [dEdge] (1) to (2);
      \draw [dEdge] (2) to (3);
      \draw [dEdge] (3) to (1);
      
      \draw [dEdge] (4) to (1);
      \draw [dEdge] (4) to (2);
      \draw [dEdge] (4) to (3);
  \end{tikzpicture}
}

\newcommand{\picCrossConditionFour}{
  \begin{tikzpicture}[
      xscale=1.0,
      yscale=1.0,
    ]
      \node[mnode] (1) at (0,0)[label={below:$v_1$}]{};
      \node[mnode] (2) at (2,0)[label={below:$v_2$}]{};
      \node[mnode] (3) at (0,2)[label={above:$v_3$}]{};
      \node[mnode] (4) at (2,2)[label={above:$v_4$}]{};
      \node[mnode] (5) at (3,1)[label={right:$v_5$}]{};

      \draw [dEdge] (1) to (2);
      \draw [dEdge] (2) to (3);
      \draw [dEdge] (3) to (1);
      
      \draw [dEdge] (4) to (1);
      \draw [dEdge] (4) to (2);
      \draw [dEdge] (4) to (3);

      \draw [dEdge] (5) to (1);
      \draw [dEdge] (5) to (2);
      \draw [dEdge] (3) to (5);
      \draw [dEdge] (5) to (4);

   \end{tikzpicture}
}

%% file: commands.tex
\newcommand{\acomment}[2]{\ \\ \fbox{\parbox{\linewidth}{{\sc #1}: #2}}}
\newcommand{\fcomment}[2]{{\color{blue}(#1)}\footnote{#1: #2}}

\newcommand{\mrf}[1]{\fcomment{MR}{#1}}
\newcommand{\mr}[1]{\acomment{MR}{#1}}

\newcommand{\pbf}[1]{\fcomment{PB}{#1}}
\newcommand{\pb}[1]{\acomment{PB}{#1}}

\newcommand{\tzf}[1]{\fcomment{TZ}{#1}}
\newcommand{\tz}[1]{\acomment{TZ}{#1}}

\newcommand{\deactivateComments}{
  \renewcommand{\mrf}[1]{}
  \renewcommand{\mr}[1]{}

  \renewcommand{\pbf}[1]{}
  \renewcommand{\pb}[1]{}

  \renewcommand{\tzf}[1]{}
  \renewcommand{\tz}[1]{}

}

\renewcommand{\L}{{\cal L}} 
\newcommand{\np}{{\sc NP}}
\newcommand{\ptime}{{\sc PTIME}}
\newcommand{\logcfl}{{\sc LOGCFL}}
\newcommand{\expspace}{{\sc Expspace}}
\newcommand{\twoexpspace}{{\sc 2Expspace}}
\newcommand{\nl}{{\sc NLogspace}}
\newcommand{\nat}{\ensuremath{\mathbb{N}}}

\newcommand{\reminder}[1]{{\large \bf Comment: #1\\}}
\newcommand{\OMIT}[1]{}

\renewcommand{\l}{\ell}
\newcommand{\A}{{\cal A}} 
\newcommand{\B}{{\cal B}} 
\newcommand{\G}{{\cal G}} 
\newcommand{\C}{{\cal C}} 
\newcommand{\D}{{\cal D}} 
\renewcommand{\H}{{\cal H}} 
\newcommand{\TW}{{\sf TW}} 
\newcommand{\AC}{{\sf AC}} 
\newcommand{\HW}{{\sf GHW}}
 \newcommand{\HHW}{{\sf HW}} 
\newcommand{\Var}{{\sf Var}} 
\newcommand{\Dom}{{\sf Dom}} 
\newcommand{\I}{{\cal I}} 
\newcommand{\R}{{\cal R}} 
\newcommand{\T}{{\cal T}} 
\newcommand{\ints}{\sf{int}} 
\newcommand{\exts}{\sf{ext}} 

\newcommand{\Enc}{{\sf Enc}} 
\newcommand{\Dec}{{\sf Dec}} 

\newcommand{\fth}{\hfill $\Box$}

\newcommand{\df}{\ensuremath{\mathrel{\smash{\stackrel{\scriptscriptstyle{
    \text{def}}}{=}}}}}



\newcommand{\apptheorem}[3]{\noindent\textbf{#1 #2.} \textit{#3}\vspace{2mm}}
\newcommand{\alemma}[2]{\apptheorem{Lemma}{#1}{#2}}
\newcommand{\aproposition}[2]{\apptheorem{Proposition}{#1}{#2}}
\newcommand{\atheorem}[2]{\apptheorem{Theorem}{#1}{#2}}
\newcommand{\acorollary}[2]{\apptheorem{Corollary}{#1}{#2}}